\documentclass[sigconf]{acmart}

\AtBeginDocument{%
  }
    
\copyrightyear{2025}
\acmYear{2025}
\setcopyright{rightsretained}
\acmConference[ASIA CCS '25]{ACM Asia Conference on Computer and Communications Security}{August 25--29, 2025}{Hanoi, Vietnam}
\acmBooktitle{ACM Asia Conference on Computer and Communications Security (ASIA CCS '25), August 25--29, 2025, Hanoi, Vietnam}\acmDOI{10.1145/3708821.3710826}
\acmISBN{979-8-4007-1410-8/25/08}

\usepackage{hyperref}
\usepackage{cleveref}
\usepackage{textcomp}
\usepackage{xcolor}
\usepackage{balance}

\usepackage{enumitem}
\usepackage{algpseudocode}
\usepackage{color}
\usepackage{xcolor}
\usepackage{pifont}
\usepackage{soul}
\usepackage{url}
\usepackage{algorithm}
\usepackage{algorithmicx}
\usepackage{xurl}
\usepackage{hhline}
\usepackage{multirow}
\usepackage{makecell}
\usepackage{caption}
\usepackage[caption=false,font=footnotesize,labelfont=sf,textfont=sf]{subfig}
\usepackage{tikz}
\usepackage[flushleft]{threeparttable}
\usepackage{framed}
\usepackage{multicol}
\usepackage{rotating} 
\usepackage{lipsum}
\usepackage{array}
\usepackage{tikz}
\usepackage{colortbl}
\usepackage{wasysym}
\usepackage[all,pdf]{xy}
\DeclareGraphicsExtensions{.eps,.ps,.jpg,.bmp,.pdf}
\graphicspath{{figure}}

\newcommand{\tikzcircle}[2][red,fill=red]{\tikz[baseline=-0.5ex]\draw[#1,radius=#2] (0,0) circle ;}%
\newcommand*\halfcirc[1][1ex]{%
  \begin{tikzpicture}
  \draw[fill] (0,0)-- (90:#1) arc (90:270:#1) -- cycle ;
  \draw (0,0) circle (#1);
  \end{tikzpicture}}

\newcommand*\xor{\oplus}
\newcommand\mydots{\ifmmode\ldots\else\makebox[1em][c]{.\hfil.\hfil.}\fi}

\newtheorem{theorem}{Theorem}

\hypersetup{colorlinks,linkcolor={red},citecolor={red},urlcolor={red}}

\begin{document}

\author{Viet Vo}
\orcid{0000-0002-5984-7981}
\affiliation{%
  \institution{Swinburne University of Technology}
  \city{Melbourne}
  \state{}
  \country{Australia}
}
\email{vvo@swin.edu.au}

\author{Shangqi Lai}
\orcid{0000-0002-0374-3593}
\affiliation{%
  \institution{CSIRO's Data61}
  \city{Melbourne}
  \country{Australia}}
\email{shangqi.lai@data61.csiro.au}

\author{Xingliang Yuan}
\orcid{0000-0002-3701-4946}
\affiliation{%
  \institution{The University of Melbourne}
  \city{Melbourne}
  \country{Australia}
}
\email{xingliang.yuan@unimelb.edu.au}

\author{Surya Nepal}
\affiliation{%
 \institution{CSIRO's Data61}
 \city{Sydney}
 \country{Australia}}
 \email{surya.nepal@data61.csiro.au}

\author{Qi Li}
\affiliation{%
  \institution{Tsinghua University}
  \city{Beijing}
  \country{China}}
\email{qli01@tsinghua.edu.cn}


\renewcommand{\shortauthors}{Vo, et al.}

\title{OblivCDN: A Practical Privacy-preserving CDN with Oblivious Content Access}

\begin{abstract}

%
%
 %

Content providers increasingly utilise Content Delivery Networks (CDNs) to enhance users' content download experience.
However, this deployment scenario raises significant security concerns regarding content confidentiality and user privacy due to the involvement of third-party providers.
%
%
Prior proposals using private information retrieval (PIR) and oblivious RAM (ORAM) have proven impractical due to high computation and communication costs, as well as integration challenges within distributed CDN architectures.

In response, we present \textsf{OblivCDN}, a practical privacy-preserving system meticulously designed for seamless integration with the existing real-world Internet-CDN infrastructure.
Our design strategically adapts Range ORAM primitives to optimise memory and disk seeks when accessing contiguous blocks of CDN content, both at the origin and edge servers, while preserving both content confidentiality and user access pattern hiding features.
Also, we carefully customise several oblivious building blocks that integrate the distributed trust model into the ORAM client, thereby eliminating the computational bottleneck in the origin server and reducing communication costs between the origin server and edge servers.
Moreover, the newly-designed ORAM client also eliminates the need for trusted hardware on edge servers, and thus significantly ameliorates the compatibility towards networks with massive legacy devices. 
In real-world streaming evaluations, \textsf{OblivCDN} demonstrates remarkable performance, downloading a $256$ MB video in just $5.6$ seconds. This achievement represents a speedup of $90\times$ compared to a strawman approach (direct ORAM adoption) and a $366\times$ improvement over the prior art, OblivP2P.

\end{abstract}


  
\maketitle


\section{Introduction}

Content Delivery Networks (CDNs), such as  Netflix~\cite{Netflix212}, Akamai~\cite{Akamai18}, and Cloudflare~\cite{Cloudflare212},  offer a popular way for file sharing and video streaming.   
The typical usage of CDNs offloads global traffic from an origin server to edge servers, located close to end-users, and improves the end-user's fetch latency. 
Such wide adaptation of CDNs is predicted to generate a massive market growth of US$\$3.2$B by 2026~\cite{ResearchMartket21} and serves a large portion ($73\%$ ) of global traffic~\cite{Intel21}.

Despite these benefits, CDNs remain significant security concerns regarding content confidentiality and user privacy. 
%
%
Unauthorised data access to CDN contents has been reported frequently~\cite{AkamaiIntel20,Adverline19,CloudHeartBleed,Pfizer07}. 
Thus, data needs to be always encrypted to protect content confidentiality in the origin and edge servers. 
In addition, end-users' privacy can be compromised by exploiting the content access patterns in CDNs~\cite{Wang21,Panchenko2016,Cai12}, even though the content remains encrypted. 
%
For example, when end-users repeatedly access the same content from the same server, the content popularity and users' personal interests can be inferred from these access patterns via fingerprinting~\cite{Panchenko2016,Cai12}. 
%
%
The access patterns are also widely exploited in encrypted storage by leakage attacks~\cite{IslamKK12,Cash15} to infer users' content requests. 
%
Therefore, hiding the access patterns of contents is crucial to protect end-users' privacy. 

%
While many privacy-preserving content distribution systems have been proposed~\cite{Jia16,Gupta16,Jia16OblivP2P,Jia2017,Silva19,Shujie20}, they often lack compatibility with real-world Internet CDN architecture~\cite{Netflix212,Akamai18,Cloudflare212} and struggle to achieve both security and efficiency. 
For example, APAC~\cite{Jia16} considers the peer-assisted CDN setting, which allows a peer node to download the content from edge servers and share content with nearby nodes.
However, this work aims to fortify the users' privacy between peers, but still maintain a deterministic access pattern on edge servers.
On the other hand, some systems~\cite{Silva19,Jia2017} require a trusted execution environment (TEE) on all edge servers.
Unfortunately, TEEs may not always be available due to infrastructure differences among regions.
It also hardly updates the infrastructure in an effective and financial way, since existing CDN providers, such as Netflix~\cite{Netflix212}, Akamai~\cite{Akamai18}, and Cloudflare~\cite{Cloudflare212}, have already deployed thousands of edge servers globally~\cite{NetflixOCA}.  

Although some existing proposals~\cite{Gupta16,Jia16OblivP2P,Shujie20} attempt to be compatible with real-world CDN systems, we stress that those works rely on heavy building blocks (e.g., Private Information Retrieval) and introduce a noticeable latency for end-users when downloading content from edge servers.
For instance, downloading a $512$KB block takes $4.1$s with OblivP2P~\cite{Jia16OblivP2P} and $2$s with Multi-CDN~\cite{Shujie20} from $2^5$ servers, which means a $256$MB video could take $35$ and $17$ minutes to be fully downloaded, respectively.
Moreover, those designs require CDN providers to create a significant number of replicas (at least twofold to support multi-server PIR~\cite{Gupta16,Jia16OblivP2P} and more to accelerate in~\cite{Shujie20}) for the entire world of CDN content~\cite{Dauterman20,Dauterman22,Frank17,Eskandarian21}.
This adaptation also increases the bandwidth cost for end-users to download the replicated content from multiple edge servers.

%

In this study, we propose \textsf{OblivCDN}, a new practical privacy-preserving CDN system. 
\textsf{OblivCDN} not only addresses the above security concerns but also seamlessly integrates with the real-world CDN architecture.
Furthermore, it achieves low-latency oblivious content access without incurring additional bandwidth costs for end-users and storage expenses at the edge servers.

\noindent{\bf Technique Overview.} 
Our observation is that the real-world CDN architecture can be easily adapted to satisfy our security goals with Oblivious RAM (ORAM)~\cite{Stefanov13}, i.e., placing the ORAM client at the origin server/end-user sides and distributing the ORAM server storage among edge servers.
However, such a naive solution causes significant latency and bandwidth costs when end-users fetch content.
Besides, it also incurs a computational bottleneck at the origin server and prohibitive communication overhead between the edge servers and the origin server.

To reduce end-users' fetch cost, \textsf{OblivCDN} considers the locality of content blocks and decomposes the end-users' requests into multiple range queries for contiguous blocks.
We adapt range ORAM~\cite{Chakraborti191} to optimise the memory and disk storage on the origin and edge servers, which enables efficient range accesses with no access pattern leakage.
Furthermore, to eliminate the computation and communication overhead, \textsf{OblivCDN} separates the metadata and physical video blocks between the origin and edge servers. 
Specifically, \textsf{OblivCDN} follows existing CDN systems~\cite{NetflixAWS,NetflixOCA} to employ a TEE-enabled origin server to manage only the metadata of video blocks (key and addressing information), while the larger content blocks are stored in separated ORAM storage on the edge servers. 
The origin server only pushes content blocks to the edges during uploads, and they are never routed back to the origin. 
As a result, the design of \textsf{OblivCDN} reduces the computation and communication overhead associated with using ORAM in CDNs.
%

%

%
%
%

%

On the other hand, we realise that distinguishing end-users' access patterns on the untrusted edge servers is possible, either through network traffic analysis~\cite{Wang21,Panchenko2016,Cai12} or attacks on encrypted databases~\cite{IslamKK12,Cash15}.
Obfuscating the access patterns becomes challenging without the presence of a trusted ORAM client deployed in edge servers or duplicating the content and bandwidth costs~\cite{Zahur16,Doerner17,WangORAM15,Bunn20}.
%
%
\textsf{OblivCDN} introduces a new design to address this challenge.
Specifically, \textsf{OblivCDN} distributes the ORAM trust assumption to two non-colluding computing service nodes located close to the edge servers.
These two nodes jointly function as an ORAM client and receive instructions from the origin server to manipulate edge storage with customised oblivious building blocks. 
The significance of this design lies in its capability to support oblivious range access and protect the access patterns of user content requests without increasing end-user bandwidth or replicating the edge storage.

We implement the prototype for \textsf{OblivCDN} and conduct an extensive evaluation. 
We first benchmark the performance of oblivious building blocks in \textsf{OblivCDN}. 
The experiment shows that with these building blocks and video metadata/data separation design, \textsf{OblivCDN} successfully reduces the origin server's communication bandwidth to a constant, which is smaller than $8$ MB, regardless of the video block size.
Next, we compare the real-world streaming performance of \textsf{OblivCDN} with a Strawman approach, which is a direct ORAM adaption. 
The results demonstrate that \textsf{OblivCDN} improves $87-90\times$ end-users' fetching latency for content between $2-256$ MB. 
Regarding communication, \textsf{OblivCDN} saves $6.1- 786\times$.
Compared to OblivP2P in a setting of $2^5$ peers~\cite{Jia16OblivP2P}, \textsf{OblivCDN} improves end-users' fetch latency by $366\times$ and $40\times$ throughput when requesting a $256$MB file.
%
%
\noindent In summary, our contributions are that:

\begin{itemize}[leftmargin=*]

\item We propose \textsf{OblivCDN}, the first CDN system designed to protect both content confidentiality and user privacy (access patterns to content) while maintaining compatibility with existing mainstream CDNs~\cite{Netflix212,Akamai18,Cloudflare212}. 
%
%

\item We formally analyse the security and performance of all entities involved in \textsf{OblivCDN}, including the origin server, edge servers, and non-colluding computing service nodes. 

\item We extensively evaluate \textsf{OblivCDN} with our Strawman design and OblivP2P~\cite{Jia16OblivP2P} under the real-world streaming setting. 
%
%

\end{itemize}

\section{Primitive building blocks}
\label{sec:preliminaries}

\noindent\textbf{Path ORAM and Data Locality}.
We leverage range ORAM~\cite{Chakraborti191}, a variant of Path ORAM~\cite{Stefanov13}, to enable efficient oblivious range access on CDN content.
Without loss of generality, the Path ORAM contains a storage server and a client.
The server maintains a full binary tree with $N$ nodes.
Each node contains $Z$ blocks with size $B$. 
If a node has less than $Z$ valid blocks, the node will be padded with dummy blocks to ensure that each node has a fixed size ($Z\times B$). 
The client maintains two data structures: a stash $S$, which keeps all blocks that are yet to be written to the server (can be blocks fetched from the server or new blocks); a position map $\mathsf{position}$, which keeps the mapping between blocks and the leaf node IDs in the binary tree.
The Path ORAM protocol consists of three algorithms:

\begin{itemize}[leftmargin=*]
	\item $({T},\mathsf{position}, S)\leftarrow\mathsf{ORAM.Init}(N, Z, B)$: Given the tree node number $N$, node size $Z$ and block size $B$, the server allocates $N\times Z\times B$ space and arranges it as a full binary tree $T$ with $\lceil\log_2{N}\rceil$ levels. 
	The client generates a position map $\mathsf{position}$ with random path information and a stash $S$.
	\item $\mathsf{ORAM.Access}(\textsf{Op}, bid, T, \mathsf{position}, S, b)$: With the binary tree $T$, the position map $\mathsf{position}$ and a block ID $bid$, the client gets the leaf node ID $\mathsf{lf}\leftarrow\mathsf{position}[bid]$ from $\mathsf{position}$.
	Then, it fetches all the blocks from the root of $T$ to $\mathsf{lf}$ and decrypts them, and inserts all these blocks into the stash $S$.
	If $\textsf{Op}=`\textsf{read}'$, the client loads $S[bid]$ into $b$.
	If $\textsf{Op}=`\textsf{write}'$, the client puts $b$ into $S[bid]$.
	Meanwhile, it assigns a new random leaf node ID to $bid$ in $\mathsf{position}$ to ensure that the next access to this block goes to a random path.
	\item $\mathsf{ORAM.Evict}(T, \mathsf{position}, S, \mathsf{lf})$: To write the blocks in $S$ back to a path with the given leaf node $\mathsf{lf}$, the client scans all blocks in $S$ and tries to fit the blocks into the nodes between the leaf node $\mathsf{lf}$ to the root node.
	If a node at the current level has enough space, the client encrypts the block and evicts the encrypted block from the stash to that node.
	Otherwise, the block remains in the stash.
	Finally, the nodes are written back to the store on the server.
\end{itemize}

In Path ORAM, each ORAM.Access always incurs path read operations in a tree path specified by the client. The accessed block will then be assigned with a new random tree path, so even the access on the same block will be randomised on the view of the server-side adversary.
Moreover, $\mathsf{ORAM.Evict}$ is always invoked after each $\mathsf{ORAM.Access}$ to write the accessed path, so the server-side adversary cannot distinguish read and write operations.

Under the CDN context, we specifically want to ensure efficient sequential accesses over data blocks as it is crucial to the streaming performance.
Hence, we employ a Path ORAM variant, named range ORAM (rORAM)~\cite{Chakraborti191}.
In rORAM, the client can have a predefined parameter $R\leq N$ to indicate the maximum sequence size that can be fetched in one access.
The server storage will be organised to have $\log_2R+1$ separate stores, denoted as $E_r$, where $r\in[0, \log_2{R}]$.
The rORAM client also keeps $\log_2R+1$ copies of data structures (i.e., $\mathsf{position}_r$ and $S_r$) that are corresponding to each $E_r$.

In rORAM, each $E_r$ stores $2^r$ blocks in a way that can serve $2^r$ sequential range queries with minimised disk seeks while preserving the randomness of blocks not in the targeted range.
Moreover, the blocks in each layer of all ORAM trees are stored adjacent to each other, which enables batch eviction that only requires fixed seeks independent of the number of blocks to be evicted.

With rORAM, CDN providers can accelerate sequential access on video blocks and thus improve the service quality in streaming. 
Particularly, let a video block sequence be $[bid,bid+2^r)$, and the stash of $E_r$ be $S_r$.
The scheme involves the following operations:

\begin{itemize}[leftmargin=*]

\item $E_r.\textsf{ReadRange}(bid)$ allows oblivious access $\mathsf{ORAM.Access}$ operation for blocks in the range $[bid,bid+2^r)$, where $2^r \leq R$. The requested $2^r$ blocks are loaded into the $S_r$.

\item $E_r.\textsf{BatchEvict}(k)$ loads all blocks from the targeted $k$ paths in $E_r$ to $S_r$. Then, it writes $k$ paths with blocks from $S_r$ to $E_r$. 

\end{itemize}

With rORAM, it only requires an $\mathcal{O}(\log_2 N)$ disk seeks to retrieve $2^r$ consecutive blocks and $\mathcal{O}(\log_2 N)$ seeks to evict $k$ paths to $E_r$, regardless of $r$.
As a leakage tradeoff, the optimal access reveals the range size supported by $E_r$.
More details are given in Appendix~\ref{subsec:rORAM_protocols}.

\noindent\textbf{Intel SGX}.
Intel SGX is an instruction set to enable hardware-assisted confidential computation~\cite{Costan161}.
It has been widely adopted by popular cloud platforms (e.g., Azure, AWS)~\cite{Azure21}.
Despite concerns about memory side-channel leakage~\cite{Lee2017InferringFC,Brasser17,Bulck2017TellingYS,Wang2017LeakyCO}, numerous techniques readily obfuscate the \textit{Enclave}'s access pattern to CPU caches in Intel SGX~\cite{Rane15,Sasy17,Ohrimenko16}. 
Additionally, cloud providers often respond promptly to TEE attacks~\cite{microsoft18}.

%

\noindent\textbf{Distributed Point Functions}.
Distributed Point Function (DPF) shares a point function $f(x)$, which evaluates to $1$ on the input $x^*$ and to $0$ on other inputs, to function shares $f_{i \in [1,m]}$ s.t.  $\Sigma^{m}_{i=1}f_i(x)=f(x)$ for any input $x$, while none of the $f_i$ \textit{information-theoretically} reveals $f$~\cite{Gilboa14,Boyle15,Boyle16}.
Under the two-party setting, the DPF consists of two probabilistic polynomial-time algorithms:
\begin{itemize}[leftmargin=*]
\item  $\textsf{Gen}(1^\lambda,x^*)$ is a key generation algorithm that takes the security parameter $1^\lambda$ and an input $x$ as inputs and outputs a tuple of keys $(k_1,k_2)$ for the point function at $x^*$. 

\item $\textsf{Eval}(k_p,x^\prime)$ is an evaluation algorithm executed at each party $p\in [1,2]$, which takes the inputs of $k_p$ and an evaluation point $x^\prime$, and outputs $y_p$, where $y_p$ is an additive share of $1$ if $x^\prime= x^*$, or a share of $0$ otherwise.
\end{itemize}

Considering the point function as an $n$-vector with $f(x), x\in[0, n)$, with DPF, the communication cost of sending such a vector is reduced to $\mathcal{O}(\log_2 n)$.
Moreover, each server only requires $\mathcal{O}(n\cdot \log_2 n)$ block cipher operations to obtain the shared vector~\cite{Boyle16}.

\section{System Overview}\label{sec:overview}

%
\begin{figure}[!t]
\centering
\includegraphics[width=\linewidth]{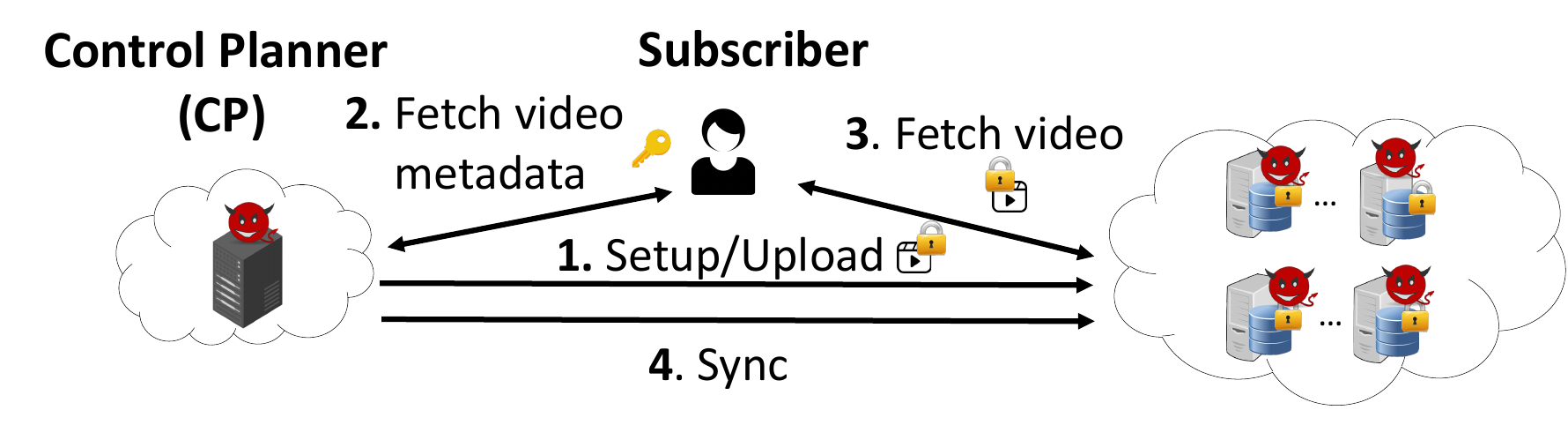}
\vspace{-15pt}
\caption{General CDN Architecture}
\vspace{-10pt}
\label{fig:arch}
\end{figure}

\subsection{System Entities and Operations}
\label{subsec:sysmodel}
In this paper, we follow the real-world CDN structure~\cite{NetflixAWS,NetflixOCA} as shown in Figure~\ref{fig:arch}.
In specific, a \textit{CDN service provider} (e.g., Netflix) hires a control planner \textit{CP} and a number of edge servers \textit{ES}s.
The \textit{CP} is hosted by a \textit{cloud service provider}~\cite{NetflixAWS}.
The \textit{ES}s are deployed within the network of an \textit{ISP partner}~\cite{NetflixOCA}. 
When a \textit{Subscriber} sends a video request, the \textit{CP} responds by providing the metadata for the requested blocks within the video.
Subsequently, the \textit{ES}s transmit sequences of blocks to the \textit{Subscriber}. 

Formally, a CDN system contains four operations as follows:

\noindent \ul{\textbf{Setup}:}
The \textit{CDN service provider} instructs the \textit{CP} to establish connections to the \textit{ES}s and initialise all required data structures.

\noindent \ul{\textbf{Upload}:} The service provider uploads the \textsf{video} to the \textit{CP}. 
Then, the \textit{CP} encrypts the \textsf{video} and distributes encrypted blocks to \textit{ES}s. 

\noindent \ul{\textbf{Fetch}:} Given a requested video ID, the \textit{CP} responds to the \textit{Subscriber} with the metadata of the video (i.e., decryption keys and the addresses of requested video blocks). 
Then, the \textit{Subscriber} accesses \textit{ES}s to download the video blocks and finally decrypts them.

\noindent \ul{\textbf{Sync}:} After users' access or video metadata updates from the \textit{CDN service provider}, the \textit{CP} instructs \textit{ES}s to update video blocks.
%
%
%
%

\subsection{Threat Model}
\label{subsec:threatmodel}

\noindent \textbf{Privacy threats}. Four parties are involved in the CDN system, i.e., the \textit{CDN service provider}, the \textit{cloud service providers}, the \textit{ISP partners}, and the \textit{Subscriber}.
In this paper, we consider the \textit{cloud service providers} and \textit{ISP partners} as untrusted due to frequent data breaches~\cite{Thalesgroup22,AkamaiIntel20,Adverline19,Pfizer07, Upguard19,Bitdefender20,Zdnet19}. 
Specifically, the \textit{cloud service providers} and \textit{ISP partners} are semi-honest, with the condition that they adhere to the Service Level Agreements with the \textit{CDN service provider}~\cite{NetflixOLA}. 
However, they do not collude with each other because practical CDN deployments often adopt a multi-cloud approach to explore the commercial benefits of various providers~\cite{AkamaiAlternative}\footnote{The collusion is still possible between \textit{ES}s within the same ISP network.}.
%

\noindent \textbf{Adversarial Capabilities}. Under our assumption, the adversary in the \textit{cloud service providers} and \textit{ISP partners} side can control the \textit{CP} and \textit{ES}s, respectively.
In particular, the adversary can access the entire software stack, including the hypervisor and OS of the instance hosting \textit{CP} and \textit{ES}s. 
They can obtain memory traces from the caches, memory bus, and RAMs, with the sole exception of accessing the processor to extract information. 
Moreover, the adversary can capture communications with other entities through traffic fingerprinting~\cite{Wang21,Panchenko2016,Cai12} and by exploiting ISP portals~\cite{NetflixOLAPortal,CDNLogTools}. 

In addition to software and hardware accesses, the adversary is collecting the access patterns of \textit{Subscriber}s on the \textit{CP} and \textit{ES}s.
The adversary can leverage the memory access pattern to raise attacks on encrypted storage, such as the leakage abuse attack~\cite{IslamKK12,Cash15}.
On the other hand, the access pattern can be exploited to conduct side-channel attacks against servers with TEE features~\cite{Brasser17,Gotzfried17,Xu15,Gruss17}.
We consider that adversaries in the \textit{CP} and \textit{ES}s can arbitrarily create \textit{Subscriber} accounts to collect the aforementioned information.

%
%

\noindent \textbf{Defence Assumption}. With real-world \textit{CDN service providers} having moved the \textit{CP} to the cloud (e.g., Netflix has moved to AWS~\cite{NetflixAWS}), we assume the TEE-enabled cloud instances~\cite{AWSTEE} are generally available.
We assume that the CDN deploys the advanced key management mechanism~\cite{Herwig2020AchievingKC} to frequently refresh the video decryption key at a low cost.
%
%
Furthermore, we note that Byzantine faulty behaviour can be prevented by existing countermeasures~\cite{Herwig2020AchievingKC}.
Besides, we assume that the communications between all parties involved in CDN are transmitted via a secure channel.

\noindent \textbf{Defence Limitation}. Our design does not address side-channel attacks that exploit speculative execution~\cite{Bulck18}, or implementation bugs~\cite{Jaehyuk17} against TEE.
We consider these attacks to be out of this work, and there are existing solutions available to address them.
For instance, MI6~\cite{Bourgeat19} and KeyStone~\cite{Lee20} can be employed to handle such attacks. 
Additionally, hardware providers have the option to release microcode patches to rectify these vulnerabilities~\cite{Bourgeat19,Costan161}.

We also do not consider side-channel attacks that exploit traffic bursts such as streaming bit-rate variation~\cite{Schuster17}, or ORAM volume attacks~\cite{BlackstoneKM20}.
These specific attacks can be mitigated by complementary works~\cite{Rishabh20} and~\cite{Kamara19}, respectively. 
Furthermore, our design does not consider attacks involving delay/denying content delivery, misrouting, and distributed DoS attacks~\cite{Guo2020CDNJB,Gilad2016CDNonDemandAA}.
Mainstream CDNs actively investigate and address them on a daily basis~\cite{cloudflare21}.
We also do not consider user anonymity (e.g., using Tor) due to copyright obligations from mainstream CDNs~\cite{NetflixCopyright22}.

%


\subsection{Formulating Design Goals}
\label{subsec:goals}

For ease of analysis, we will break down practical security goals (\textbf{SG}s), efficiency goals (\textbf{EG}s), and deployment requirements (\textbf{DR}) that align with mainstream CDNs~\cite{NetflixOCA} as follows:

\noindent \textbf{SG1}: \ul{\textbf{Video Confidentiality}}.
Video blocks of the requested \textsf{video} remain encrypted at the \textit{ES}s. 
Only the \textit{Subscriber}, who obtains the decryption keys for those blocks from the \textit{CP}, can decrypt them. 
The keys are refreshed per video request.

\noindent \textbf{SG2}: \ul{\textbf{Video Access Patterns}}. 
The \textit{ISP partner} cannot know which blocks of which video are sent to the uncompromised \textit{Subscriber}.

\noindent \textbf{SG3}: \ul{\textbf{Metadata Access Protection}}.
The \textit{cloud service provider} cannot know which video is requested by uncompromised \textit{Subscriber}.

\noindent \textbf{EG1:} \ul{\textbf{Efficient Computation}}.
The \textit{CP} does not operate on actual video blocks once the video has been pushed to the \textit{ES}s~\cite{Cloudflare212}.

\noindent \textbf{EG2: \ul{Low Bandwidth Overhead}}. 
Video blocks should not be redirected to the \textit{CP} for any purpose, e.g., cryptographic operations, following the current CDN design~\cite{NetflixOCA}.
Only low bandwidth is required to manipulate video \textsf{block}s between \textit{ES}s and \textit{CS}s.

\noindent \textbf{EG3:} \ul{\textbf{Low Fetch/Sync Latency}}.
The \textit{Subscriber} directly retrieves video blocks in sequential accesses from the \textit{ES}s, and the \textit{ES}s access/update blocks with optimal disk seeks. 

\noindent  \textbf{DR}: \ul{\textbf{CDN Compatibility}}. 
With mainstream CDNs moving the \textit{CP} to the clouds~\cite{NetflixAWS,AWSTEE}, we consider the usage of confidential computing (i.e., TEE) in the \textit{CP} to be compatible.
However, Internet-CDN infrastructures~\cite{Netflix212,Cloudflare212,Akamai18} deployed thousands of \textit{ES}s (e.g., Netflix-OCA devices~\cite{Netflix212}) globally. 
To be compatible with legacy devices, we consider the worst case when all \textit{ES}s do not have TEE.
Furthermore, we aim to avoid any solutions that rely on massive replicas and huge bandwidth costs when accessing the content to enable our design to be incrementally deployed on existing CDNs.  

%
%
%
%


\begin{figure}[!t]
\centering
\includegraphics[width=\linewidth]{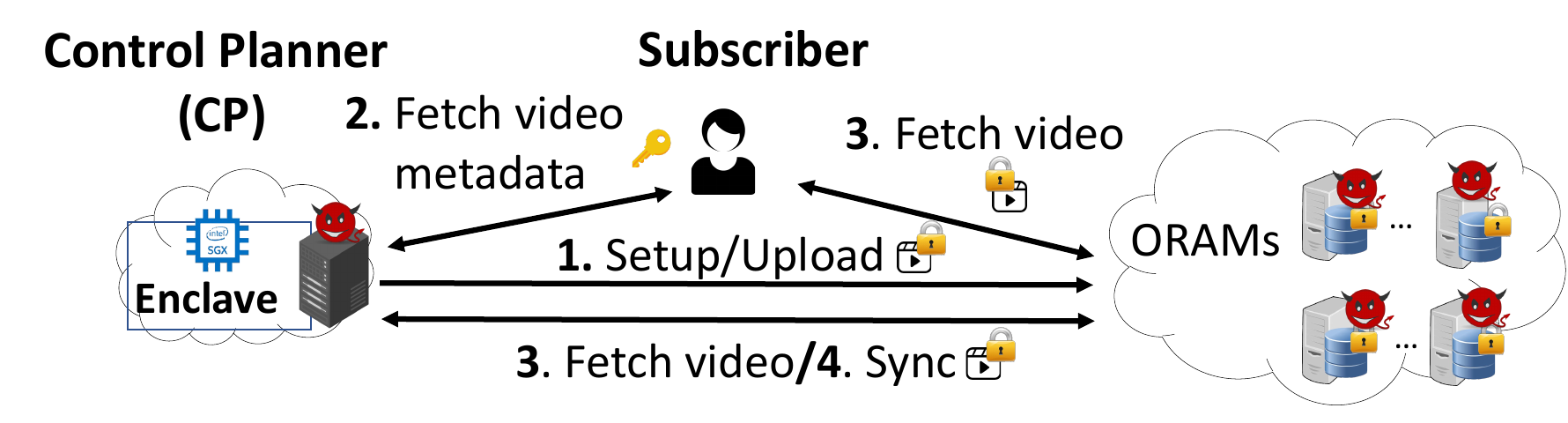}
\vspace{-15pt}
\caption{Strawman Design}
\vspace{-10pt}
\label{fig:strawman}
\end{figure}

\subsection{A Strawman CDN design using ORAM}
\label{subsec:strawman} 
%
%
%
With the above assumptions, it is intuitive to have a Strawman design that relies on TEE and Path ORAM solutions~\cite{Stefanov13,Chakraborti191}. 
As shown in Figure~\ref{fig:strawman}, the Strawman design leverages the \textit{Enclave} within the \textit{CP} as the ORAM client and the ORAM storage is distributed among the \textit{ES}s.
The \textit{Enclave} maintains two data structures, namely, 
\begin{itemize}[leftmargin=*]
 \item $\textsf{VideoMap}$, which stores the map between a video ID $vid$ and the corresponding block identifier list $\{bid_i\}, i\in[0,V)$, where $V$ is the number of blocks for a video.  
 \item $\textsf{BlockStateMap}$, which stores the map between a block identifier $bid$ and its block \textsf{state} $(c,l)$, where $c$ the latest \textbf{counter} for AES encryption/decryption and $l$ is the \textbf{path} information from the standard ORAM position map (e.g., $\mathsf{position}[bid]$, see Section~\ref{sec:preliminaries}).
\end{itemize}
, and it instantiates the four CDN operations as follows:

\noindent \ul{\textbf{Setup}:}
The \textit{CDN service provider} runs remote attestation with the \textit{CP} to attest the \textit{Enclave} code. Then, the \textit{Enclave} establishes connections to all \textit{ES}s and initialises ORAM trees on \textit{ES}s.

\noindent \ul{\textbf{Upload}:} The \textit{CDN service provider} uploads the \textsf{video} to the \textit{Enclave} through the secure channel established during attestation. 
The \textit{Enclave} firstly trunks the video into $V$ blocks $b_i$ and assigns those blocks' block states $(c_i,l_i)$ into $\textsf{BlockStateMap}[bid_i]$. 
Then, the \textit{Enclave} runs $\mathsf{ORAM.Access}(`\textsf{write}', bid_i, T, l_i, S, b_i)$ to write those video blocks. After running $\mathsf{ORAM.Access}$, the video blocks will be placed into the ORAM stash $S$ inside the \textit{Enclave}.
The $\{bid_i\}$ list of those blocks will be added into the $\textsf{VideoMap}[vid]$.

\noindent \ul{\textbf{Fetch}:} Given a requested $vid$, the \textit{Enclave} queries the $\textsf{VideoMap}$ and $\textsf{BlockStateMap}$ and responds to the \textit{Subscriber} the identifier list $\{bid_i\}$ and the \textsf{state}s $(c_i,l_i)$ of the requested video.
Then, the \textit{Subscriber} runs $\mathsf{ORAM.Access}(`\textsf{read}', bid_i, T, l_i, S, b_i)$ for $i\in[0,V)$ to download all video blocks from the \textit{ES}s following the metadata.
The \textit{Enclave} also runs the same $\mathsf{ORAM.Access}$ operation to download the retrieved blocks into the ORAM stash $S$ on the background and updates the path and counter $(c,l)$ of accessed $bid$ following the Path ORAM protocol accordingly.

\noindent \ul{\textbf{Sync}:} After each $\mathsf{ORAM.Access}$ operation, the \textit{Enclave} refers to $\textsf{BlockStateMap}$ to run $\mathsf{ORAM.Evict}(T,\textsf{BlockStateMap},S,l_i)$ and flush the content inside the ORAM stash $S$ into accessed path $l_i$ in the ORAM storage $T$ on \textit{ES}s.

This design can meet all desired security goals (\S\ref{subsec:goals}) as follows: 

\noindent \textbf{SG1}: All video blocks remain encrypted at the \textit{ES}s. Only the \textit{Enclave} in the \textit{CP} can generate and share decryption keys with \textit{Subscriber}s. The video block will be re-encrypted in each \textbf{Sync} operation.

\noindent \textbf{SG2}: With ORAM, fetched blocks are always re-encrypted and stored in new random ORAM paths.
Hence, the \textit{ES}s cannot distinguish which video and its blocks have been accessed (see \S\ref{sec:preliminaries}).

\noindent \textbf{SG3}: The video metadata is placed inside the \textit{Enclave}, which is inaccessible to the untrusted part of the \textit{CP}. 
Moreover, we follow prior works (\cite{Sasy17,Mishra18,Lai21}, just list a few) to use oblivious data structures (Path ORAM) and operations (see Appendix~\ref{subsec:oblivEnclave}) to store all data structures (i.e., $\textsf{VideoMap}$, $\textsf{BlockStateMap}$, and ORAM Stash $S$) inside the \textit{Enclave}.
This effectively hides the \textit{Enclave}'s memory access patterns and prevents relevant side-channel attacks~\cite{Lee2017InferringFC,Brasser17,Bulck2017TellingYS,Wang2017LeakyCO}. 

However, we notice that few requirements are partially or hardly achieved under the Strawman design, which include:

\noindent{\textbf{EG1}: \textbf{Computation Bottleneck}}.
The \textit{CP} directly manipulates the video block inside the \textit{Enclave}, which causes a non-trivial computation cost in the \textbf{Sync} operation.

\noindent{\textbf{EG2}: \textbf{Intercontinental Bandwidth Blowup}}.
Huge intercontinental communication costs are incurred in \textbf{Fetch} and \textbf{Sync} operations as the \textit{CP} needs to refresh all accessed blocks within the \textit{Enclave}.

\noindent{\textbf{EG3}: \textbf{Long Fetch/Sync Latency}}. The ORAM incurs an $\mathcal{O}(\log_2N)$ cost~\cite{Stefanov13} to access or evict one block, resulting in a noticeable delay when fetching or synchronising an actual video (with many blocks). 

\noindent{\textbf{DR}: \textbf{Subscriber Bandwidth Blowup}}. The Strawman is compatible with the current CDN structure~\cite{NetflixAWS,NetflixOCA}, having the TEE on \textit{CP} only. 
However, \textit{Subscriber}s have to download $\mathcal{O}(\log_2N)$ blocks for each targeted video block which leads to a surging bandwidth cost.

\begin{table}
  \centering
  \small
  \caption{Schemes against design goals in \S\ref{subsec:goals}}
  \vspace{-5pt}
   \label{tab:strawman_sec_eff}
   \tabcolsep 0.025in
   \footnotesize{\tikzcircle[black, fill=black]{2.5pt}\textemdash fully achieve, \halfcirc[2.5pt] \textemdash partially achieve, \tikzcircle[black, fill=white]{2.5pt} \textemdash not meet}\\
   \begin{tabular}{cccccccc}
    \toprule
      Scheme &  \textbf{SG1}  &  \textbf{SG2}  & \textbf{SG3}  & \textbf{EG1} & \textbf{EG2}  & \textbf{EG3}  & \textbf{DR} \\
    \midrule
    \centering
        Strawman &  \CIRCLE   &  \CIRCLE    & \CIRCLE  & \Circle & \Circle & \Circle & \LEFTcircle   \\
        rORAM-Strawman &  \CIRCLE   &  \CIRCLE & \CIRCLE &  \Circle &  \Circle & \CIRCLE & \LEFTcircle   \\
        Video Data/Metadata Separation &  \CIRCLE    &  \LEFTcircle     & \CIRCLE  & \CIRCLE & \CIRCLE  & \CIRCLE   &  \LEFTcircle \\
        \textit{CP} with Distributed Trust &  \CIRCLE    & \CIRCLE     & \CIRCLE  & \CIRCLE & \CIRCLE  & \CIRCLE  &  \CIRCLE \\
    \bottomrule
  \end{tabular}
  \vspace{-10pt}
\end{table}

\section{OblivCDN Design}
\label{sec:main}
In this section, we present \textsf{OblivCDN}, a CDN system that achieves all design goals set in Section~\ref{subsec:goals}.
We will show our three attempts to gradually refine our design to meet the final requirement.
In each attempt, we will elaborate on how we design customised components and integrate them into the CDN system to resolve specific challenges while maintaining other properties.
In Table~\ref{tab:strawman_sec_eff}, we summarise the security and efficiency of all our designs. 

\subsection{Attempt 1: rORAM-Strawman}
Without changing the system architecture in Figure~\ref{fig:strawman}, we can simply achieve the \textbf{EG3} requirement by replacing generic the Path ORAM instance with rORAM.
In particular, with rORAM, the data structures inside the \textit{Enclave} are modified as follows:
\begin{itemize}[leftmargin=*]
 \item $\textsf{VideoMap}$ stores the map between a video ID $vid$ and the corresponding block identifier list with ranges $\{(bid_i,r_i)\}, i\in[0,R_V)$, where $R_V$ is the number of ranges generated from the video. 
\end{itemize}
Also, rORAM operations (\textcolor{blue}{blue} operations) will be invoked instead of generic Path ORAM operations in CDN operations:

\noindent \ul{\textbf{Setup}:}
The \textit{CDN service provider} runs remote attestation with the \textit{CP} to attest the \textit{Enclave} code. Then, the \textit{Enclave} establishes connections to all \textit{ES}s and initialises \textcolor{blue}{rORAM storage $E_r$} on \textit{ES}s.

\noindent \ul{\textbf{Upload}:} The \textit{CDN service provider} uploads the \textsf{video} to the \textit{Enclave} through the secure channel established during attestation. 
The \textit{Enclave} firstly trunks the video into $V$ blocks and \textcolor{blue}{groups those blocks into $R_V$ ranges with $2^{r_i}$ consecutive $bid$s starting with $bid_i$}.
Their block states \textcolor{blue}{$\{(c_i,l_i),... ,(c_{i+2^{r_i}-1},l_{i+2^{r_i}-1})\}$} are added into the corresponding entries in $\textsf{BlockStateMap}$. 
Then, the \textit{Enclave} runs \textcolor{blue}{$E_{r_i}.\textsf{ReadRange}(bid_i)$} to write those video blocks. After running \textcolor{blue}{$E_r.\textsf{ReadRange}$}, the video blocks will be placed into the \textcolor{blue}{rORAM stash $S_r$} inside the \textit{Enclave}.
The \textcolor{blue}{$\{(bid_i,r_i)\}$} list of those blocks will be added into the $\textsf{VideoMap}[vid]$.

\noindent \ul{\textbf{Fetch}:} Given a requested $vid$, the \textit{Enclave} queries the $\textsf{VideoMap}$ and $\textsf{BlockStateMap}$ and responds to the \textit{Subscriber} the identifier list \textcolor{blue}{$\{(bid_i,r_i)\}$} and the \textsf{state}s \textcolor{blue}{$\{(c_i,l_i),... ,(c_{i+2^{r_i}-1},l_{i+2^{r_i}-1})\}$} of the requested video.
Then, the \textit{Subscriber} runs \textcolor{blue}{$E_{r_i}.\textsf{ReadRange}(bid_i)$} for $i\in[0,R_V)$ to download all video blocks from the \textit{ES}s following the metadata.
The \textit{Enclave} also runs the same \textcolor{blue}{$E_{r_i}.\textsf{ReadRange}$} operation to download the retrieved blocks into the \textcolor{blue}{rORAM stash $S_r$} on the background and update the path and counter $(c,l)$ of accessed $bid$s following the rORAM protocol accordingly.

\noindent \ul{\textbf{Sync}:} After each \textcolor{blue}{$E_r.\textsf{ReadRange}$} operation, the \textit{Enclave} refers to $\textsf{BlockStateMap}$ to run \textcolor{blue}{$E_r.\textsf{BatchEvict}(2^r)$} and flush the content inside the \textcolor{blue}{rORAM stash $S_r$} into $2^r$ accessed paths into the \textcolor{blue}{rORAM storage $E_r$} on \textit{ES}s.

As a result, \textbf{EG3} can be fully satisfied as follows:

\noindent \textbf{EG3}: With rORAM, the \textit{Subscriber} can fetch required video blocks in sequence after only $\mathcal{O}(\log_2N)$ seeks on targeted \textit{ES}s, independent of the number of blocks requested. The \textit{CP} can also update the video block in batch with optimal $\mathcal{O}(\log_2N)$ disk seeks.



\begin{figure}[!t]
\centering
\includegraphics[width=\linewidth]{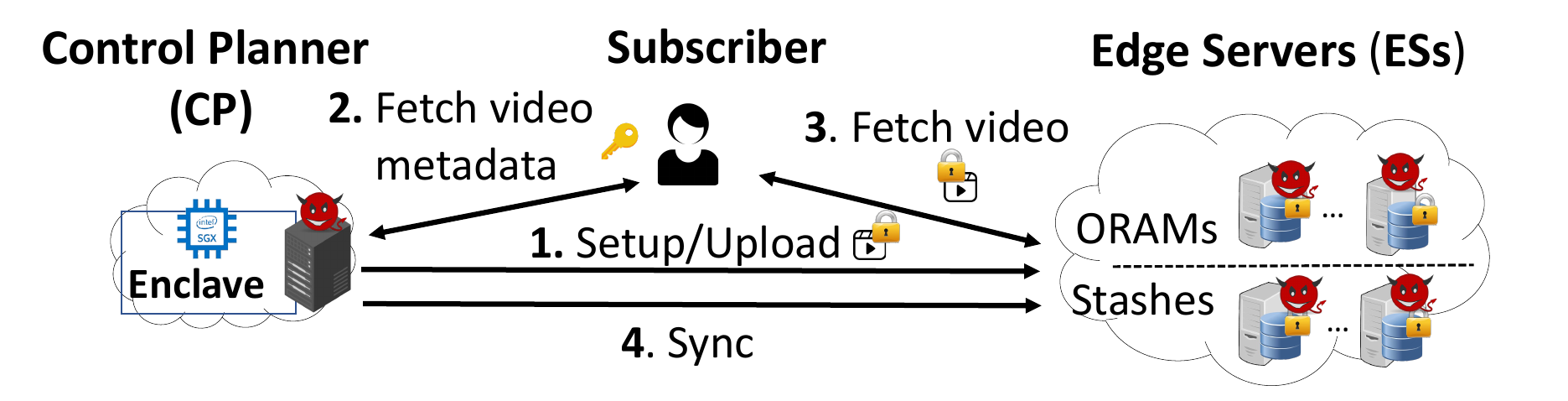}
\vspace{-15pt}
\caption{Video Data/Metadata Separation}
\vspace{-5pt}
\label{fig:separation}
\end{figure}

\subsection{Attempt 2: Video Data/Metadata Separation}
\noindent{\textbf{Overview.}} In the Strawman and rORAM-Strawman approaches, one major bottleneck is the substantial computational and communication costs associated with synchronising video blocks between \textit{CP} and \textit{ES}s.
Also, \textit{Subscriber}s incur a huge bandwidth cost ($\mathcal{O}(2^r\times\log_2N)$) to fetch targeted $2^r$ video blocks.
In order to achieve \textbf{EG1}, \textbf{EG2} and \textbf{DR}, we devise a customised rORAM protocol and integrate it into the CDN system to minimise the computation and communication between \textit{CP} and \textit{ES}s, as well as the bandwidth cost for \textit{Subscriber}s to download video blocks.

As shown in Figure~\ref{fig:separation}, there are two major changes upon the rORAM-Strawman design:
First, we split the \textit{ES}s into two groups with an equal amount of instances.
While the first group of those instances still host rORAM storage $E_r$ as in the rORAM-Strawman, the second instance group is dedicated to hosting the rORAM stash $S_r$ of each rORAM storage.
This design allows the \textit{CP}/\textit{Subscriber}s to outsource all real video data and operations to \textit{ES}s and thus avoids expensive data block manipulation inside the \textit{Enclave} and bandwidth cost between \textit{CP} and \textit{ES}s as well as \textit{Subscriber}s and \textit{ES}s.

To foster the first change, we introduce several new data structures on the \textit{CP} to manage the metadata of CDN videos.
Those data structures keep track of the real video blocks' location information and assist the \textit{CP} in instructing \textit{ES}s to manipulate real video blocks without retrieving them.
With the new data structure, when the \textit{CP} receives any \textbf{Fetch} requests from the \textit{Subscriber}, it lets the \textit{Subscriber} obtain real video blocks from \textit{ES}s via rORAM operations while performing the same \textbf{Fetch} operation on the metadata.
Later, when the \textit{CP} executes the \textbf{Sync} operation, the \textit{Enclave} performs ORAM eviction on the metadata and then remotely \textit{instructs} the \textit{ES}s to keep the ORAMs storing actual video blocks synchronised.

\begin{center}
\begin{table}[!t]

  \caption{Data structures in CDN Entities after Separating Video Data and Metadata}
  \vspace{-5pt}
  \centering
  \footnotesize
  \label{tab:components}
  \begin{tabular}{|p{1.5cm}|c|c|}
    \hline
    \multirow{2}{*}{\parbox{1.7cm}{\centering \textit{Control Plane (CP)}}} & \cellcolor[gray]{0.9} & \cellcolor[gray]{0.9} \textsf{MetaStash}$_r$ and \textsf{BlockTracker}$_r$, $r \in [0,\textnormal{log}_2R]$\\
    \cline{3-3}
    & \multirow{-2}{*}{\cellcolor[gray]{.9}\textit{Enclave}}  & \cellcolor[gray]{0.9} \textsf{VideoMap} and \textsf{BlockStateMap}\\
        \specialrule{.3em}{.05em}{.05em}

        \multirow{2}{*}{\parbox{1.7cm}{\centering \textit{Edge Servers (ESs)}}}  & $E_r$  &  {\centering ORAM supports $2^r$ sequential blocks access}\\ 
  
        \cline{2-3}
        
        & $S_r$ &  {\centering The corresponding stash of $E_r$} \\ 
    \hline 
  \end{tabular}
  \vspace{-10pt}
\end{table}
\end{center}

\noindent{\textbf{New Data Structures.}} To implement the above design,  we introduce crucial data structures as outlined in Table~\ref{tab:components}.
We now detail the design rationale of these structures as follows:

On the \textit{ES} side, each \textit{ES} will be assigned to host the rORAM storage $E_r$ or the stash $S_r$ on the \textit{ES}s, where $r \in [0,\textnormal{log}_2R]$.
The \textit{ES}s for the rORAM storage still initialises $E_r$, where $r\in[0, \log_2{R}]$, to enable access to $2^r$ sequential video blocks following the rORAM scheme (\S\ref{sec:preliminaries}).
On the other hand, the \textit{ES}s for the rORAM stash have the predefined parameters for stash size and its counterparty \textit{ES}s with $E_r$, so it can generate $S_r$ and pair with a specific $E_r$ to execute $E_r.\textsf{ReadRange}$ and {$E_r.\textsf{BatchEvict}$ on behalf of the \textit{CP}/\textit{Subscriber}s\footnote{With $S_r$ on \textit{ES}s, the stash only stores the encrypted video blocks fetched from $E_r$ to preserve \textbf{SG1}. The \textit{Subscriber}s will keep the decryption counters $c$ locally after receiving block states from \textit{CP} and use them decrypt the video blocks from $S_r$.}.

On the \textit{CP} side, we utilise a small memory within the \textit{Enclave} to store the metadata for $E_r$ and $S_r$, denoted as $\textsf{BlockTracker}_r$ and $\textsf{MetaStash}_r$, respectively.
As shown in Figure~\ref{fig:separated_ops}, each $\textsf{BlockTracker}_r$ and $\textsf{MetaStash}_r$ have the same ORAM and stash sizes and structure as their counterparty $E_r$ and $S_r$ on \textit{ES}s.
The only difference is that $\textsf{BlockTracker}_r$ and $\textsf{MetaStash}_r$ store the video block identifiers ($bid$), while $E_r$ and $S_r$ store the actual video blocks.
These meta structures facilitate the following oblivious accesses:
\begin{itemize}[leftmargin=*]

\item $\textsf{MetaData}_r.\textsf{ReadRange}(bid)$: Load a list of $2^r$ block identifiers $[bid,bid+2^r)$ from $\textsf{BlockTracker}_r$ to $\textsf{MetaStash}_r$ following rORAM $\textsf{ReadRange}$. 

\item $\textsf{MetaData}_r.\textsf{BatchEvict}(k)$: Write $k$ paths in $\textsf{BlockTracker}_r$ with block identifiers from $\textsf{MetaStash}_r$ following rORAM $\textsf{BatchEvict}$.
\end{itemize}
The use of these meta structures enables the \textit{Enclave} to perform the range ORAM access/eviction locally within the \textit{CP} without manipulating CDN video blocks stored in the \textit{ES}s.
Specifically, the \textit{Enclave} monitors the routing traces of blocks between the $\textsf{BlockTracker}_r$ and $\textsf{stash}_r$ with \textbf{Upload} and \textbf{Fetch} operations.
After that, these traces enable the \textit{Enclave} to execute logical eviction over $bid$s in $\textsf{MetaStash}_r$ and run the \textbf{Sync} operation with the \textit{ES}s to perform the same eviction on video blocks with much shorter instructions.

\noindent{\textbf{Updated CDN Operations.}} As shown in Figure~\ref{fig:separated_ops}, the CDN operations will be updated with new data structures. 
Particularly, \textcolor{blue}{blue} operations are used compared to the rORAM-Strawman as follows:

\noindent \ul{\textbf{Setup}:} The \textit{CDN service provider} runs remote attestation with the \textit{CP} to attest the \textit{Enclave} code. Then, the \textit{Enclave} establishes connections to all \textit{ES}s and initialises $E_r$ and \textcolor{blue}{$S_r$} on \textit{ES}s.

\begin{figure}[!t]
\centering
\includegraphics[width=0.9\linewidth]{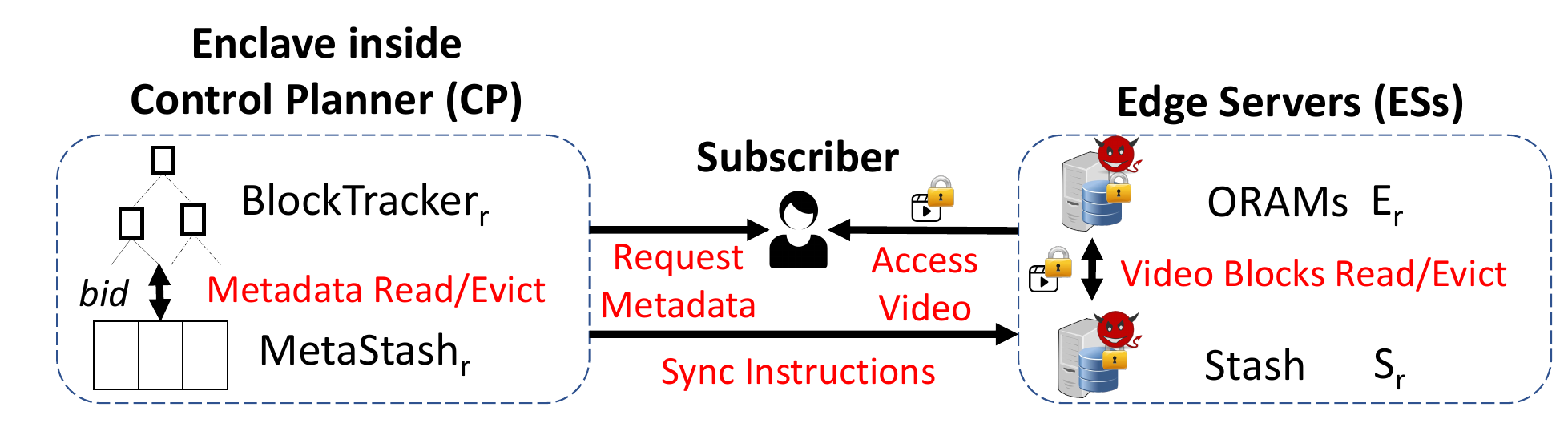} 
\caption{New Data Structures and Updated Operations with Separated Video Data/Metadata}
\vspace{-10pt}
\label{fig:separated_ops}
\end{figure}

\noindent \ul{\textbf{Upload}:} The \textit{CDN service provider} uploads the \textsf{video} to the \textit{Enclave} through the secure channel established during attestation. 
The \textit{Enclave} firstly trunks the video into $V$ blocks and groups those blocks into $R_V$ ranges with $2^{r_i}$ consecutive $bid$s starting with $bid_i$.
The \textit{Enclave} runs the \textcolor{blue}{$\textsf{MetaData}_{r_i}.\textsf{ReadRange}(bid_i)$ to put $bid$s in the range starting with $bid_i$ into $\textsf{MetaStash}_{r_i}$.}
Their block states $\{(c_i,l_i),... ,(c_{i+2^{r_i}-1},l_{i+2^{r_i}-1})\}$ are added into the corresponding entries in $\textsf{BlockStateMap}$. 
Then, the \textit{Enclave} runs $E_{r_i}.\textsf{ReadRange}(bid_i)$ to write those video blocks. After running $E_r.\textsf{ReadRange}$, the video blocks will be \textcolor{blue}{encrypted and placed into the $S_r$ inside the \textit{ES}s}.
The $\{(bid_i,r_i)\}$ list of those blocks will be added into the $\textsf{VideoMap}[vid]$.

\noindent \ul{\textbf{Fetch}:} Given a requested $vid$, the \textit{Enclave} queries the $\textsf{VideoMap}$ and $\textsf{BlockStateMap}$ and responds to the \textit{Subscriber} the identifier list $\{(bid_i,r_i)\}$ and the \textsf{state}s $\{(c_i,l_i),... ,(c_{i+2^{r_i}-1},l_{i+2^{r_i}-1})\}$ of the requested video.
Then, the \textit{Subscriber} runs $E_{r_i}.\textsf{ReadRange}(bid_i)$ for $i\in[0,R_V)$ to instruct \textit{ES}s to \textcolor{blue}{load video blocks from $E_r$ to $S_r$}.
\textcolor{blue}{The \textit{Subscriber} can directly download the video blocks from \textit{ES}s with $S_r$.}
Meanwhile, the \textit{Enclave} runs \textcolor{blue}{$\textsf{MetaData}_{r_i}.\textsf{ReadRange}(bid_i)$} to load $bid$s of the retrieved blocks into the \textcolor{blue}{$\textsf{MetaStash}_{r_i}$} on the background and update the path and counter $(c,l)$ of accessed $bid$s following the rORAM protocol accordingly.

\noindent \ul{\textbf{Sync}:} After each \textcolor{blue}{$\textsf{ReadRange}$} operation, the \textit{Enclave} follows the $\textsf{BlockStateMap}$ to run the \textcolor{blue}{$\textsf{MetaData}_r.\textsf{BatchEvict}(2^r)$ protocol to simulate the eviction on accessed $bid$s}. Then, \textcolor{blue}{the \textit{Enclave} follows the eviction result in the $\textsf{BlockTracker}_r$ to send instructions (i.e., refreshed AES-CTR masks and disk locations generated from the block state $(c,l)$) to \textit{ES}s with $S_r$. Those \textit{ES}s will re-encrypt video blocks and place them in designated locations inside \textcolor{blue}{$E_r$} on \textit{ES}s}.

This attempt is promising to minimise the intercontinental communication cost, as\textit{CP} only sends short instructions to \textit{ES}s.
It also allows \textit{Subscriber}s to only download $2^r$ required video blocks from \textit{ES}s because the ORAM access has been done between \textit{ES}s.
However, we realise that the attempt introduces new security threats.
This is because the read/write instructions from the \textit{CP} reveal the original storage location of a block inside the stash and the next location to write it to \textit{ES}s.
Since \textit{ES}s are deployed by a semi-honest party, i.e., the \textit{ISP partner}, this is a direct violation of the \textbf{SG2} requirement.
Although this can be remedied by employing TEE with oblivious primitives on \textit{ES}s, it also causes compatibility issues, since we do not expect \textit{ISP partner}s to equip their legacy devices with TEE (\textbf{DR}).

\begin{figure}[!t]
\centering
\includegraphics[width=\linewidth]{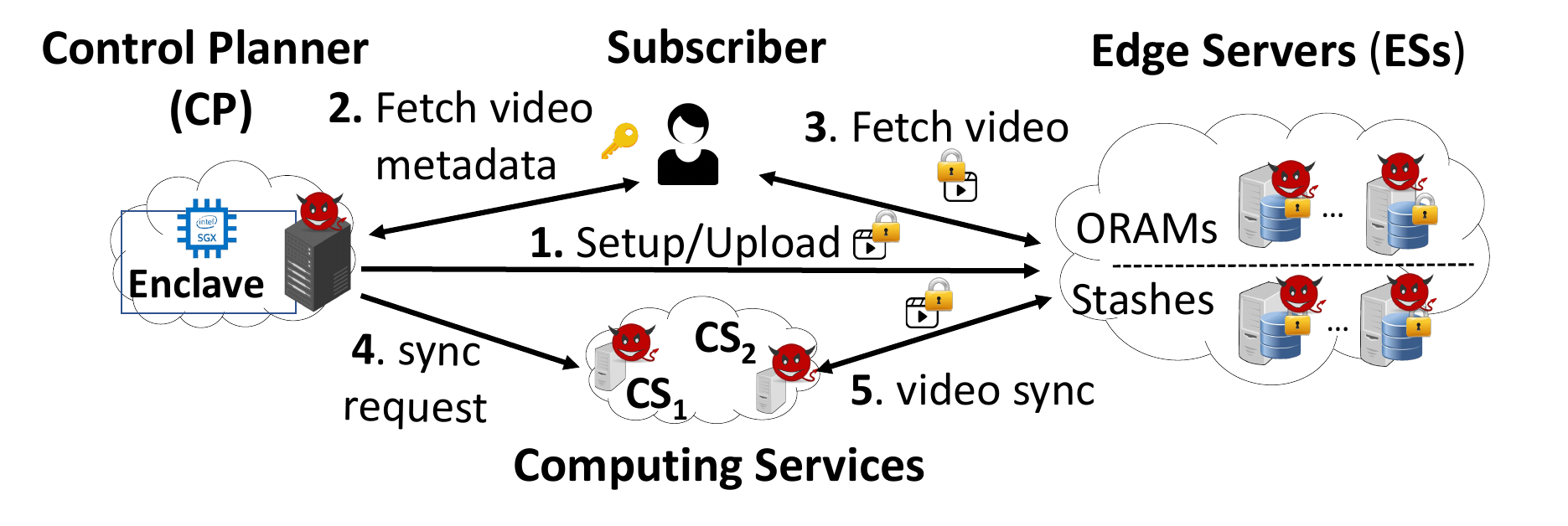}
\vspace{-15pt}
\caption{\textit{CP} with Distributed Trust}
\vspace{-10pt}
\label{fig:oblivcdn}
\end{figure}

\subsection{Attempt 3: \textit{CP} with Distributed Trust}\label{subsec:cp}
\noindent{\textbf{Overview.}} In Attempt 2, the core design is to delegate a part of the rORAM client to \textit{ES}s to alleviate the communication cost.
However, as we discussed above, this design violates the trust assumptions in \S\ref{subsec:threatmodel}.
To resolve this issue, our goal is to allow the party who manipulates video blocks to follow \textit{CP}'s instructions without knowing them, thus removing the trust requirement on that party.

As shown in \S\ref{sec:overview}, real-world \textit{CDN service provider}s have actively sought to deploy the CDN service with multiple cloud providers~\cite{AkamaiAlternative}.
In light of this observation, we consider integrating the distributed trust model~\cite{Dauterman20,Dauterman22,Eskandarian21} into CDN systems to make a secure split on \textit{CP} functionality.
Particularly, two semi-honest servers, denoted as computing service node \textit{CS}s ($\textit{CS}_1$ and $\textit{CS}_2$), are hired from non-colluding \textit{cloud service provider}s or \textit{ISP partner}s\footnote{The non-colluding assumption can be effectively achieved with local law and conflict of interests, especially for two providers in the same region~\cite{kamara2011outsourcing}.} that are close to \textit{Subscriber}s (see Figure~\ref{fig:oblivcdn}).
The \textit{CDN service provider} delegates \textit{CP}'s ORAM read/write operations into \textit{CS}s by instructing them to run oblivious video block routing protocols devised for CDN systems.
The new design preserves the benefits of separating video data/metadata.
Moreover, it ensures that, as far as the adversary can only compromise one \textit{CS}, he/she cannot learn the video access pattern. 
This security guarantee is valid even if the adversary is on the \textit{ISP partner} side, i.e., can compromise one \textit{CS} and \textit{ES}s.

\begin{figure}[!t]
\vspace{-2pt}
\centering
\includegraphics[width=1\linewidth]{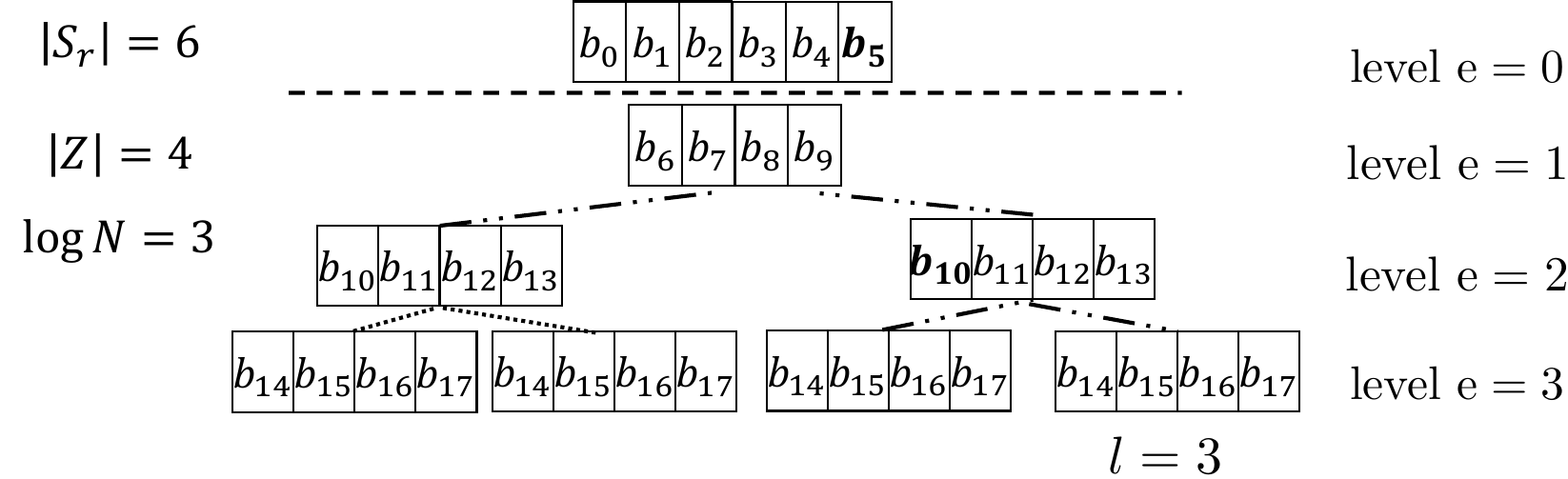}
\vspace{-15pt}
\caption{\textsf{OblivCDN} labels blocks with the index $b_s$ format, where $s$ signifies the block order. With the stash size $|S_r|=6$ and a 8-node ORAM tree with node size $Z=4$, $s\in[0,17]$ ($0$ to $|S_r|+\log_2N \times Z - 1$) will be assigned to blocks in the stash and each path in an ascending order.}
\vspace{-12pt}
\label{fig:index_oram}
\end{figure}

\noindent{\textbf{Video Block ID Allocation.}} To enable oblivious video block routing, our insight is that for a tree-based ORAM, it is always has a fixed number of blocks ($|S_r|+\log_2N \times Z$) in the stash and a given path while the required block is hidden in those blocks.
Therefore, if we assign a unified block ID to those blocks, we can locate the required block by a given path and ID.
In particular, as shown in Figure~\ref{fig:index_oram}, we denote each block in ORAM as $b_s$, where $s$ signifies block order in the path, which is determined as follows: 
\begin{enumerate}[leftmargin=*]
	\item The stash occupies the $s\in[0,|S_r|)$.
	\item For each path, the block will be assigned with $s$ in ascending order from $|S_r|$ to $|S_r|+\log_2N \times Z - 1$.
\end{enumerate}
We can easily target a specific block with its path and the newly-added ID, i.e., if $s<|S_r|$, we can find it in the $s$ position at the stash.
Otherwise, we can use $s$ to compute its level $e$ via $[\left(s-|S_r|\right)/Z] +1$ and offset $o$ in the level by $\left(s - |S_r|\right)~\%~Z$.
For instance, $\mathbf{b_{5}}$ is at the $5^{th}$ position in the stash, and $\mathbf{b_{10}}$ in path $l=3$ is at the $0$ position of the second level at path $l=3$ according to the setting in Figure~\ref{fig:index_oram}.

We then adapt the $\textsf{BlockStateMap}$ to store the map between a block identifier $bid$ and the block \textsf{state} $(c,l,s)$, which enables the \textit{CP} to compute the block location in path $l$ with the block \textsf{state}.

\noindent{\textbf{Oblivious Block Routing.}} The unified block ID in the path enables the \textit{CP} to generate instructions for two \textit{CS}s to jointly function as a rORAM client.
In particular, for $E_r.\textsf{ReadRange}$ operations, we consider it as a two-step process, i.e., 1) fetching $2^r$ blocks from $E_r$ to $S_r$ and 2) accessing $S_r$ for required blocks.

For the block fetching process, each targeted path can be represented as a vector with $\log_2N\times Z$ elements, while the \textit{CP} knows the exact location of targeted video blocks in that path.
Therefore, we can convert each block fetching process to a PIR process that aims to retrieve one block from a fixed offset at the targeted path in $E_r$.
We implement the following \textsf{BlockRangeReturn} operation with DPF operations (see \S\ref{sec:preliminaries}) that enables efficient $\textsf{ReadRange}$ from $E_r$ under the \textit{CP} with distributed trust setting:

\begin{figure}
\vspace{-15pt}
\centering
\subfloat[BlockRangeRet]{
		\label{fig:pri_range_ret}
		\includegraphics[width=0.5\linewidth,]{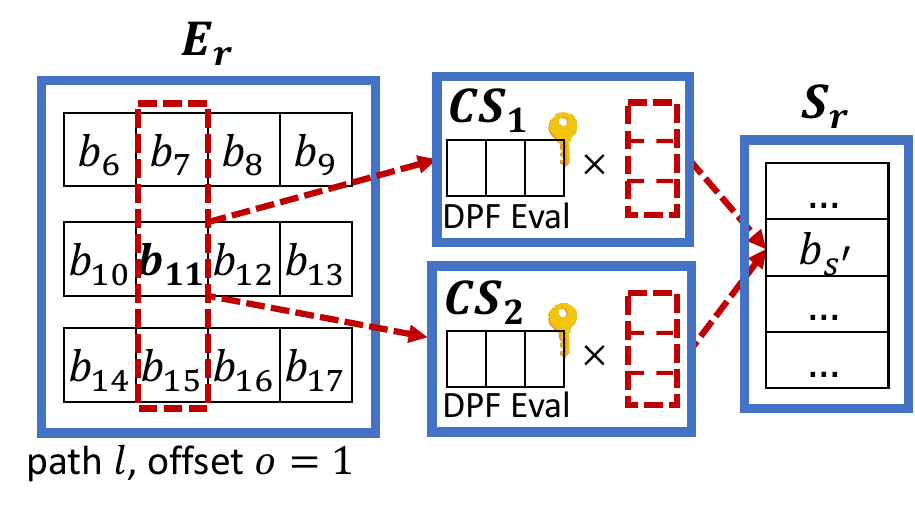}
	}
	\hspace{20pt}
	\subfloat[PerReFunc]{
		\label{fig:per_re_func}
		\includegraphics[width=0.3\linewidth]{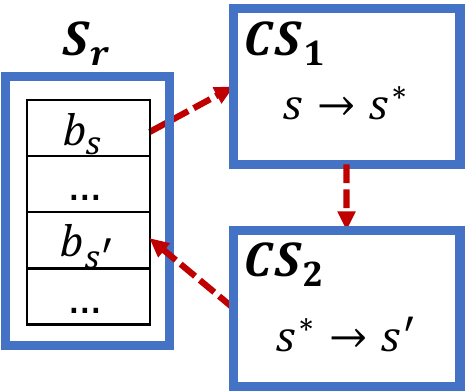}
	}
	\vspace{-5pt}
\caption{Examples on how to manipulate one block under our rORAM \textsf{ReadRange} design with distributed trust.}
\vspace{-15pt}
\label{fig:evalpath}
\end{figure}

\subsubsection{\textsf{BlockRangeReturn}}  
\label{subsubsec:BlockRangeRet}

\textsf{BlockRangeRet} obliviously routes video blocks from $E_r$ to $S_r$ without revealing which specific blocks are being routed to both the $E_r$ to $S_r$, as well as to each of $\textit{CS}_{p\in\{1,2\}}$.

We use an example of \textsf{BlockRangeRet} in Figure~\ref{fig:pri_range_ret} to demonstrate how it proceeds.
In the example, $b_{11} (s=11)$, which is encrypted with counter $c$ under the AES-CTR mode, is routed from $E_r$ to the new address $s^\prime$ in $S_r$ and re-encrypted with $c'$ as follows:

\noindent\textbf{On the \textit{Enclave}}:
\begin{enumerate}[leftmargin=*]
	\item Generates ($k_{1},k_{2}$) with $\textsf{DPF}.\textsf{Gen}(1^\lambda,11)$.
	\item Computes the re-encryption mask $m=G(c)\oplus G(c')$ with PRF $G(\cdot)$ and shares it as $m_1$ and $m_2$, where $m_1\oplus m_2=m$.
	\item Computes the offset $o=1$ with $Z=4$ and $\log N=3$.
	\item Gives the instruction $(k_{p}, m_p, l, o, s')$ to the query party (can be either the \textit{Enclave} in \textbf{Upload} or the \textit{Subscriber}s in \textbf{Fetch}), who then sends the instruction to $\textit{CS}_{p\in\{1,2\}}$.
\end{enumerate}
\noindent\textbf{On each $\textit{CS}_{p\in\{1,2\}}$}}:
\begin{enumerate}[leftmargin=*]
	\item Executes $y[1]\leftarrow\textsf{DPF}.\textsf{Eval}(k_{p},7)$, $y[2]\leftarrow\textsf{DPF}.\textsf{Eval}(k_{p},11)$, $y[3]\leftarrow\textsf{DPF}.\textsf{Eval}(k_{p},15)$ (i.e., $y[e]\leftarrow\textsf{DPF}.\textsf{Eval}(k_{p},|S_r|+(e-1)\times Z+o), e\in[1,\log N]$) for all block IDs in the offset.
	\item Retrieves corresponding encrypted blocks ($b_7, b_{11}, b_{15}$) from the path $l$ on $E_r$ into vector $b$.
	\item Computes shared block $b_{11}$ with $b_p\leftarrow\oplus_{e=1}^{\log N}(y[e]\times b[e])$.
	\item Re-encrypts the share with shared key $b_p\leftarrow b_p\oplus m_p$.
	\item Jointly writes $b_1\oplus b_2=b_{11}\oplus G(c)\oplus G(c')$ into the $s'$ position on $S_r$, so the original mask $G(c)$ on $b_{11}$ will be canceled and replaced with $ G(c')$.
\end{enumerate}
\textsf{BlockRangeRet} efficiently protects the targeted block against adversaries who can compromise one $\textit{CS}$ and $\textit{ES}$s at most since all blocks in the same offset of the path are accessed (no access pattern leakage).
Moreover, the final block will be re-encrypted before loading into $S_r$, so the adversary on $\textit{ES}$s cannot correlate the blocks from $E_r$ to that on $S_r$.
We note that \textsf{BlockRangeRet} can execute in a batch setting of $2^r$ video blocks.
This setting optimises disk seeks on $E_r$ and saves communication costs between $E_r$ and $\textit{CS}_{p\in\{1,2\}}$.
Algorithm~\ref{alg:BlockRangeRet} in Appendix~\ref{sec:oblivious} presents \textsf{BlockRangeRet} in detail.

For the stash access process, we should stress that the major difference between $\textsf{ReadRange}$ in the original rORAM and our design is that the stash is also placed in an untrusted party where the access pattern should also be suppressed. 
Therefore, we design one extra primitive to \textsf{Permute and Re-encryption} blocks inside $S_r$ after \textit{CP}/\textit{Subscriber}s accessing $S_r$ for specific blocks.

\subsubsection{\textsf{Permute Re-encryption Function}} 

\textsf{PerReFunc} permutes and re-encrypts all video blocks in $S_r$ following the meta structure $\textsf{MetaStash}_r$ in the \textit{Enclave} of the \textit{CP}.

Since all blocks in $S_r$ will be permuted and re-encrypted in the same way, we leverage an example from Figure~\ref{fig:per_re_func} to demonstrate the state of $S_r$ before and after the execution for one block.
In this example, the block was encrypted with counter $c$ and will be relocated from address $s$ to $s^\prime$ and re-encrypted with the counter $c'$: 

\noindent\textbf{On the \textit{Enclave}}:
	\begin{enumerate}[leftmargin=*]
		\item Selects a permutation with intermediate location, $\omega_1=s\mapsto s^*$ and $\omega_2=s^*\mapsto s'$.
		\item Computes the re-encryption mask $m=G(c)\oplus G(c')$ with PRF $G(\cdot)$ and shares it as $m_1$ and $m_2$, where $m_1\oplus m_2=m$.
	\item Send $(\omega_p,m_p)$ to $\textit{CS}_{p\in\{1,2\}}$.
	\end{enumerate}
\noindent\textbf{On $\textit{CS}_1$}:
	\begin{enumerate}[leftmargin=*]
		\item Retrieves encrypted block $b_s$ from $S_r$.
		\item Re-encrypts $b_s$ with shared key $b_s\leftarrow b_s\oplus m_1$.
		\item Uses $\omega_1$ to permute $b_s$ to $b_{s^*}$ and sends to $\textit{CS}_2$.
	\end{enumerate}
\noindent\textbf{On $\textit{CS}_2$}:
	\begin{enumerate}[leftmargin=*]
		\item Re-encrypts $b_{s^*}$ with shared key $b_{s^*}\leftarrow b_{s^*}\oplus m_2$.
		\item Uses $\omega_2$ to permute $b_{s^*}$ to $b_{s'}$ and sends to $S_r$.
	\end{enumerate}
%
After the execution, neither $\textit{ES}$s nor $\textit{CS}_{p\in\{1,2\}}$ can distinguish which blocks have been relocated to which locations. 
This ensures that the access pattern on $S_r$ is hidden after \textit{CP}/\textit{Subscriber}s' access from $S_r$.
%
Algorithm~\ref{alg:PerReFunc} in Appendix~\ref{sec:oblivious} presents \textsf{PerReFunc} in detail. 


Similarly, we also decompose the $E_r.\textsf{BatchEvict}$ operation into two steps: 1) fetch all blocks from $E_r$'s targeted paths to $S_r$; 2) evict blocks from $S_r$ to $E_r$ following the new block state.

The block fetching process for the eviction is simpler because it inherently requires to access all blocks in each targeted path.
Hence, we can consider it as the \textsf{PerReFunc} function between $E_r$ and $S_r$, which can be implemented with the following protocol:

\begin{figure}
\vspace{-7pt}
\centering
\subfloat[PathRangeRet]{
		\label{fig:path_range_ret}
		\includegraphics[width=0.43\linewidth,]{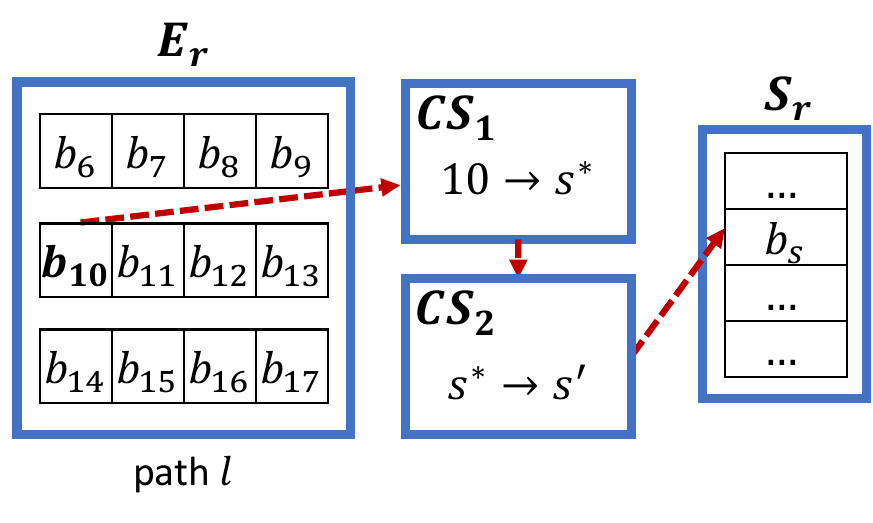}
	}
	\hfill
	\subfloat[PriRangeEvict]{
		\label{fig:pri_range_evict}
		\includegraphics[width=0.5\linewidth,]{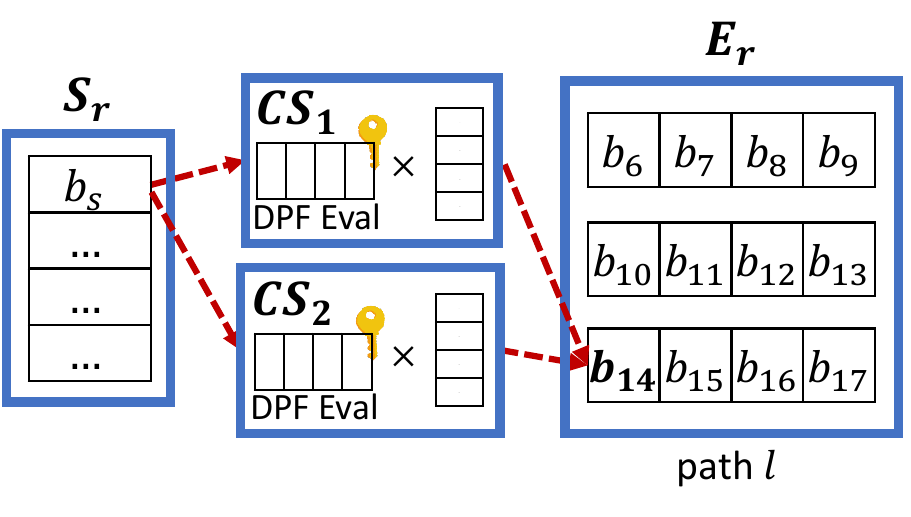}
	}
\vspace{-5pt}
\caption{Examples on how to manipulate one block under our rORAM \textsf{BatchEvict} design with distributed trust.}
\vspace{-10pt}
\label{fig:evalevict}
\end{figure}


\vspace{-5pt}
\subsubsection{\textsf{PathRangeReturn}} 
\textsf{PathRangeRet} privately routes blocks from $E_r$ to $S_r$ without revealing which specific blocks from $E_r$ are being routed to which locations in $S_r$.
During the routing, neither of $\textit{CS}_{p\in\{1,2\}}$ knows the origin or destination of the blocks.

An example of \textsf{PathRangeRet} is illustrated in Figure~\ref{fig:path_range_ret}. 
In this example, block $b_{10}$ is re-encrypted and written to $s$ in $S_r$:

\noindent\textbf{On the \textit{Enclave}}:
	\begin{enumerate}[leftmargin=*]
		\item Selects a permutation with intermediate location, $\omega_1=10\mapsto s^*$ and $\omega_2=s^*\mapsto s$.
		\item Computes the re-encryption mask $m=G(c)\oplus G(c')$ with PRF $G(\cdot)$ and shares it as $m_1$ and $m_2$, where $m_1\oplus m_2=m$.
	\item Sends $(\omega_p,m_p)$ to $\textit{CS}_{p\in\{1,2\}}$.
	\end{enumerate}
\noindent\textbf{On $\textit{CS}_1$}:
	\begin{enumerate}[leftmargin=*]
		\item Retrieves encrypted block $b_{10}$ from $E_r$.
		\item Re-encrypts $b_{10}$ with shared key $b_{10}\leftarrow b_{10}\oplus m_1$.
		\item Uses $\omega_1$ to permute $b_{10}$ to $b_{s^*}$ and sends to $\textit{CS}_2$.
	\end{enumerate}
\noindent\textbf{On $\textit{CS}_2$}:
	\begin{enumerate}[leftmargin=*]
		\item Re-encrypts $b_{s^*}$ with shared key $b_{s^*}\leftarrow b_{s^*}\oplus m_2$.
		\item Uses $\omega_2$ to permute $b_{s^*}$ to $b_{s}$ and sends to $S_r$.
	\end{enumerate}

%
%
\textsf{PathRangeRet} can also be executed in a batch setting, processing all blocks in $2^r$ ORAM paths with the same process.
Algorithm~\ref{alg:PathRangeRet} in Appendix~\ref{sec:oblivious} presents \textsf{PathRangeRet} in detail.

We then can leverage a reverse process of \textsf{BlockRangeRet} to route the blocks in $S_r$ back to $E_r$ with a similar PIR process with DPF evaluations as in the following protocol:
%

\subsubsection{\textsf{PrivateRangeEviction}}

\textsf{PriRangeEvict} routes video blocks from $S_r$ to $E_r$ without revealing which specific blocks are written to which addresses in $E_r$.
During the routing, both the origin and destination of the blocks remain hidden from each of $\textit{CS}_{p\in\{1,2\}}$.

Figure~\ref{fig:pri_range_evict} shows an example of routing the encrypted block $b_s$ to a specific address $b_{14} (s=14)$ in $E_r$, which proceeds as follows:

\begin{figure*}[!t]
\vspace{-7pt}
\begin{minipage}{\linewidth}
\begin{algorithm}[H]
	\footnotesize
	\raggedright
	\caption{OblivCDN System Operations}
	\label{oblivcdn}
	\underline{\textsc{Setup}}\\[2pt]
	\textbf{Input:} The addresses of \textit{ES}s and \textit{CS}s; The number of nodes $N$, The number of blocks for each node $Z$, and The block size $B$ for each rORAM tree; The maximum range $R$ supported by rORAM; The stash size $|S_r|, r\in[0, log_2 R]$ for each stash.\\[2pt]
	1: The \textit{CDN service provider} runs remote attestation with the \textit{CP} to attest the \textit{Enclave} code.\\[2pt]
	2: The \textit{Enclave} initialises all metadata structures. The \textsf{VideoMap} and \textsf{BlockStateMap} are initialised as Path ORAM with $N\times Z\times B$ space. The $\log_2 R + 1$ $\textsf{BlockTracker}_r$ with $N\times Z$ blocks for $bid$s is set for corresponding rORAM tree, and $|S_r|$ space for $bid$s is allocated for $\log_2 R + 1$ $\textsf{MetaStash}_r$.\\[2pt]
	3: The \textit{Enclave} establishes connections to all \textit{ES}s and initialise $\log_2 R + 1$ rORAM storage $E_r$ with $N\times Z\times B$ space and stash $S_r$ with $|S_r|$ space on \textit{ES}s.\\[2pt]
	4: The \textit{Enclave} establishes connections to all \textit{CS}s.\\[5pt]
	\underline{\textsc{Upload}}\\[2pt]
	\textbf{Input:} A video and its $vid$\\[2pt]
	1. The \textit{CDN service provider} uploads the video to the \textit{Enclave} through the secure channel established during attestation. \\[2pt]
	2. \textbf{On the \textit{Enclave}}:
	\begin{itemize}
		\item Trunks the video into $V$ blocks with size $B$.
		\item Groups those blocks into $R_V$ ranges, each range starts with $bid_i$ and contains $2^{r_i}$ consecutive $bid$s (i.e., $bid_i$, $bid_{i}+1$, ..., $bid_{i}+2^{r_i}-1$).
		\item For each range, initialises $\textsf{BlockStateMap}[\textsf{bid}_a]=(c_a,l_a,s_a)$ for $a\in[i,i+2^{r_i})$, where $c_a \leftarrow 0$, $l_i\xleftarrow{\$}[0,\frac{N}{2})$, $l_a\leftarrow (l_{a-1}+1)\mod \frac{N}{2}$, and $s_a\xleftarrow{\$}[0,|S_{r_i}|)$, ensuring that if $a\neq a'$, $s_a\neq s_{a'}$.
		\item For each range, stores $bid_a$ for $a\in[i,i+2^{r_i})$ into $\textsf{MetaStash}_{r_i}[s_a]$.
		\item For each range, follows $\textsf{BlockStateMap}[\textsf{bid}_a]$ to encrypt the block $b_a$ for $a\in[i,i+2^{r_i})$ and write into $S_{r_i}$.
		\item Stores the $\{(bid_i,r_i)\}, i\in[0, R_V)$ list into $\textsf{VideoMap}[vid]$.
		\item Invoke the \textbf{Sync} operation.\\[5pt]
	\end{itemize}
	\underline{\textsc{Fetch}}\\[2pt]
	\textbf{Input:} A video id $vid$\\[2pt]
	1. The \textit{Subscriber} sends the video id $vid$ to the \textit{Enclave}. \\[2pt]
	2. \textbf{On the \textit{Enclave}}:
	\begin{itemize}
		\item Reads $\textsf{VideoMap}[vid]$ to obtain $\{(bid_i,r_i)\}, i\in[0, R_V)$ and sends to the \textit{Subscriber}.
		\item For each range, reads $\textsf{BlockStateMap}[\textsf{bid}_a]=(c_a,l_a,s_a)$ for $a\in[i,i+2^{r_i})$ and sends to the \textit{Subscriber}.
	\end{itemize}
	3. \textbf{On the \textit{Subscriber}'s Client}:
	\begin{itemize}
		\item For each range, refers to $\textsf{BlockStateMap}[\textsf{bid}_a]=(c_a,l_a,s_a)$ for $a\in[i,i+2^{r_i})$ to compute the location of $b_a$, i.e., if $s_a<|S_{r_i}|$, then $b_a$ is on the position $s_a$ of $S_r$ ($l_a=0,e_a=0,o_a=s_a$); otherwise, $b_a$ is at the level $e_a=[\left(s_a-|S_{r_i}|\right)/Z] +1$ of the path $l_a$ of $E_r$, and its offset $o_a$ in the level is $\left(s - |S_{r_i}|\right)~\%~Z$.
		\item For each range, directly downloads $b_a$ with address $(E_{r_i}, S_{r_i}, l_a, e_a, o_a)$ for $a\in[i,i+2^{r_i})$ and decrypts $b_a$ with $c_a$.
		\item Notifies \textit{ES}s once download finished.
	\end{itemize}
	4. \textbf{On the \textit{Enclave} (Run simultaneously with Step 3)}:
	\begin{itemize}
		\item Refers to $\{(bid_i,r_i)\}, i\in[0, R_V)$ to run $\textsf{MetaData}_{r_i}.\textsf{ReadRange}(bid_i)$.
		\item For each $bid$ loaded into $\textsf{MetaStash}_{r_i}$, refers to $\textsf{BlockStateMap}[\textsf{bid}]=(c,l,s)$ to run \textsf{BlockRangeRet} between $E_{r_i}$ and $S_{r_i}$, and update $(c,l,s)$.
		\item For each range, permutes the $\textsf{MetaStash}_{r_i}$.
		\item Follow the permutation to instructs \textit{CS}s to run \textsf{PerReFunc} on $S_{r_i}$ and update $(c,s)$ for each $bid$ in $\textsf{MetaStash}_{r_i}$.
		\item Invoke the \textbf{Sync} operation.
	\end{itemize}
	4. \textbf{On the \textit{ES}s}:
	\begin{itemize}
		\item Replace $E_r$ and $S_r$ with shuffled version after sufficient number of \textit{Subscriber}s finishes download or same content is requested.\\[5pt]
	\end{itemize}
	\underline{\textsc{Sync}}\\[2pt]
	1. \textbf{On the \textit{Enclave}}:
	\begin{itemize}
		\item For each range $r\in[0,\log_2 R]$, check whether $2^r$ paths have been accessed on $\textsf{BlockTracker}_{r}$, skip the following steps if the condition is not met.
		\item Runs $\textsf{MetaData}_r.\textsf{BatchEvict}(2^r)$, which firstly loads all $bid$s in $2^r$ accessed paths into $\textsf{MetaStash}_{r}$.
		\item For each $bid$ loaded into $\textsf{MetaStash}_{r}$, refers to $\textsf{BlockStateMap}[\textsf{bid}]=(c,l,s)$ to run \textsf{PathRangeRet} between $E_{r}$ and $S_{r}$, and update $(c,s)$.
		\item Refers to $\textsf{BlockStateMap}[\textsf{bid}]=(c,l,s)$ to continually run $\textsf{MetaData}_r.\textsf{BatchEvict}(2^r)$, which writes $2^r$ accessed path of $\textsf{BlockTracker}_{r}$ with $bid$s from $\textsf{MetaStash}_{r}$.
		\item Follow the eviction process to instructs \textit{CS}s to run \textsf{PriRangeEvict} to evict blocks from $S_r$ to $E_r$ and update $(c,s)$ for each valid $bid$ evicted from $\textsf{MetaStash}_r$.
	\end{itemize}
\end{algorithm}
\end{minipage}
\vspace{-10pt}
\end{figure*}

\noindent\textbf{On the \textit{Enclave}}:
	\begin{enumerate}[leftmargin=*]
	\item Generates ($k_{1},k_{2}$) with $\textsf{DPF}.\textsf{Gen}(1^\lambda,s)$.
	\item Computes the re-encryption mask $m=G(c)\oplus G(c')$ with PRF $G(\cdot)$ and shares it as $m_1$ and $m_2$, where $m_1\oplus m_2=m$.
	\item Sends $(k_{p}, m_p, l, s=14)$ to $\textit{CS}_{p\in\{1,2\}}$.
	\end{enumerate}
\noindent\textbf{On each $\textit{CS}_{p\in\{1,2\}}$}:
	\begin{enumerate}[leftmargin=*]
	\item Executes $y[i]\leftarrow\textsf{DPF}.\textsf{Eval}(k_{p},i)$ for all block ID $s\in[0,|S_r|)$ in the stash.
	\item Retrieves the entire stash from $S_r$ into vector $b$.
	\item Computes shared block $b_{s}$ with $b_p\leftarrow\oplus_{i=0}^{|S_r|-1}(y[i]\times b[i])$.
	\item Re-encrypts the share with shared key $b_p\leftarrow b_p\oplus m_p$.
	\item Jointly writes $b_1\oplus b_2=b_{s}\oplus G(c)\oplus G(c')$ into the $b_{14}$ position on the path $l$ of $E_r$, so the original mask $G(c)$ on $b_{s}$ will be canceled and replaced with $ G(c')$.
	\end{enumerate}
With $E_r.\textsf{BatchEvict}$, all blocks in the targeted $2^r$ paths will be rewritten in batch to ensure the obliviousness of evictions.
Algorithm~\ref{alg:PriRangeEvict} in Appendix~\ref{sec:oblivious} presents \textsf{PriRangeEvict} in detail.

\noindent \textbf{Remarks}. In the above design, we consider that the \textit{CP}/\textit{Subscriber}s need to wait for \textit{CS}s to route blocks from $E_r$ to $S_r$ before accessing those blocks, and thus incurs an extra delay for the process.
Nonetheless, we realise that the CDN service is a typical ``read-only'' service whose content will never be changed but only shuffled after the initial upload to the cloud. 
This means we can follow~\cite{tople2019pro} to perform reads directly on the given address and use another background thread to obliviously route the blocks into a new memory space.
By the end of each range access, we replace the current rORAM tree $E_r$ with the new one and thus waive the waiting time for each access.
Such an approach can be extended to support multi-users by combining distinct video access into one batch. 
%

At the end of this section, we present the full system operations of \textsf{OblivCDN} in Algorithm~~\ref{oblivcdn}.
Note that all vector accesses (i.e., access stashes for \textsf{VideoMap} and \textsf{BlockStateMap}) inside the \textit{Enclave} are done by the oblivious operations in Appendix~\ref{subsec:oblivEnclave}.

\subsection{Theoretical Complexity}
\label{sec:analysis}

\label{subsec:analysis_performance}

\textbf{Intercontinental Communication Cost}. 
The intercontinental Communication costs in \textsf{OblivCDN} can be divided into 1) the cost of uploading videos, and 2) the cost of sending oblivious block routing instructions to \textit{CS}s.
As in existing CDN systems~\cite{NetflixAWS,NetflixOCA}, uploading a video is a one-time cost and linear to the video size. 
Moreover, the cost of sending instructions can be depicted as:
\begin{itemize}[leftmargin=*]    
    \item In \textsf{BlockRangeRet}, communication with each $\textit{CS}_{p \in \{1,2\}}$ is dominant by the DPF key size $\log_2 (\log_2 N \times Z + |S_r|)\times\lambda$ bits ($\mathcal{O}(r)$). 
    \item In \textsf{PerReFunc}, communication with each $\textit{CS}_{p \in \{1,2\}}$ is bounded by $|S_r|= 2^r \times v$ ($\mathcal{O}(2^r)$), where $v$ is the fixed-size stash factor~\cite{Chakraborti191}. 
    
    \item In \textsf{PathRangeRet}, communication with each $\textit{CS}_{p \in \{1,2\}}$ is bounded by the number of blocks in a path, which is $\log_2 N \times Z$ ($\mathcal{O}(\log_2 N)$).

    \item In \textsf{PriRangeEvict}, communication with each $\textit{CS}_{p \in \{1,2\}}$ is dominant by the DPF key size $\log_2 (\log_2 N \times Z + |S_r|)\times\lambda$ bits ($\mathcal{O}(r)$). 
\end{itemize}

As a result, the entire intercontinental communication cost (considering each rORAM operation manipulates $2^r$ paths/blocks in batch) for instructing rORAM operations is $\mathcal{O}(r\times2^r)$ (for \textsf{ReadRange}) and $\mathcal{O}(r\times2^r\times\log_2N)$ (for \textsf{BatchEvict}) regardless of the block size/video size.
In the real-world CDN, the video size is always much larger than the instruction size because a smaller $r$ (e.g., $r=8$) is already enough for most cases.
Hence, compared to loading the video back to the \textit{CP} in the Strawman design, \textsf{OblivCDN} significantly reduces the intercontinental communication cost without compromising security guarantees.
We will demonstrate the benefit of our design with real-world video access scenarios in \S\ref{sec:experiment}.

\noindent\textbf{Local Region Communication Cost}. 
In the local region, the communication costs can be divided into 1) the cost of fetching videos, and 2) the cost of routing video blocks between \textit{ES}s and \textit{CS}s.
Since we follow the read-only ORAM design~\cite{tople2019pro} to provide the video service, the communication cost for \textit{Subscriber}s' fetches remains unaltered compared to existing CDN systems~\cite{NetflixAWS,NetflixOCA}.

Regarding the communications between \textit{ES}s and \textit{CS}s, in the \textbf{Fetch} operation, $\textit{CS}_{p \in \{1,2\}}$ firstly acquires $2^r\times\log_2N\times B$ bytes from \textit{ES}s and writes $2^r\times B$ bytes back in \textsf{BlockRangeRet}, and then reads/writes $2^r\times v\times B$ bytes in \textsf{PerReFunc}.
In the \textbf{Sync} operation, $\textit{CS}_{p \in \{1,2\}}$ firstly reads/writes $2^r\times\log_2N\times Z\times B$ bytes with \textit{ES}s in \textsf{PathRangeRet}.
Then, it reads $2^r\times v\times B$ bytes from \textit{ES}s and tries to fill \textit{ES}s with $E_r$ by $2^r\times\log_2N\times Z\times B$ bytes via \textsf{PriRangeEvict}.

The above analysis shows that our design maintains a comparable bandwidth cost ($\mathcal{O}(2^r\times v)$) against the original rORAM($\mathcal{O}(2^r\times\log_2N)$)~\cite{Chakraborti191} while ensuring the obliviousness of all operations.

\noindent\textbf{Disk Seek on $\textit{ES}$s}. 
In both \textbf{Fetch} and \textbf{Sync} operations, each $E_r$ still takes $\mathcal{O}(\log_2N)$ seeks to fetch $2^r$ blocks and evict $2^r$ paths.

\noindent\textbf{Computation Cost}. 
The computation cost in the \textsf{OblivCDN} is minimised, as it only requires evaluating block cipher, permutation, and exclusive-or operations among all components.
To be specific, there are $2^r$ DPF keys and re-encryption masks to be generated and evaluated under \textsf{BlockRangeRet}, and $2^r\times v$ permutations and re-encryption masks for \textsf{PerReFunc} during \textsf{ReadRange}.
Also, $2^r\times \log_2N\times Z$ DPF keys/permutations and re-encryption masks are generated and used for \textsf{PathRangeRet} and \textsf{PriRangeEvict}.%

\subsection{Security Guarantees}
\label{subsec:analysis_security}
As mentioned in \S\ref{subsec:threatmodel}, \textsf{OblivCDN} follows real-world CDN systems~\cite{NetflixAWS,NetflixOCA} to consider four parties, i.e., the \textit{CDN service provider}, the \textit{cloud service providers}, the \textit{ISP partners}, and the \textit{Subscriber}.
\textsf{OblivCDN} guarantees that the content and the access pattern of all same-length videos will be indistinguishable against: the \textit{CDN service provider}s who host the \textit{CP} and one $\textit{CS}$, and the \textit{ISP partners} who host \textit{ES}s and potentially also control one $\textit{CS}$.
Due to the page limitation, we provide a detailed security analysis in Appendix~\ref{app:security_analysis}.

\section{Experiment and Evaluation}
\label{sec:experiment}
\subsection{Setup}
\noindent \textbf{Implementation}. 
We developed the \textit{CP}  and \textit{CS}s using C++, while the \textit{ES}s were implemented in Java.
The \textit{ES}s serve as  storage daemons for  memory/hard drive I/O access as requested.
To facilitate communication between these entities, we employed Apache Thrift.
For the implementation of the \textit{Enclave} within the \textit{CP} and DPF scheme, 
we utilised Intel SGX and Intrinsics, respectively.
For cryptographic primitives, such as PRF, PRG, and AES, we used the OpenSSL.
Overall, the prototype comprises $48,272$ LOC\footnote{\url{https://github.com/MonashCybersecurityLab/OblivCDN}}.
%
%

%

\noindent \textbf{Platform}. We used three devices to deploy \textsf{OblivCDN}. 
A server equipped with an Intel Xeon E2288G 3.70GHz CPU (8 cores) and 128GB RAM serves as the \textit{ES}s.
One SGX-enabled workstation featuring an Intel Core i7-8850H 2.60GHz CPU (6 cores) and 32GB RAM functions as the \textit{CP}. 
Another workstation with the same configuration is used to run the \textit{CS}s.
Note that, by design, the \textit{CS}s should be deployed to two non-colluding servers.
However, for evaluation purposes, we deployed them on the same workstation to simulate a scenario in which these two servers are in close proximity to the \textit{ES}s.
All servers are connected to a Gigabit network.

\begin{figure}[!t]
	\vspace{-7pt}
	\centering
	\subfloat[BlockRangeRet]{
		\label{fig:pri_range_ret_comp}
		\includegraphics[width=0.43\linewidth,height=2.6cm]{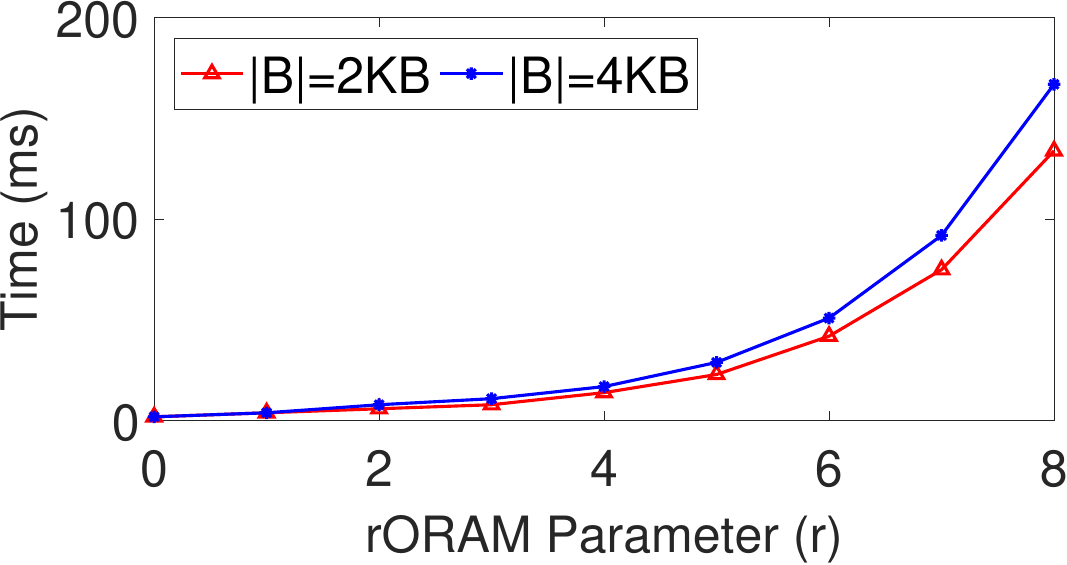}}
	\hfill
	\subfloat[PerReFunc]{
		\label{fig:per_re_func_comp}
		\includegraphics[width=0.43\linewidth,height=2.6cm]{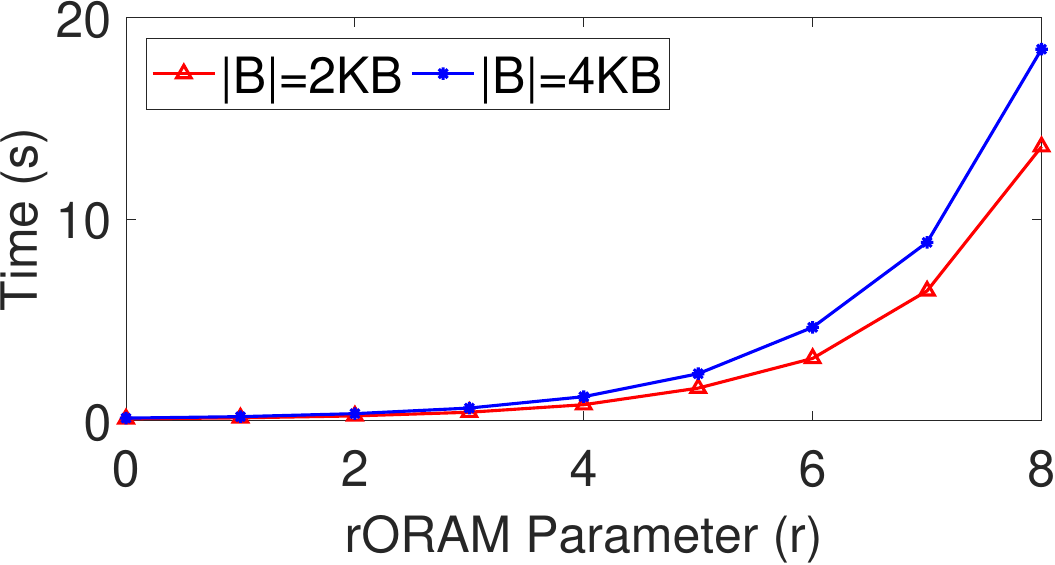}}
	\vspace{-5pt}
	\subfloat[PathRangeRet]{
		\label{fig:sec_range_ret_comp}
		\includegraphics[width=0.43\linewidth,height= 2.6cm]{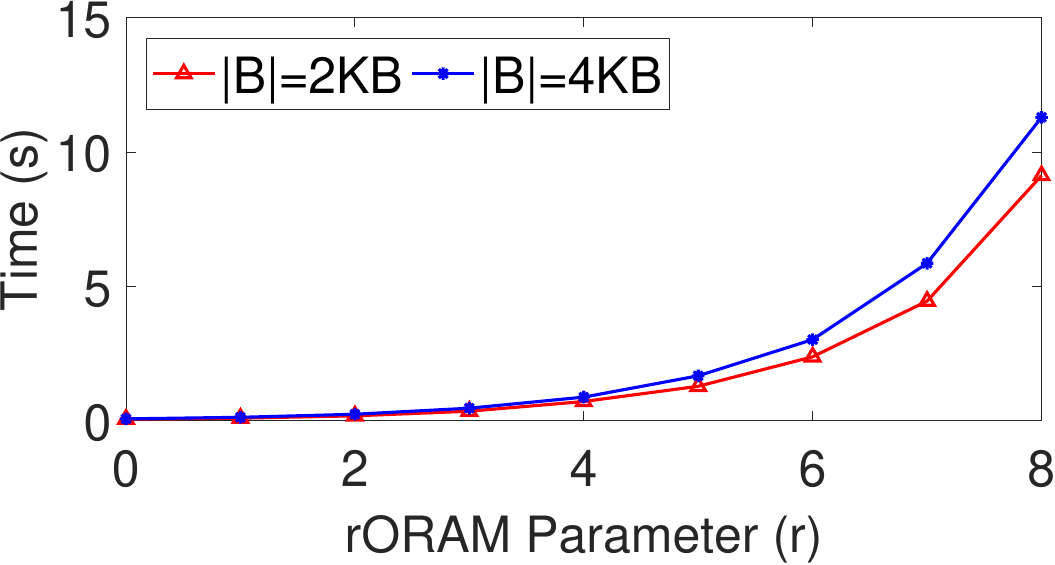}}
	\hfill
	\subfloat[PriRangeEvict]{
		\label{fig:pri_range_evict_comp}
		\includegraphics[width=0.43\linewidth,height=2.6cm]{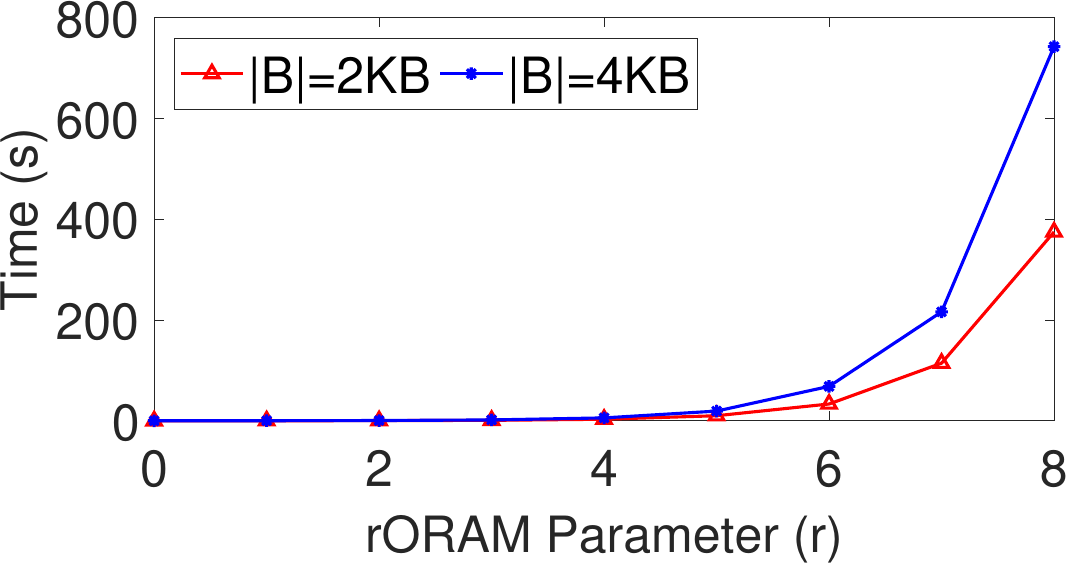}}
    \vspace{-5pt}
	\caption{Computational cost of oblivious operations.}
    \vspace{-10pt}
	\label{fig:primitive_comp}
\end{figure}

%

\subsection{Benchmark Oblivious Building Blocks}

%
%


\noindent \textbf{Benchmark Configuration}.
We first evaluate four oblivious building blocks in \textsf{OblivCDN} (\S\ref{subsec:cp}) with three metrics: computational, local I/O, and communication costs.
For this benchmark, we use $2^8$ 5-MB videos, which can be stored by an ORAM tree with 65536 leaf nodes with 5 blocks, i.e., $\log_2N=17$ and $Z=5$.
To enable range access, we deployed nine \textit{ES}s, with each \textit{ES} dedicated to supporting both an $E_r$ and $S_r$, where $r\in [0,8]$.
We adapted rORAM disk layout~\cite{Chakraborti191} to create SSD storage optimised for oblivious range access.
Consequently, each $E_r$ has $12$ GB and $23$ GB storage capacity when using a block size of $|B|=2$ and $4$ KB, respectively.
To achieve a negligible false positive rate of $2^{-128}$, we configured the fixed-size stash factor $v=105$, so the stash size $|S_r|=105 \times 2^r$ blocks~\cite{Chakraborti191}. 
In the \textit{CP}, the \textit{Enclave} only requires $21.05$ MB to store the metadata for $2^8$ videos (with around $2^{19}$ blocks to track).

\noindent \textbf{Overall Runtime Cost}. 
Figure~\ref{fig:primitive_comp} presents the overall cost of running four proposed oblivious building blocks.
%
%
As shown in the figure, most of our oblivious operations can be completed within $20$ s, except for \textsf{PriRangeEvict} (Figure~\ref{fig:pri_range_evict_comp}).
This is because \textsf{PriRangeEvict} involves evicting $2^r\times \log_2 N \times Z$ blocks.
Moreover, for each block eviction, both $\textit{CS}_{p\in \{1,2\}}$ need to scan the entire stash.
In contrast, the operations of the other building blocks work on fewer blocks, i.e., $2^r \times \log_2N$ blocks for \textsf{BlockRangeRet}, and have fewer operations on blocks, i.e., each scan can permute and re-encrypt $\log_2 N \times Z$ blocks in \textsf{PathRangeRet}
and $|S_r|$ blocks in \textsf{PerReFunc}.
We note that these building blocks are executed during the video upload or in the background process.
Our later evaluation in \S\ref{subsec:intercontinental} will show that they do not impact the \textit{Subscriber}'s experience when fetching the video from the edge.

\begin{figure}[!t]
	\centering
 	\vspace{-7pt}
	\subfloat[BlockRangeRet]{
		\label{fig:pri_range_ret_io}
		\includegraphics[width=0.43\linewidth,height= 2.6cm]{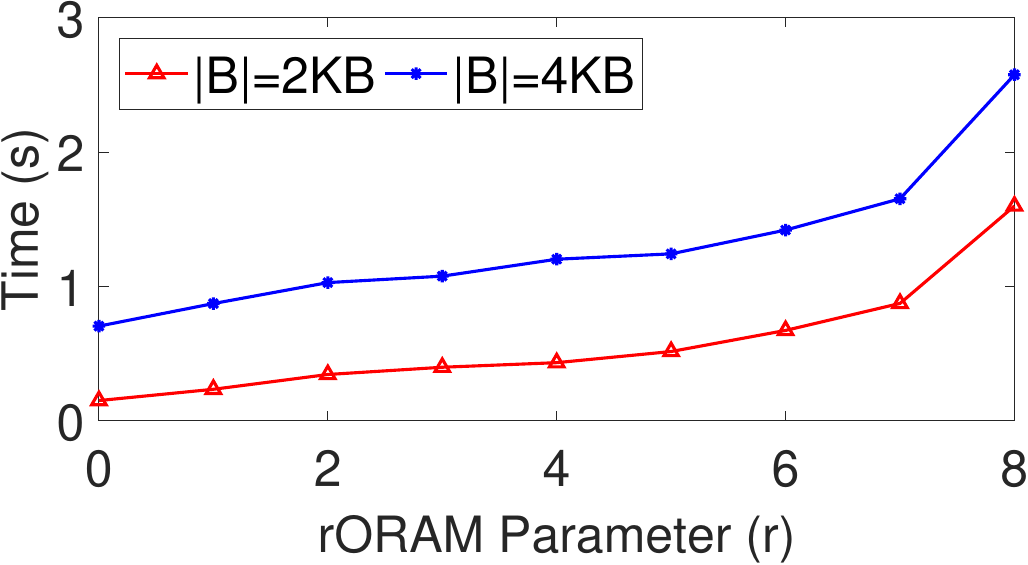}
	}
	\hfill
	\subfloat[PerReFunc]{
		\label{fig:per_re_func_io}
		\includegraphics[width=0.43\linewidth,height= 2.6cm]{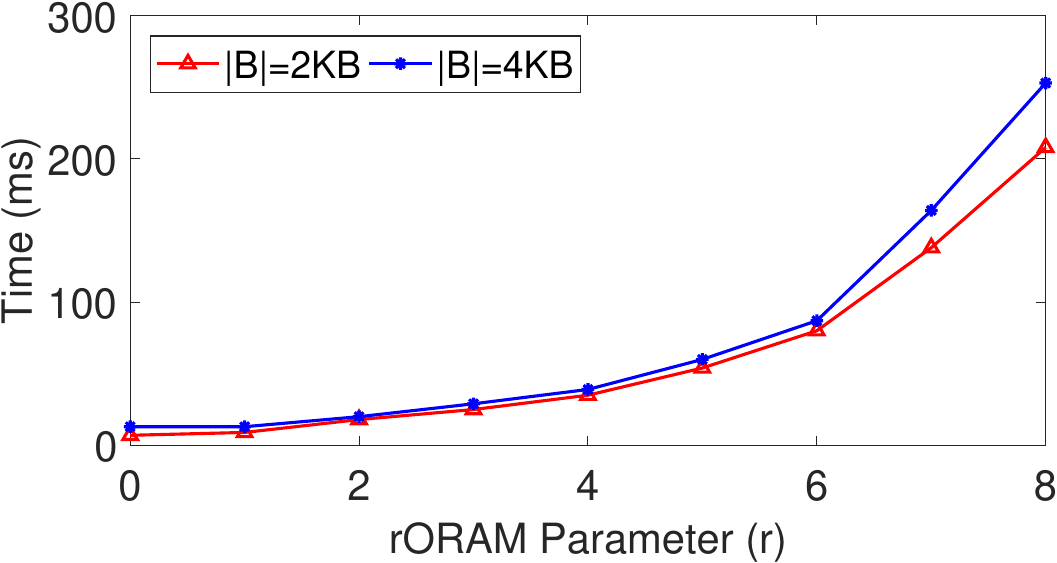}
	}
	\vspace{-5pt}
	\subfloat[PathRangeRet]{
		\label{fig:sec_range_ret_io}
		\includegraphics[width=0.43\linewidth,height= 2.6cm]{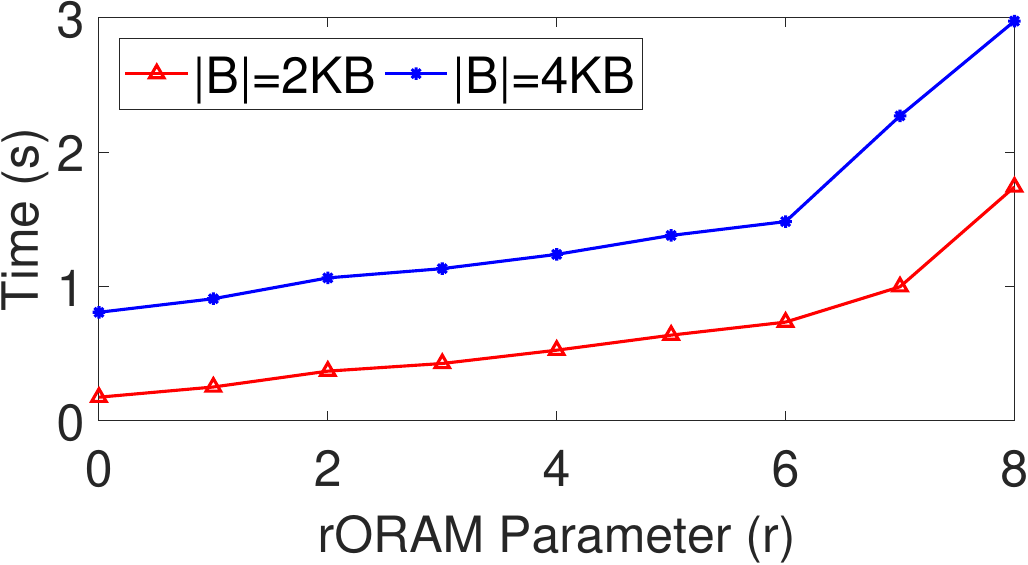}
	}
	\hfill
	\subfloat[PriRangeEvict]{
		\label{fig:pri_range_evict_io}
		\includegraphics[width=0.43\linewidth,height= 2.6cm]{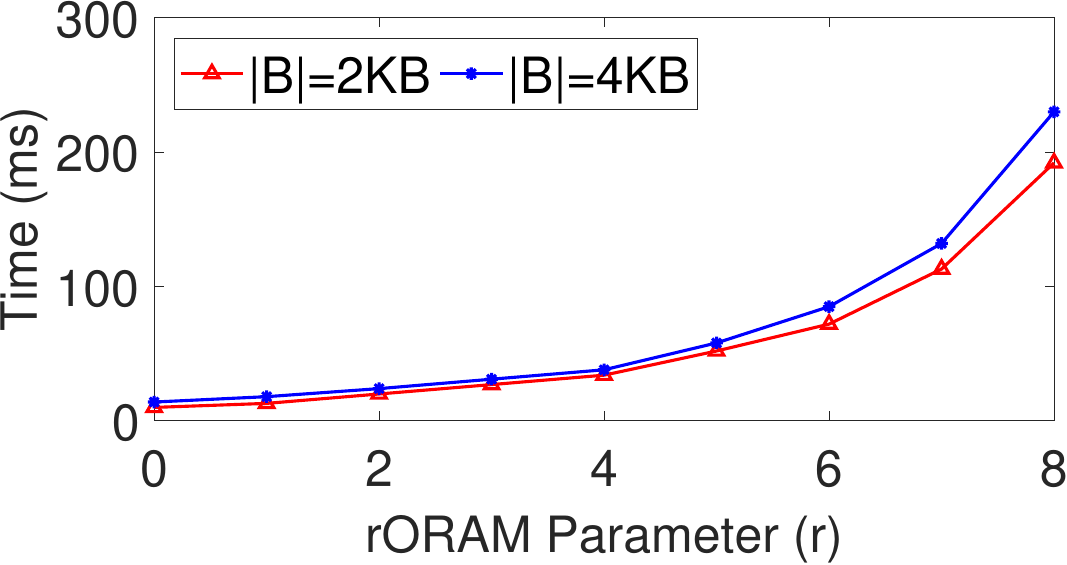}
	}
    \vspace{-5pt}
	\caption{I/O cost of oblivious operations in the \textit{ES}s.}
    \vspace{-10pt}
	\label{fig:primitive_io} 
\end{figure}

\noindent \textbf{Local I/O Cost}. 
We further dig into the I/O costs incurred by the $ES$s during these oblivious operations.
Figure~\ref{fig:primitive_io} shows that the I/O cost is almost negligible ($<3$ s) among the overall cost.
%

When looking into each operation, we can see that the disk seek is still the dominant cost for the I/O.
In specific, even \textsf{BlockRangeRet} read less blocks than \textsf{PathRangeRet}, i.e., $2^r \times \log_2N$ blocks v.s. $2^r \times \log_2N\times Z$ blocks, they incur the same I/O time because they still perform sequential seeks to all $2^r \times \log_2N\times Z$ rORAM tree nodes.
%
Due to the physical feature of the hard drive, the time for writing the block is faster than reading.
Hence, as shown in Figure~\ref{fig:pri_range_evict_io}, the cost of writing $2^r$ paths ($2^r \times \log_2N\times Z$ blocks) in \textsf{PriRangeEvict} is much faster than the reading overhead of the aforementioned two building blocks.
Finally, the cost of \textsf{PerReFunc} is the smallest among all building blocks (see Figure~\ref{fig:per_re_func_io}).
This is because \textsf{PerReFunc} only operates on $S_r$, which resides in volatile memory.

%

\noindent \textbf{Network Communication Cost}. 
Figure~\ref{fig:ovliv_op_comm} shows the network communication cost with the four proposed oblivious operations when $r=8$.
It is consistent with our theoretical analysis that the \textit{CP} incurs a constant and moderate ($<8$ MB) inter-continental communication cost regardless of the block size and video size as it only sends instructions to the \textit{CS}s.
%

Regarding the local region communication, we can observe a symmetric communication cost for \textsf{PerReFunc} and \textsf{PathRangeRet} since they leverage \textit{CS}s to retrieve blocks from rORAM/stashes and re-route them to the stashes.
For the other two operations (\textsf{BlockRangeRet} and \textsf{PriRangeEvict}), the traffic is re-shaped by DPF evaluations:
In \textsf{BlockRangeRet}, both \textit{CS}s receive $2^r \times \log_2N$ blocks from $E_r$ but filter out to $2^r$ blocks that \textit{CS}s intend to reroute from the ORAM tree to the stash.
Our result shows that \textit{CS}s only send $1$ MB (for $2$ KB blocks) or $2$ MB (for $4$ KB blocks) to the stash even though they fetch $16.1$ or $32.2$ MB from the ORAM tree.

Similarly, in \textsf{BlockRangeRet}, both \textit{CS}s receive blocks from the entire stash ($2^r\times v$ blocks) but only send a share of the targeted blocks ($2^r\times \log_2 N \times Z$ blocks) to be evicted to the ORAM tree.
Since $v$ is always larger than $\log_2 N \times Z$ to support batch fetch, this process retrieves $106$ (for $2$ KB blocks) and $211.4$ MB (for $4$ KB blocks) from the stash but only sends $50.3$ and $100.2$ MB to the ORAM tree.
We further detail the breakdown cost under different rORAM parameters ($r$) in Appendix~\ref{app_sec:buildingblock}. 
%

\subsection{Inter-continental Streaming Simulation}\label{subsec:intercontinental}

\textbf{Streaming Setting}. We replicate a real-world scenario where a \textit{Subscriber} in Australia requests a video from a \textit{CP} located in the U.S.
The service provider has deployed \textit{CS}s and \textit{ES}s in Australia according to \textsf{OblivCDN} design.
To simulate the intercontinental communication, we refer to real-world intercontinental delays~\cite{matt2020aws} between cloud regions to introduce a $400$ ms round-trip delay between \textit{CP}\&\textit{CS}s and \textit{CP}\&\textit{ES}s. 
Additionally, we put a $1$ ms round-trip delay between the \textit{Subscriber} and \textit{ES}s to simulate a tiny cost incurred in the local zone~\cite{matt2020aws}.
Under the network setting, we test the streaming performance of three types of videos as listed in Table~\ref{tlb:video_info} and compare \textsf{OblivCDN} with the Strawman design (\S\ref{subsec:strawman}) and OblivP2P~\cite{Jia16OblivP2P}\footnote{While OblivP2P\cite{Jia16OblivP2P} and Multi-CDN~\cite{Shujie20} are compatible with real-world CDN architecture, we only pick OblivP2P for comparison, because Multi-CDN uses a different setting that relies on the collaboration between multiple \textit{CDN provider}s for privacy.}.
Note that with a larger block size (256 KB and 512 KB), the rORAM with the same parameter, i.e., $\log_2N=17$, $Z=5$, and $r\in[0,8]$, can store 108 GB and 207 GB of CDN content in the \textit{ES}s, respectively.
Hence, our simulation can be a good illustration of the performance of streaming $180-345$ one-hour SD (720p) videos through the Disney+ app on a mobile device \textit{Subscriber}~\cite{Disney21}.

\begin{figure}[!t]
	\centering
	\vspace{-7pt}
	\subfloat[Inter-continental]{
		\label{fig:inter_comm}
		\includegraphics[width=0.45\linewidth,height= 2.7cm]{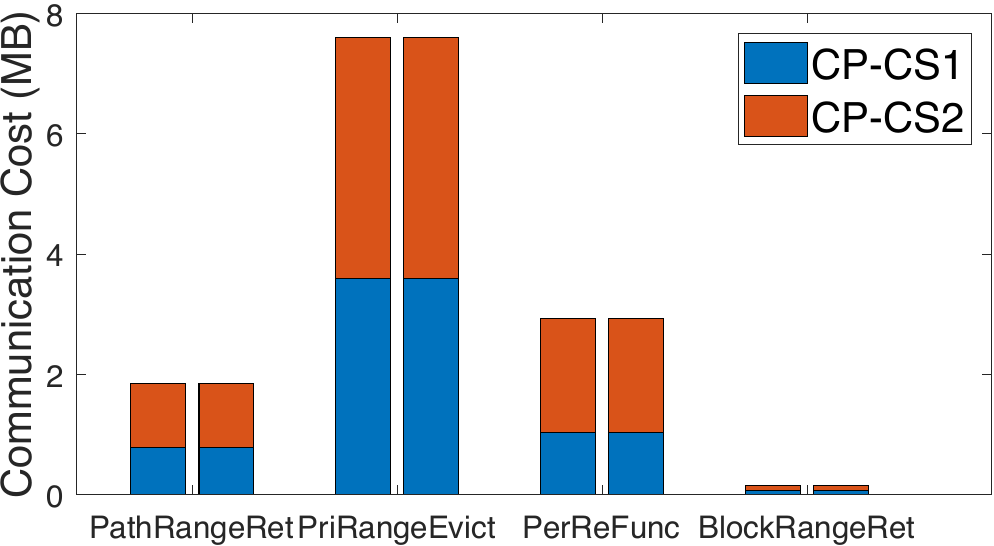}
	}
	\hfill
	\subfloat[Local region]{
		\label{fig:intra_comm}
		\includegraphics[width=0.45\linewidth,height= 2.7cm]{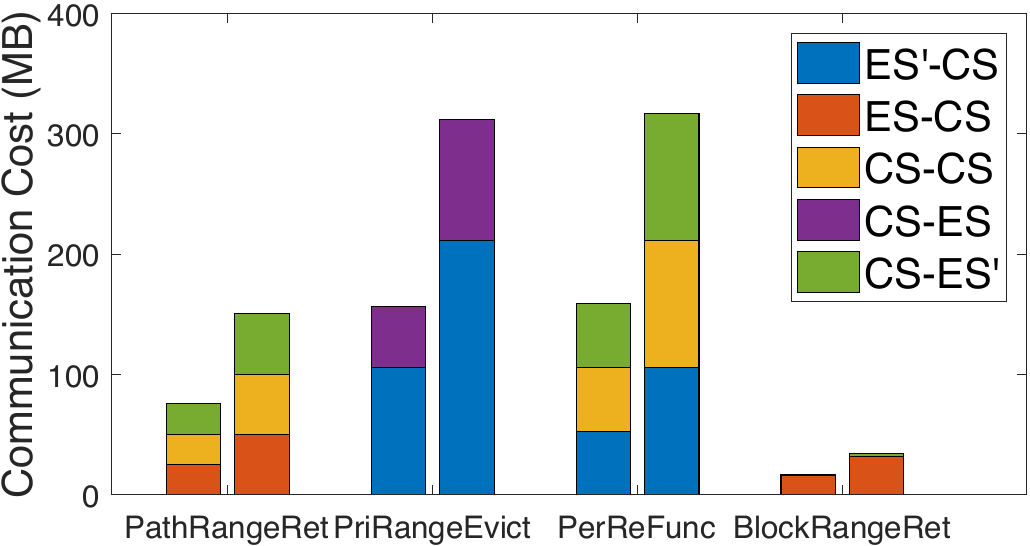}
	}
	\vspace{-5pt}
	\caption{The communication cost (MB) of oblivious operations when $r=8$. For each operation, the left bar is for $|B|=2$ KB, and the right one is for $|B|=4$ KB.}
	 \vspace{-5pt}
    \label{fig:ovliv_op_comm}
\end{figure}

\begin{table}[!t]
\small
	\centering
	\caption{Video specifications used in the evaluation.}
	\vspace{-5pt}
	\begin{tabular}{|c|c|c|c|}
	\hline
		\textbf{Video Size} & $2$ MB & $127$ MB & $256$ MB  \\
	\hline
		\textbf{Block Size} & $4$ KB & $256$ KB & $512$ KB  \\
	\hline
		\textbf{Length/Quality} & $30$ s/360p & $5$ mins/720p & $5$ mins/1080p  \\
	\hline
	\end{tabular}
 	\vspace{-5pt}
	\label{tlb:video_info}
\end{table}

\begin{table}[!t]
\small
	\centering
	\caption{The upload latency (s) of test videos under the Strawman and our approach, $(\times)$: improvement.}
	\vspace{-5pt}
	\begin{tabular}{|c|c|c|c|}
	\hline
		Size & $2$ MB & $127$ MB & $256$ MB  \\
	\hline
		Strawman & 3403 ($109\times$) & 8356 ($203\times$) & 13403 ($263\times$) \\
	\hline
		\textbf{\textsf{OblivCDN}} & \textbf{31} & \textbf{41} & \textbf{51} \\
	\hline
	\end{tabular}
	\vspace{-5pt}
	\label{tlb:video_upload_latency}
\end{table}

\noindent \textbf{Video Upload Cost}. 
We start with the upload cost comparison between the Strawman and \textsf{OblivCDN}.
Note that OblivP2P is a peer-to-peer communication system with no upload phase, so we omit the comparison with OblivP2P in this part.

During the video upload phase, the \textit{CP} should upload a video to the targeting \textit{ES}s.
In the Strawman, \textit{CP} conducts this task via Path ORAM operations, which involves heavy intercontinental data transfer.
As shown in Table~\ref{tlb:video_upload_latency}, the Strawman approach takes $\sim$1 hour to upload a $2$ MB video and almost 4 hours to upload a $256$ MB one.
\textsf{OblivCDN} significantly accelerates the process by offloading the ORAM operations to \textit{ES}s and \textit{CS}s.
The \textit{CP} only needs to upload all video blocks into $S_r$ and wait for the \textbf{Sync} operation to be done.
Our result shows that the upload latency for \textsf{OblivCDN} is at least $109\times$ faster than the Strawman, as it only requires $31$ s to upload the $2$ MB video and $51$ s for the $256$ MB one.

\begin{table}[!t]
\small
	\centering
	\caption{The upload communication cost (MB) of test videos under the Strawman and our approach, $(\times)$: improvement.}
	\vspace{-5pt}
	\begin{tabular}{|c|c|c|c|}
	\hline
		Size & $2$ MB & $127$ MB & $256$ MB  \\
	\hline
		Strawman & 159.6 ($7.2\times$) & 10220 ($69\times$) & 20440 ($74\times$) \\
	\hline
		\textbf{\textsf{OblivCDN}}  & \textbf{22} & \textbf{148}  &  \textbf{276}  \\
	\hline
	\end{tabular}
	\vspace{-10pt}
	\label{tlb:video_upload_comm}
\end{table}

Moreover, the Strawman design experiences communication blow-ups for large videos as it requires fetching the entire ORAM path to write one video block.
Table~\ref{tlb:video_upload_latency} shows that it requires $80\times$ more communication cost to upload the original video to the edge.
In contrast, \textsf{OblivCDN} leverages the batch read/evict strategy to save bandwidth.
The communication cost is reduced by $7.2\times$ at least, and for the $256$ MB video, it only sends $276$ MB data to the edge, which is $74\times$ smaller than that in the Strawman.

\begin{table}[!t]
\small
	\centering
	\caption{\textit{Subscriber}'s fetch latency (s), $(\times)$: improvement}
	\vspace{-5pt}
	\begin{tabular}{|c|c|c|c|c|c|c|}
 \hline
		\textbf{Video Size} & $2$ MB & $127$ MB & $256$ MB  \\
	\hline
		Strawman  & 457.9 ($89\times$)& 466 ($87\times$) & 504.4 ($90\times$)\\
  	\hline
		\textsf{OblivP2P}  & 16.1 ($3.1\times$)& 1016.9 ($191.8\times$)& 2050.1 ($366\times$)\\
    \hline
  	\textbf{\textsf{OblivCDN}} & \textbf{5.1} & \textbf{5.3} & \textbf{5.6} \\
	\hline
	\end{tabular}
	\vspace{-5pt}
	\label{tlb:video_subscriber_fetch}
\end{table}

\begin{table}[!t]
\small
	\centering
	\caption{The background latency (s) of fetch under the Strawman and our approach, $(\times)$: improvement.}
	\vspace{-5pt}
	\begin{tabular}{|c|c|c|c|c|c|c|}
	\hline
		\textbf{Video Size} & $2$ MB & $127$ MB & $256$ MB  \\
	\hline
		Strawman & 3403.3 ($82\times$) & 8355.9 ($201\times$) & 13403 ($322\times$) \\
	\hline
		\textbf{\textsf{OblivCDN}} & \textbf{41.5} & \textbf{41.5} & \textbf{41.5} \\
	\hline
	\end{tabular}
	\vspace{-5pt}
	\label{tlb:video_background_fetch_latency}
\end{table}

\begin{table}[!t]
\small
	\centering
	\caption{The background communication cost (MB) of fetch under the Strawman and our approach, $(\times)$: improvement.}
	\vspace{-5pt}
	\begin{tabular}{|c|c|c|c|}
	\hline
		\textbf{Video Size} & $2$ MB & $127$ MB & $256$ MB  \\
	\hline
		Strawman & 159.6 ($6.1\times$) & 10220 ($393\times$) & 20440 ($786\times$) \\
	\hline
		\textbf{\textsf{OblivCDN}} & \textbf{26} & \textbf{26} & \textbf{26} \\
	\hline
	\end{tabular}
	\vspace{-5pt}
	\label{tlb:video_background_fetch_comm}
\end{table}

\begin{table}[!t]
\small
	\centering
	\caption{Query throughput (MByte/s), $(\times)$: improvement.}
	\vspace{-5pt}
	\begin{tabular}{|c|c|c|c|c|c|c|}
 \hline
		\textbf{Sup. Peers} & $2^5$ & $2^7$ & $2^9$  \\
	\hline
  	\textsf{OblivP2P}   &0.0625 ($40.6\times$)& 0.25 ($42.2\times$)& 0.5 ($84\times$)\\
  \hline
  		\textbf{\textsf{OblivCDN}}  & \textbf{2.64} & \textbf{10.56} & \textbf{42}\\
	\hline
	\end{tabular}
	\vspace{-10pt}
	\label{tlb:video_subscriber_throughput}
\end{table}

\noindent \textbf{Video Fetch Cost}. 
For the fetch operation, we split it into two phases, i.e., the foreground download process for \textit{Subscriber}s to obtain the requested content, and the background synchronisation process for the \textit{CP} to instruct \textit{CS}s to re-route the blocks.

In the first phase, the strawman approach follows the PathORAM design to scan a path and filter the targeted block in that path.
On the other hand, \textsf{OblivCDN} follows the ``read-only'' ORAM design~\cite{tople2019pro} to allow direct downloads on targeted blocks.
As shown in Table~\ref{tlb:video_subscriber_fetch}, \textsf{OblivCDN} only requires $5.1$ s to retrieve a video under the network setting.
As a result, \textsf{OblivCDN} accelerates  $\sim90\times$ faster than the fetch process of the Strawman.
\textsf{OblivCDN} is also much more efficient than OblivP2P ($366\times$ faster) since OblivP2P requires multiple peers to jointly execute the PIR scheme for each video block retrieval.

In the latter phase, \textsf{OblivCDN} incurs only a constant cost with given parameters, which is consistent with our theoretical analysis.
In contrast, the Strawman approach imposes prohibitively high costs, as shown in Table~\ref{tlb:video_background_fetch_latency} and~\ref{tlb:video_background_fetch_comm}.
This discrepancy arises because the Strawman operates as the PathORAM client during this phase.
It needs to retrieve all data blocks from the \textit{ES} and execute the Path ORAM protocol to relocate them into new paths.
Since OblivP2P does not have the background phase, we omit the comparison here.
%


\noindent \textbf{Streaming Throughput}. 
%
In \textsf{OblivCDN}, each \textit{ES} operates independently for \textit{Subscriber}s to fetch.
Moreover, it minimises the blocking time of \textbf{Sync} operations as \textit{ES} utilises another background thread to run oblivious routing protocols with \textit{CS}s on a copy of the ORAM tree and replace the old rORAM tree with the updated one based on requests.
This means if we employ more \textit{ES}s, we can easily improve \textit{Subscriber}s' streaming throughput via multi-thread downloading.

We subsequently conduct a streaming performance comparison with OblivP2P~\cite{Jia16OblivP2P} by scaling the number of \textit{ES}s in \textsf{OblivCDN} to match the number of supporting peers in OblivP2P.
As shown in Table~\ref{tlb:video_subscriber_throughput}, we observe that \textsf{OblivCDN} improves $40.6\times-84\times$ query throughput compared to OblivP2P.
Even with $2^5$ servers, \textsf{OblivCDN} can effectively support the streaming of Ultra High Definition (4K 
UHD) videos, meeting the minimal requirement of $1.875$ MBps~\cite{NetflixRequire}.


%


%

%
%
%
%

\noindent \textbf{Adaptation to multi-user concurrent accesses}.
The current design of \textit{ES}s allows end-users to directly fetch entire encrypted video blocks.
Hence, \textsf{OblivCDN} can be integrated with multi-user support systems using untrusted Load Schedulers {\color{blue}\cite{Chakraborti19,Dauterman21}}.
Specifically, each \textit{ES} can work with a scheduler to group queries in a batch setting.
Alternatively, an ephemeral key mechanism could be employed to support concurrent access by multiple users for the same video requests.
These directions are left for future work.

\section{Related works}

\noindent \textbf{Access Pattern Hiding CDN Systems}.
Protecting content and user privacy in outsourced CDN systems is an important topic and has been widely discussed.
While encryption is a straightforward solution for data confidentiality, several prior works mainly focus on the access pattern hiding feature:
%
%
The first line of those works rely on anonymous routing~\cite{Michael14}, which could be prone to long-term traffic analysis~\cite{TaoWang14} and may have potential copyright issues~\cite{NetflixCopyright22}. 
%
%
%
%
%
%
The second line focuses on fortifying users' privacy in peer-assisted CDN networks~\cite{Jia16}, which does not consider users' privacy when fetching from edge servers.
The closest line of work aims to transfer CDN systems to oblivious ones with customised primitives, such as ORAM/PIR~\cite{Gupta16,Jia16OblivP2P} and queries with distributed trust~\cite{Shujie20}.
However, these works are not scalable in the real-world scenario with multiple edge servers, which will multiply the bandwidth cost of the current design. 
Some recent works suggest having TEE-enabled devices~\cite{Jia2017,Silva19} for oblivious operations.
These designs incur compatibility issues in CDNs with massive legacy devices. 

\noindent \textbf{Oblivious Query Systems}. ORAMs have been widely applied in encrypted systems to ensure the obliviousness of queries.
This includes searchable encryption~\cite{RaphaelBO17,GharehChamani18,Vo20}, oblivious storage and databases~\cite{Cetin16,Vincent15,Chen20,Elaine20,Mishra18,Sasy17,Hoang19,Ahmad18}, and data analytic frameworks~\cite{Zheng2017,Eskandarian19}.
We note that the adoption of ORAMs in those studies is similar to our Strawman design and will incur noticeable delays or compatibility issues under the CDN setting.

Some recent works employ the distributed trust model to enable oblivious queries.
With two servers, several applications, such as ORAM storage~\cite{Doerner17,Bunn20}, data search and analytics~\cite{Dauterman20,Dauterman22,Frank17}, private messaging~\cite{Eskandarian21}, and contact tracing~\cite{Trieu20,Dittmer20}, are efficiently supported. 
However, those systems often duplicate the database into multiple servers' storage to support its functionality.
This increases both \textit{CDN provider}s' running costs and \textit{Subscriber}s' bandwidth costs and thus makes the design not practical.

\noindent \textbf{ORAM Optimisations}.
ORAM design has been optimised in various directions to improve computation efficiency~\cite{WangORAM15,Zahur16,Doerner17,Hoang20}, bandwidth cost~\cite{Ren15,Chen19}, and query throughput~\cite{Chakraborti19,Chakraborti191,Demertzis17,tople2019pro}. 
%
However, these works often have the eviction from a local stash to an outsourced ORAM storage. 
Instead, our system brings ORAM design to a practical CDN setting, which allows the eviction between both the outsourced stash and the ORAM storage deployed in \textit{ES}s. 
%

%
%


\section{Conclusions} 

\textsf{OblivCDN} provides a novel combination of customised oblivious building blocks, range ORAM, and TEE in CDN systems to protect content confidentiality and end-user privacy from data breaches and access pattern exploitation.
\textsf{OblivCDN} seamlessly integrates with mainstream CDNs, requiring minimal modifications.
In real-world streaming evaluation, \textsf{OblivCDN} effectively meets CDN requirements: low user fetch latency, efficient computation, and minimal communication cost, even in intercontinental environments.


\section{Acknowledgements}
 
The authors would like to thank the anonymous reviewers for their valuable comments. This work was supported in part by the Australian Research Council Future Fellowship grant FT240100043 and a Data61-Monash Collaborative Research Project (D61 Challenge).

\bibliographystyle{ACM-Reference-Format}
\bibliography{sgxse}

\appendix

\section{$\textnormal{r}$ORAM Protocol}
\label{subsec:rORAM_protocols}

\begin{figure}[!t]
\centering
\includegraphics[width=0.9\linewidth]{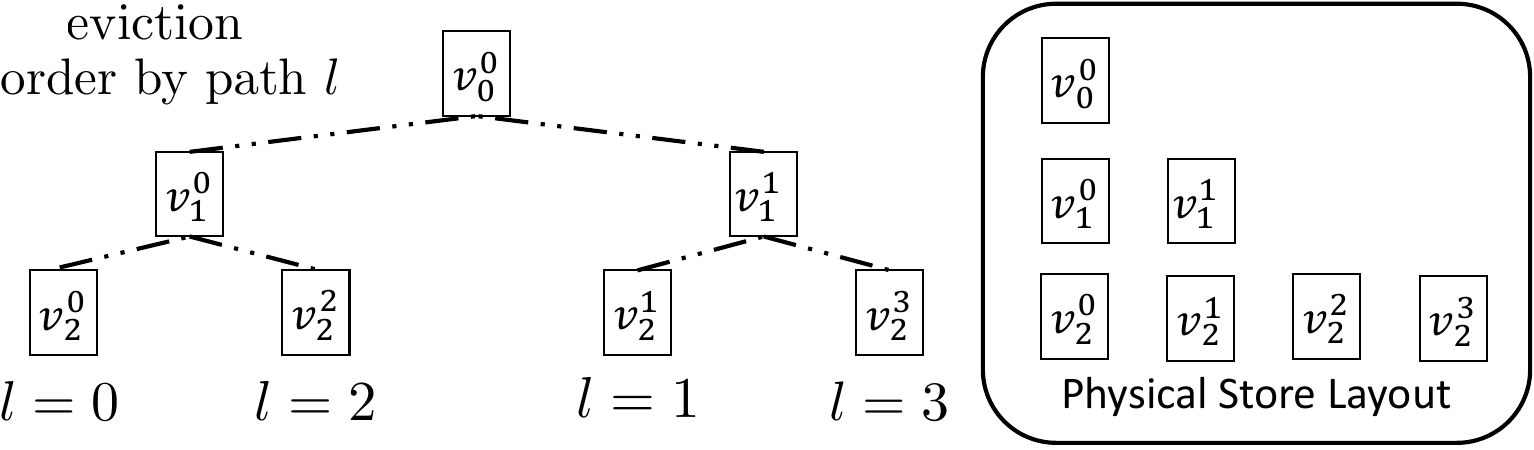}
\caption{rORAM physical and logical storage layout.}
\label{fig:rORAM_layout}
\vspace{-10pt}
\end{figure}

Let $N$ be the number of blocks in an ORAM storage with the height $L$, and $\textsf{bid}_{i \in [0,N)}$ be the identifier of a block, and $R$ be the largest sequence of \textsf{block}s allowed to fetch each time from storage. 
The rORAM scheme~\cite{Chakraborti19} consists of $\textnormal{log}_2R + 1$ ORAM storage deployed at the server side, denoted as $E_r$, where $r \in [0, \textnormal{log}_2R]$. 
The client maintains the corresponding $S_r$ and a position map $\textsf{position}_r$. 

A unique design for rORAM is that it partitions the server's disk to store each layer of the ORAM tree in consecutive physical space (see the right part of Figure~\ref{fig:rORAM_layout}) but arrange them in a bit-reversed order, which means the paths are encoded as 0(00), 2(10), 1(01), 3(11) in a three-layer ORAM tree (the left part of Figure~\ref{fig:rORAM_layout}).
The above design enables sequential accesses on consecutive paths and minimises the overlapping between consecutive paths to mitigate overflow when evicting blocks.
For instance, when accessing the path 0, 1, 2, the rORAM client can sequentially access $v_0^0,v_0^0,v_0^0$ on layer 0, $v_1^0,v_1^1,v_1^0$ on layer 1 and $v_2^0,v_2^1,v_2^2$ on layer 2 to finalise the access, which only requires $\mathcal{O}(\log_2N)$ seeks.
The eviction on those paths are following the same process but in a bottom-up order, which also indicates $\mathcal{O}(\log_2N)$ seeks.

Furthermore, rORAM realises a location-sensitive mapping strategy to assign $bid$s and paths to blocks.
In particular, to store $2^r$ blocks into $E_r$, the rORAM client assigns a random $bid$ and path $l$ to the starting block, and then assigns $bid+i$ and $(l+i) \mod 2^{(L-1)}$ for $i\in[1,2^r)$ to the following blocks.
By doing this, rORAM ensures that for blocks that belong to the same range, they are always in consecutive paths, and thus store close to each other to maintain minimised seek cost $\mathcal{O}(\log_2N)$.
For blocks that are not in the same range, the ORAM randomisation strategy is still in place to ensure the full obliviousness.
At the end, we provide the detailed rORAM protocols as in Algorithm~\ref{alg:roram_readrange} and~\ref{alg:roram_batchevict}.

\begin{figure}[!t]
\vspace{-7pt}
\begin{minipage}{\linewidth}
\begin{algorithm}[H]
\caption{$E_r.\textsf{ReadRange}(bid)$}
\label{alg:roram_readrange}
\begin{algorithmic}[1]
\small
\State Let $U :=[bid,bid+2^r)$ 
\State Scan $S_r$ for blocks in the range $U$.
\State $p \leftarrow \textsf{position}_r.\textsf{query}(bid)$ \Comment{get the ORAM path $p$ for the block $\textsf{bid}$}.
\For{$e=0, \dots, L-1$} 
\State Read ORAM buckets $V=\{v^{t ~\textsf{mod}~ 2^e}_e: t \in [p,p + 2^r) \}$
\State Identify blocks of $U$ from $V$.  
\EndFor
\State $p^\prime \leftarrow [0,N)$ \Comment{identify a new random path}
\State $\textsf{position}_r.\textsf{update}(bid+i,p^\prime + i)$ for $i \in [0,2^r)$, wrap around $N$.
\end{algorithmic}
\end{algorithm}
\end{minipage}
\vspace{-10pt}
\end{figure}




\begin{figure}[!t]
\begin{minipage}{\linewidth}
\begin{algorithm}[H]
\caption{$E_r.\textsf{BatchEvict}(k)$}
\label{alg:roram_batchevict}
\begin{algorithmic}[1]
\small
\State Let $G_r$ be an integer tracking the eviction schedule of $E_r$
\For{$e=0,\dots,L-1$} \Comment{fetch blocks on $k$ paths}
\State Read ORAM buckets $V=\{v^{t ~\textsf{mod}~ 2^e}_e: t \in [G_r,G_r + k) \}$
\State Put not-dummy blocks in $V$ to $S_r$
\EndFor
\For{$e=L-1,\dots,0$} \Comment{evict $k$ paths: bottom-up, level-by-level}
\For{$i \in \{ t ~\textsf{mod}~ 2^e : t \in \in [G_r,G_r + k) \}$} \Comment{for each path}
\State \Comment{Let $l_j$ be the current ORAM path of the $\textsf{block}_j$}
\State $S^\prime \leftarrow \{\textsf{block}_j \in S_r$ : $l_j \equiv i~ (\textsf{mod} ~2^e)\}$
\State $S^\prime \leftarrow$ Select $\textsf{min}(|S^\prime|,Z)$ \textsf{blocks} from $S^\prime$
\State $S_r \leftarrow S_r ~ / ~S^\prime$
\State \Comment{Let $v^{i~ \textsf{mod} ~2^e}_e$ be the partition $e$ with offset $(i~ \textsf{mod} ~2^e)$}
\State $v^{i~ \textsf{mod} ~2^e}_e \leftarrow S^\prime$ \Comment{note the starting offset to write $S^\prime$ at $e$}
\EndFor
\EndFor
\For{$e=0,\dots,L-1$} 
\State Write ORAM buckets $\{v^{t~ \textsf{mod} ~2^e}_e : t \in [G_r,G_r + k) \}$ to $E_r$
\EndFor
\end{algorithmic}
\end{algorithm}
\end{minipage}
\end{figure}

\section{Algorithms of Building Blocks}\label{sec:oblivious}

In this section, we present the set of building blocks used to build the system in \S\ref{sec:main}.  
In general, we denote $G$ as pseudo random function we used in this paper. 

\begin{center}
\begin{figure}[!t]
\begin{minipage}[b]{\linewidth}
\fbox{%
  \parbox{\textwidth}{
$
  \begin{array}{l|l}
    \underline{\textsf{OGet}(S,k)}: & \underline{\textsf{OUpdate}(S,k,v)}: \\
   1: v\leftarrow\bot & 1: \textbf{for}~i=0~\textbf{to} ~|S|-1~\textbf{do}\\
   2: \textbf{for}~i=0~\textbf{to} ~|S|-1~\textbf{do} & 2: \hspace{6pt}c\leftarrow\textbf{ocmp}(i,k)\\
   3: \hspace{6pt}c\leftarrow\textbf{ocmp}(i,k) & 3: \hspace{6pt}\textbf{oupdate}(c,v,S[i]) \\
   4: \hspace{6pt}\textbf{oupdate}(c,S[i],v) &  \\
   5: \textbf{return}~v &  \\
  \end{array}
$
  }%
}
\end{minipage}
\vspace{-5pt}
\caption{\textnormal{\textsf{OGet}} and \textnormal{\textsf{OUpdate}}}
\label{tab:oprimitive}

\end{figure}
\end{center}
%
%
%
%
%
%
%

\subsection{Oblivious Functions in the \textit{Enclave}}
\label{subsec:oblivEnclave}
Many techniques have been proposed to mitigate memory access side-channel attacks~\cite{Lee2017InferringFC,Brasser17,Bulck2017TellingYS,Wang2017LeakyCO} by obfuscating the access of the \textit{Enclave} to CPU's caches~\cite{Rane15,Sasy17,Ohrimenko16}. 
Figure~\ref{tab:oprimitive} demonstrates common oblivious functions, i.e., \textsf{OGet} and \textsf{OUpdate}, that enable the \textit{Enclave} to \textit{obliviously} access the vector, including the stash in \textsf{VideoMap}, \textsf{BlockStateMap},and the Strawman design). Note that \textbf{oupdate} (i.e., oblivious update) and \textbf{ocmp} (i.e., oblivious comparision) are constructed based on \textsf{CMOV} and \textsf{SETE} instructions. 

%
%
%

\begin{figure}[!t]
\vspace{-7pt}
\begin{minipage}{\linewidth}
\begin{algorithm}[H]

\caption{$\textsf{BlockRangeRet}(Q)$} 
\label{alg:BlockRangeRet}

\begin{algorithmic}[1]

\footnotesize
\Require Query List $Q=\{(\textsf{bid}_i,c_i,l_i,l'_i,s_i,s^\prime_i)\}$, where $i\in[a,a+2^r)$
\State \textbf{Inside the \textit{Enclave}:}
\State Init query batches $Q_1$ and $Q_2$
\For{$(\textsf{bid}_i,c_i,l_i,s_i,s^\prime_i) \in Q$} \Comment{generate instructions}

\State  $(k_{i,1},k_{i,2}) \leftarrow \textsf{DPF.Gen}(1^\lambda, s_i)$
\State  $m_{i,1} \xleftarrow{\$} (0,1)^\lambda$
\State $m_{i,2} \leftarrow m_{i,1}\oplus G(c_i)\oplus G(c_i+1)$
\State $o_i\leftarrow (s_i-|S_r|)\% Z$
\State $Q_1 \leftarrow Q_1 \cup \{(k_{i,1},m_{i,1},l_i,o_i,s^\prime_i)\}$ 
\State $Q_2 \leftarrow Q_2 \cup \{(k_{i,2},m_{i,2},l_i,o_i,s^\prime_i)\}$ 

\State Oblivious update $\textsf{BlockState}[\textsf{bid}_i]$ with $(c_i+1)$, $l'_i$ and $s'_i$  

\EndFor
\State \textit{Enclave} sends $Q_p$ to $CS_p$, $p\in\{1,2\}$.
\State \textbf{On each $\textit{CS}_p$, $p\in [1,2]$:}
\For{$i \in [a,a+2^i)$} \Comment{convert instructions to oblivious operations}
\For{$e \in[1,\log_2N]$} \Comment{on each layer of $E_r$}
	\State Compute target block ID $s_e=|S_r|+(e-1)\times N + o_i$
	\State Retrieve block $b_{s_e}$ at the offset $o_i$ of the $e$ level of path $l_i$ inside $E_r$
	\State $y_{i,p}^e\leftarrow \textsf{Eval}(k_{i,p}, s_e)$
	\EndFor
    \State $b_{i,p} \leftarrow \bigoplus\limits^{\log_2N}_{e=1}  \left(y_{i,p}^e\times b_{s_e} \right)$  
    \State Send $b_{i,p} \leftarrow (b_{i,p}\xor m_{i,p})$ and $s'_i$ to $S_r$
\EndFor  
\State \textbf{On $S_r$:}
\For{$i \in [a,a+2^i)$} \Comment{Update on $S_r$}
\State Store $(b_{i,1}\xor b_{i,2})$ at $s^\prime_i$ of $S_r$ %
\EndFor
    
\end{algorithmic}
\end{algorithm}
\end{minipage}
\vspace{-5pt}
\end{figure}

\subsection{Block Range Return}
\label{app:BlockRangeRet}

\textsf{BlockRangeRet} enables routing a collection of $2^r$ video \textsf{block}s from the ORAM storage $E_r$ to the $S_r$~(Algorithm~\ref{alg:BlockRangeRet}).

Let $Q$ be the range query collection, where each block $b_{s_i}$, identified by $\textsf{bid}_i$, with block state $(c_i,l_i,s_i)$.
The algorithm re-encrypts $b_{s_i}$ with new counter $c_i+1$ and places the re-encrypted versions of the block in a randomly selected $s'_i$ index in the $S_r$. and its block state is updated accordingly.

\subsection{Permutation \& Re-encryption}
\label{app:PerReFunc}

\textsf{PerReFunc} re-encrypts and shuffles all \textsf{block}s in the $S_r$, which is presented in Algorithm~\ref{alg:PerReFunc}.

The \textit{Enclave} in the \textit{CP} triggers \textsf{PerReFunc}  by permuting the $\textsf{MetaStash}_r$ via the mapping $\omega$ that indicates new random location and re-encrypted block for every \textsf{block} within the $S_r$.
Then, $\omega$ and re-encryption keys are then shared to each $\textit{CP}_{p \in \{1,2\}}$.
In this way, each of them cannot determine the re-encryption keys or trace the location replacement of \textsf{block}s in the $S_r$.



\begin{figure}[!t]
\vspace{-7pt}
\begin{minipage}{\linewidth}
\begin{algorithm}[H]
\caption{$\textsf{PerReFunc}(Q)$}
\label{alg:PerReFunc}
\begin{algorithmic}[1]
\footnotesize
\Require Query List $Q=\{(bid_i,s_i,s^\prime_i)\}$, where $i\in[0, |S_r|)$
\State \textbf{Inside the \textit{Enclave}:}
\State Init query batches $Q_1$ and $Q_2$
\For{$(bid_i,s_i,s^\prime_i) \in Q$} \Comment{generate instructions}
\State Sample a permutation $ \omega_{i,1}: s_i\mapsto s^*_i; \omega_{i,2}: s^*_i\mapsto s'_i$
\State  $m_{i,1} \xleftarrow{\$} (0,1)^\lambda$
\State $m_{i,2} \leftarrow m_{i,1}\oplus G(c_i)\oplus G(c_i+1)$
\State $Q_1 \leftarrow Q_1 \cup \{(\omega_{i,1},m_{i,1})\}$ 
\State $Q_2 \leftarrow Q_2 \cup \{(\omega_{i,2},m_{i,2})\}$ 
\State Oblivious update $\textsf{BlockState}[\textsf{bid}_i]$ with $(c_i+1)$ and $s'_i$  
\EndFor
\State \textit{Enclave} sends $Q_p$ to $CS_p$, $p\in\{1,2\}$.
\State \textbf{On $\textit{CS}_1$:}
\For{$s\in[0,|S_r|)$}
\State Retrieve block $b_{s}$ from $S_r$
\State Compute $b_{s}\leftarrow b_{s}\oplus m_{{s},1}$
\State Permute $b_{s}$ with $\omega_{{s},1}$
\EndFor
\State Send permuted blocks to $\textit{CS}_2$ in order
\State \textbf{On $\textit{CS}_2$:}
\For{$s\in[0,|S_r|)$}
\State Received block $b_{s}$ from $\textit{CS}_1$
\State Compute $b_{s}\leftarrow b_{s}\oplus m_{{s},2}$
\State Permute $b_{s}$ with $\omega_{{s},2}$
\EndFor
\State Store permuted blocks back to $S_r$ in order




\end{algorithmic}
\end{algorithm}
\end{minipage}
\end{figure}

\subsection{Path Range Return}
\label{app:PathRangeRet}

Algorithm~\ref{alg:PathRangeRet} details \textsf{PathRangeRet}.
Given the set $Q$ which contains \textsf{blocks} in $2^r$ ORAM paths of an ORAM storage $E_r$, the \textsf{PathRangeRet} allows the \textit{CP} to route them  to \textit{ES}s that maintains the $S_r$. 
For each block $\textsf{bid}_i$ in $Q$, the \textit{Enclave} generates the random permutation $\omega$ and re-encryption keys that can be shared to \textit{CS}s.
The path information is given, so $\textit{CS}_{1}$ and $\textit{CS}_{2}$ can correctly fetch, re-encrypt and permute the blocks according to the instruction.
But throughout the process, $\textit{CS}_{1}$ will only knows where the block comes from while $\textit{CS}_{2}$ knows where the block will be stored in the $S_r$.

\begin{figure}[!t]
\vspace{-12pt}
\begin{minipage}{\linewidth}
\begin{algorithm}[H]
\caption{\textsf{PathRangeRet}($Q$)} 
\label{alg:PathRangeRet}
\begin{algorithmic}[1]

\footnotesize
\Require Query List $Q=\{(bid_i,l_i,s_i,s^\prime_i)\}$, where $i\in[0, 2^r\times\log_2N\times Z)$
\State \textbf{Inside the \textit{Enclave}:}
\State Init  query batches $Q_1$ and $Q_2$
   \For{$(bid_i,l_i,s_i,s^\prime_i) \in Q$}  \Comment{generate instructions}
	\State Sample a permutation $ \omega_{i,1}: s_i\mapsto s^*_i; \omega_{i,2}: s^*_i\mapsto s'_i$
	\State  $m_{i,1} \xleftarrow{\$} (0,1)^\lambda$
	\State $m_{i,2} \leftarrow m_{i,1}\oplus G(c_i)\oplus G(c_i+1)$
	\State $Q_1 \leftarrow Q_1 \cup \{(l_i,\omega_{i,1},m_{i,1})\}$ 
	\State $Q_2 \leftarrow Q_2 \cup \{(l_i,\omega_{i,2},m_{i,2})\}$ 
    \State Oblivious update $\textsf{BlockState}[\textsf{bid}_i]$ with $(c_i+1)$ and $s'_i$ 
     
    \EndFor  
	\State \textit{Enclave} sends $Q_p$ to $CS_p$, $p\in\{1,2\}$.

	\State \textbf{On $\textit{CS}_1$:}
	\For{each instruction $\{(l_i,\omega_{i,1},m_{i,1})\}$ in $Q_1$}
	\State Retrieve block $b_{s_i}$ from path $l_i$ of $E_r$
	\State Compute $b_{s_i}\leftarrow b_{s_i}\oplus m_{{s_i},1}$
	\State Permute $b_{s_i}$ with $\omega_{{s_i},1}$
	\EndFor
	\State Send permuted blocks to $\textit{CS}_2$ in order
	\State \textbf{On $\textit{CS}_1$:}
	\For{each instruction $\{l_i,\omega_{i,1},m_{i,1}\}$ in $Q_2$}
	\State Retrieve block $b_{s_i}$ of path $l_i$ from $\textit{CS}_1$
	\State Compute $b_{s_i}\leftarrow b_{s_i}\oplus m_{{s_i},2}$
	\State Permute $b_{s_i}$ with $\omega_{{s_i},2}$
	\EndFor
\State Store permuted blocks back to $S_r$ in order
\end{algorithmic}
\end{algorithm}
\end{minipage}
\vspace{-10pt}
\end{figure}

\begin{figure}[!t]
\vspace{-7pt}
\begin{minipage}{\linewidth}
\begin{algorithm}[H]
\caption{$\textsf{PriRangeEvict}(Q)$} 
\label{alg:PriRangeEvict}
\begin{algorithmic}[1]
\footnotesize 
\Require Query List $Q=\{(\textsf{bid}_i,c_i,l_i,l'_i,s_i,s^\prime_i)\}$, where $i\in[a,a+2^r\times\log_2N\times Z)$
\State \textbf{Inside the \textit{Enclave}:}
    \State Init query batches $Q_1$ and $Q_2$

\For{$(\textsf{bid}_i,c_i,l_i,s_i,s^\prime_i) \in Q$} \Comment{generate instructions}

\State  $(k_{i,1},k_{i,2}) \leftarrow \textsf{DPF.Gen}(1^\lambda, s_i)$
\State  $m_{i,1} \xleftarrow{\$} (0,1)^\lambda$
\State $m_{i,2} \leftarrow m_{i,1}\oplus G(c_i)\oplus G(c_i+1)$ if $\textsf{bid}_i \neq -1$; otherwise, $k^2_i \xleftarrow{\$} (0,1)^\lambda$
   
\State $Q_1 \leftarrow Q_1 \cup \{(k_{i,1},m_{i,1},l_i,s^\prime_i)\}$ 
\State $Q_2 \leftarrow Q_2 \cup \{(k_{i,2},m_{i,2},l_i,s^\prime_i)\}$ 
\State Oblivious update $\textsf{BlockState}[\textsf{bid}_i]$ with $(c_i+1)$ and $s'_i$ 
   \EndFor

\State \textit{Enclave} sends $Q_p$ to $CS_p$, $p\in\{1,2\}$
\State \textbf{On each $\textit{CS}_p$, $p\in [1,2]$:}
    \State Retrieve the entire $S_r$ with block $b_s, s\in[0, |S_r|)$
    \For{each instruction $(k_{i,p},m_{i,p},l_i,s^\prime_i)$ in $Q_p$} \Comment{convert instructions to oblivious operations}
    \For{$s\in[0,|S_r|)$}
    \State $y_{i,p}^s\leftarrow \textsf{Eval}(k_{i,p}, s)$
    \EndFor
    \State $b_{i,p} \leftarrow \bigoplus\limits^{|S_r|-1}_{s=0}  \left(y_{i,p}^s\times b_{s} \right)$  

    \EndFor 
    \State Send $b_{i,p} \leftarrow (b_{i,p}\xor m_{i,p})$ and $(l_i,s'_i)$ to $E_r$
\State \textbf{On $E_r$:}
\For{$i \in [a,a+2^r\times\log_2N\times Z)$} \Comment{Update on $E_r$}
\State Compute $e_i\leftarrow [\left(s'_i-|S_r|\right)/Z] +1$ and offset $o_i\leftarrow\left(s'_i - |S_r|\right)~\%~Z$
\State Store $(b_{i,1}\xor b_{i,2})$ at offset $o_i$ of the layer $e_i$ of path $l_i$ on $E_r$ %
\EndFor
\end{algorithmic}
\end{algorithm}
\end{minipage}
\vspace{-5pt}
\end{figure}

\subsection{Private Range Eviction}
\label{app:PriRangeEvict}
Given the set $Q$, the collection of $2^r \times \log_2N \times Z$ blocks currently located in the $S_r$, the \textsf{PriRangeEvict} privately writes them to $2^r$ ORAM paths in the storage $E_r$.

Let a $\textsf{block}_i$, with $\textsf{bid}_i$ and the current state $c_i$, is located at the index $s_i$ in the $S_r$.
\textsf{PriRangeEvict} privately route the block to the location $s'_i$ in  $E_r$.
Algorithm~\ref{alg:PriRangeEvict} presents this building block.
Note that $s'_i$ can be easily converted to layer and offset information that can be used to uniformly locate a block with path information $l_i$.

\section{Security Analysis}
\label{app:security_analysis}

\subsection{Overview}
As mentioned, in \textsf{OblivCDN}, we follow real-world CDN systems~\cite{NetflixAWS,NetflixOCA} to consider four parties, i.e., the \textit{CDN service provider}, the \textit{cloud service providers}, the \textit{ISP partners}, and the \textit{Subscriber}, where the \textit{cloud service providers} and \textit{ISP partners} are considered as untrusted.
Hence, the security analysis will start from the most simple case, where the \textit{CDN service provider} leverages three distinct \textit{cloud service providers} to deploy \textit{CP}, $\textit{CS}_1$, and $\textit{CS}_2$, respectively.
And for the \textit{ISP partner}, the \textit{CDN service provider} is contracted with an \textit{ISP partner} to host all \textit{ES}s.
In the above model, it clear that each party can only observer the local view for the entire \textsf{OblivCDN} system, and we will depict the view of those parties and demonstrate their security guarantees under the real/ideal paradigm.

Then, we will discuss more realistic cases with bandwidth-saving and economical considerations:
The first case is that the \textit{CDN service provider} deploys one $\textit{CS}$ on the \textit{ISP partner} who hosts \textit{ES}s, which will further enable the \textit{ISP partner} to collect information from $\textit{CS}$.
The second one, on the top of the first case, further considers the \textit{CDN service provider} to use the servers from the same \textit{cloud service provider} but located in different regions (e.g., AWS instances in multiple available zones) to deploy \textit{CP} and another \textit{CS}, respectively.
This enables the \textit{cloud service provider} to collect the information both from \textit{CP} and one \textit{CS}.
Our goal is to illustrate that even \textsf{OblivCDN} contains some servers controlled by the same concerned party, its security guarantee will not be compromised if \textit{cloud service providers} and \textit{ISP partners} are not colluding.


%


%
%




\subsection{Primitive Security}
\noindent{\bf PRF/PRP Security.} Both pseudorandom functions (PRFs) and pseudorandom permutations (PRPs) can be considered as indistinguishable from the random functions to adversaries~\cite{KatzLindell2007}, 
In the following sections, we use random functions to replace PRFs/PRPs, while the adversary only has a negligible chance to distinguish them.

\noindent{\bf ORAM Security.} The ORAM security ensures that any access request sequences with the same length always result the an indistinguishable access pattern.
The only thing can be observe is the statistical information regarding the ORAM store, which is the number of tree nodes $N$, the number of blocks in each node $Z$, and the block size $B$ for Path ORAM~\cite{Stefanov13}.
Meanwhile, the rORAM additionally reveals the maximum supported range $R$, and range size $2^r, r\in[0,\log_2R]$ supported by each ORAM tree~\cite{Chakraborti191}.
A simulator for Path ORAM ($\mathcal{S}_{ORAM}$) and rORAM ($\mathcal{S}_{rORAM}$) can use the above-mentioned statistical information to initialise and simulate random access pattern accordingly without lossing its correctness.

\noindent{\bf TEE Security.} For TEE, it provides an isolated memory region (called Enclave) in a server that cannot be directly read by other parts that are considered as untrusted.
Furthermore, all data loaded into the Enclave will be automatically encrypted and only decrypted when loaded and processing inside CPU~\cite{Costan161}.
However, adversaries can use side-channels~\cite{Brasser17,Gotzfried17,Xu15,Gruss17} to extract the access pattern of that isolated memory region even they cannot read the memory or decrypt its content.

\noindent{\bf DPF Security.} DPF guarantees that the primitive leaks no information regarding the input $x^*$ and $f(x)$ to the corrupted party~\cite{Boyle16}.
Formally, given the point function $f$, the leakage function of DPF regarding $f$ can be defined as $\mathcal{L}_{\textrm{DPF}}(f)=(l, \mathbb{F})$, where $l$ is the input size, and $\mathbb{F}$ is $f$'s output field.
The simulator $\mathcal{S}_{\textrm{DPF}}$ takes the security parameter $\lambda$ and $\mathcal{L}_{\textrm{DPF}}(f)$ as inputs to simulate DPF keys.
The simulated keys are indistinguishable from the keys generated by the real DPF instantiation, while maintaining DPF's correctness.

\subsection{Security Definition for \textsf{OblivCDN} Entites}
We start from the security definition of each individual entities of \textsf{OblivCDN}.
For the  adversaries on \textit{CP}, denoted as $\mathcal{A}_{\textit{CP}}$, its view can be defined as the Path ORAM/rORAM parameters (i.e., $N,Z,B,R,r$), which can be extracted from the access pattern on allocated location even it is inside the \textit{Enclave}.
Nonetheless, the above leaked information will not have any negative impact on the data-oblivious guarantee to the \textit{Subscriber} as stated in the following theorem:

\begin{theorem}\label{theo:cp}
Let $\Pi$ be the scheme of ({\textbf{Setup}, \textbf{Upload}, \textbf{Fetch}, \textbf{Sync}}). 
The adversary on \textit{CP} ($\mathcal{A}_\textit{CP}$) chooses the Path ORAM and rORAM parameters ($N,Z,B,R$) for \textbf{Setup} and selects a list of same-length \textsf{video}s for the \textit{CDN service provider} to \textbf{Upload}.
Later, $\mathcal{A}_\textit{CP}$ can arbitrarily \textbf{Fetch} videos and observe the access pattern on \textit{CP}.
\textsf{OblivCDN} ensures that the \textit{Enclave} in the CP is data-oblivious, assuming that Path ORAM and rORAM are data-oblivious, the PRF construction is secure, the permutations are selected uniformly at random, and TEE isolation and encryption features are available.
\end{theorem}

Similarly, we can define the adversary on one of $\textit{CS}_{p \in \{1,2\}}$ ($\mathcal{A}_{\textit{CS}_p}$), who also chooses the Path ORAM and rORAM parameters, as well as a list of same-length \textsf{video}s for \textbf{Upload} and \textbf{Fetch}.
However, we claim that adversaries can only learn the ORAM parameters if they only compromise one $\textit{CS}$:

\begin{theorem}\label{theo:cs}
Let $\Pi$ be the scheme of ({\textbf{Setup}, \textbf{Upload}, \textbf{Fetch}, \textbf{Sync}}). 
The adversary on $\textit{CS}_{p \in \{1,2\}}$ ($\mathcal{A}_{\textit{CS}_p}$) chooses the Path ORAM and rORAM parameters ($N,Z,B,R$) for \textbf{Setup} and selects a list of same-length \textsf{video}s for the \textit{CDN service provider} to \textbf{Upload}.
Later, $\mathcal{A}_{\textit{CS}_p}$ can arbitrarily \textbf{Fetch} videos and observe the tokens and blocks sent in and out on $\textit{CS}_p$.
\textsf{OblivCDN} ensures that \textit{CS} is data-oblivious against the adversary, assuming that rORAM is data-oblivious, the DPF and PRF constructions are secure, and the permutations are selected uniformly at random.
\end{theorem}

Finally, we discuss the adversary who can control \textit{ES}s ($\mathcal{A}_\textit{ES}$). 
Note that $\mathcal{A}_\textit{ES}$ have the same capability as all prior adversary but the following theorem indicates that its view can still be simulated by ORAM parameters ($N,Z,B,R$):

\begin{theorem}\label{theo:es}
Let $\Pi$ be the scheme of ({\textbf{Setup}, \textbf{Upload}, \textbf{Fetch}, \textbf{Sync}}). 
The adversary on $\textit{ES}$ ($\mathcal{A}_{\textit{ES}}$) chooses the Path ORAM and rORAM parameters ($N,Z,B,R$) for \textbf{Setup} and selects a list of same-length \textsf{video}s for the \textit{CDN service provider} to \textbf{Upload}.
Later, $\mathcal{A}_{\textit{ES}}$ can arbitrarily \textbf{Fetch} videos and observe the blocks sent in and out on $\textit{ES}$.
\textsf{OblivCDN} ensures that \textit{ES} is data-oblivious against the adversary, assuming that rORAM is data-oblivious, the  PRF constructions is secure, and the permutations are selected uniformly at random.
\end{theorem}

We provide the proof for above theorems in Appendix~\ref{app:analysis_enclave_ind}.

\subsection{Security Analysis for \textsf{OblivCDN} Entities}
\label{app:analysis_enclave_ind}

We provide the security proof for each individual parts of \textsf{OblivCDN}, starting from \textit{Enclave} in the \textit{CP} (ref. Theorem~\ref{theo:cp}) as follows:

%


\begin{proof}\textit{(Security of the \textit{Enclave}).}
We follow the real/ideal paradigm to consider follow two games:
In the real world game, $\mathcal{A}_\textit{CP}$ interacts with a real system implemented in the \textit{CP}, called $\Pi_{\textit{CP}}$, which runs the algorithms for \textit{CP} as stated in Algorithm~\ref{oblivcdn}.
In the ideal world game, $\mathcal{A}_\textit{CP}$ interacts with a simulator $\mathcal{S}_{\textit{CP}}$ simulating the system.
In both experiments, the adversary provides the same ORAM parameters ($N,Z,B,R$) and can supply any number of same-length \textsf{video}s operated in \textbf{Upload}, \textbf{Fetch}, and \textbf{Sync} operations.

We start from the ideal world game, where the simulator can be constructed as follows:

\noindent \ul{\textbf{Setup}:}
\begin{enumerate}[leftmargin=*]
	\item Runs the Path ORAM simulator $\mathcal{S}_{ORAM}$ with $N,Z,B$ to initiate data structures \textsf{VideoMap} and \textsf{BlockStateMap};
	\item Runs the rORAM simulator $\mathcal{S}_{rORAM}$ with $N,Z,B,R$ to initiate $\textsf{BlockTracker}_r$, $\forall r \in [0, \textnormal{log}_2R]$;
	\item Initiates flat arrays $\textsf{MetaStash}_r \forall r \in [0, \textnormal{log}_2R]$;
	\item Sends instructions to \textit{ES}s to initialise the data structures with the same parameters;
\end{enumerate}
Throughout the \textbf{Setup} process, $\mathcal{A}_\textit{CP}$ observer exactly the ORAM parameters $N,Z,B,R$ from the memory allocations.

\noindent \ul{\textbf{Upload}:} 
\begin{enumerate}[leftmargin=*]
	\item Runs $\mathcal{S}_{ORAM}$ to update \textsf{VideoMap} and $\textsf{BlockStateMap}$ to track the states of newly added video and blocks;
	\item Runs $\mathcal{S}_{ORAM}$ to access $\textsf{BlockStateMap}$ and get block states;
	\item Selects randomly generated masks to encrypt all blocks;
	\item Insert blocks into $\textsf{MetaStash}_r \forall r \in [0, \textnormal{log}_2R]$ according to their state $s$;
	\item Write blocks into $S_r$ on \textit{ES};
\end{enumerate}
The \textbf{Upload} process still reveals the ORAM parameters $N,Z,B,R$ from the ORAM accesses as it always read fixed $\log_2N\times Z$ blocks for metadata accesses.
It additionally leaks range $2^r$ for each rORAM tree and stash since it always access $2^r$ blocks in those data structures.
No information will be revealed from the video blocks sent out to \textit{ES} since they are encrypted by random masks.

\noindent \ul{\textbf{Fetch}:} 
\begin{enumerate}[leftmargin=*]
	\item Runs $\mathcal{S}_{ORAM}$ to load $bid$ and states of given video ID from \textsf{VideoMap} and $\textsf{BlockStateMap}$;
	\item Sends the information to \textit{Subscriber}s;
	\item Runs $\mathcal{S}_{rORAM}$ to simulate the process of fetching $bid$s from $\textsf{BlockTracker}_r$ to $\textsf{MetaStash}_r$;
	\item Runs $\mathcal{S}_{ORAM}$ to access $\textsf{BlockStateMap}$ and get block states for blocks added into $\textsf{BlockTracker}_r$;
	\item Refers to the block state changes from $\mathcal{S}_{rORAM}$ to runs DPF simulator $\mathcal{S}_{DPF}$ to generate $2^r$ key pairs;
	\item Selects $2^r$ random masks as the re-encryption keys;
	\item Combines DPF key, random masks and block states as $2^r$ \textsf{BlockRangeRet} tokens and sends to \textit{CS}s;
	\item Runs $\mathcal{S}_{ORAM}$ to update the block states in $\textsf{BlockStateMap}$;
	\item Generates $|S_r|$ random permutations to permute $\textsf{MetaStash}_r$;
	\item Selects $|S_r|$ random masks as the re-encryption keys;
	\item Combines permutations and random masks as $|S_r|$ \textsf{PerReFunc} tokens and sends to \textit{CS}s;
	\item Runs $\mathcal{S}_{ORAM}$ to update the block states in $\textsf{BlockStateMap}$;
\end{enumerate}
During \textbf{Fetch}, the \textit{Enclave} reveals the ORAM parameters $N,Z,B,R$ from the ORAM accesses on its metadata as well as the number of tokens sent to \textit{CS}s.
Furthermore, it reveals range $2^r$ for each rORAM tree and stash from the above process.
Nonetheless, the token itself does not reveal any extra information, because the DPF key is simulated via the input field size, which is $|S_r|+\log_2N\times Z$ (depicted by ORAM parameters), the masks/permutations are selected randomly, and the location information (path, offset, and new ID) inside are simulated via the rORAM simulator $\mathcal{S}_{rORAM}$.

\noindent \ul{\textbf{Sync}:} 
\begin{enumerate}[leftmargin=*]
	\item Runs $\mathcal{S}_{rORAM}$ to simulate the process of fetching all $2^r$ paths from $\textsf{BlockTracker}_r$ to $\textsf{MetaStash}_r$;
	\item Runs $\mathcal{S}_{ORAM}$ to access $\textsf{BlockStateMap}$ and get block states for blocks added into $\textsf{BlockTracker}_r$;
	\item Refers to the block state changes from $\mathcal{S}_{rORAM}$ to select $2^r\times\log_2N\times Z$ random permutations;
	\item Selects $2^r\times\log_2N\times Z$ random masks as the re-encryption keys;
	\item Combines permutations and random masks as $2^r\times\log_2N\times Z$ \textsf{PathRangeRet} tokens and sends to \textit{CS}s;
	\item Runs $\mathcal{S}_{ORAM}$ to update the block states in $\textsf{BlockStateMap}$;
	\item Runs $\mathcal{S}_{rORAM}$ to simulate the process of evicting $bid$s from $\textsf{MetaStash}_r$ to $2^r$ paths in $\textsf{BlockTracker}_r$;
	\item Refers to the block state changes from $\mathcal{S}_{rORAM}$ to runs DPF simulator $\mathcal{S}_{DPF}$ to generate $2^r\times\log_2N\times Z$ key pairs;
	\item Selects $2^r\times\log_2N\times Z$ random masks as the re-encryption keys;
	\item Combines DPF key, random masks and block states as $2^r\times\log_2N\times Z$ \textsf{PriRangeEvict} tokens and sends to \textit{CS}s;
	\item Runs $\mathcal{S}_{ORAM}$ to update the block states in $\textsf{BlockStateMap}$;
\end{enumerate}
The \textbf{Sync} process also reveals the ORAM parameters $N,Z,B,R$ and range $2^r$ of rORAM tree and stash from the ORAM accesses on its metadata as well as the number of tokens sent to \textit{CS}s.
On the other hand, the tokens do not reveal extra information as they are all simulated with ORAM parameters or sampled uniformly at random.

We now prove that the view generated by  $\mathcal{S}_{\textit{CP}}$ in the ideal world is indistinguishable from the real world execution of $\Pi_{\textit{CP}}$, which runs the algorithms for \textit{CP} as stated in Algorithm~\ref{oblivcdn}, against the adversary $\mathcal{A}_\textit{CP}$.
The following sequence of hybrid executions $\mathcal{H}_0,\cdots,\mathcal{H}_3$ gradually change invocations of the simulators to real sub-protocols, which demonstrates that the above ideal world execution is indistinguishable from the real-world as long as all sub-protocols are proven to be secure:

\noindent \textbf{Hybrid 0.} 
We start with the ideal world as our initial hybrid, $\mathcal{H}_0$.

\noindent \textbf{Hybrid 1.} \textbf{Hybrid 1} is the same as \textbf{Hybrid 0}, except that the simulator $\mathcal{S}_{\textit{CP}}$ uses real PRP/PRF outputs to replace the permutations and encryption/re-encryption keys for \textbf{Upload}, \textbf{Fetch}, and \textbf{Sync}.
Security of pseudorandom functions holds that $\mathcal{A}_\textit{CP}$'s advantage in distinguishing $\mathcal{H}_0$ from $\mathcal{H}_1$ is $negl(\lambda)$.

\noindent \textbf{Hybrid 2.} \textbf{Hybrid 2} is the same as \textbf{Hybrid 1}, except that the simulator $\mathcal{S}_{\textit{CP}}$ uses a real DPF scheme to generate DPF keys in the tokens during \textbf{Fetch} and \textbf{Sync}.
The DPF security ensures that $\mathcal{A}_\textit{CP}$'s advantage in distinguishing $\mathcal{H}_1$ from $\mathcal{H}_0$ is $negl(\lambda)$.

\noindent \textbf{Hybrid 3.} \textbf{Hybrid 3} is the same as \textbf{Hybrid 2}, except that the simulator $\mathcal{S}_{\textit{CP}}$ uses a real Path ORAM/rORAM instance to enable access pattern hiding access.
Accordingly, subsequent accesses on the same-length requests cause the same distribution of RAM access patterns. 
Security of Path ORAM and rORAM implies that $\mathcal{A}_\textit{CP}$'s advantage in distinguishing $\mathcal{H}_3$ from $\mathcal{H}_2$ is $negl(\lambda)$.

The above hybrid executions show that the simulated view is indistinguishable from the real world.
Hence, we can conclude that the \textit{Enclave} in the \textit{CP} is data-oblivious, and $\mathcal{A}_\textit{CP}$ cannot infer sensitive information via memory access side-channels.
\end{proof}

For the adversary on \textit{CS}s, we will have the following proof for Theorem~\ref{theo:cs}:
\begin{proof}\textit{(Security of the \textit{CS}).}
We follow the real/ideal paradigm to consider follow two games:
In the real world game, an adversary $\mathcal{A}_{\textit{CS}_p}$ interacts with a real implementation of $\textit{CS}_{p \in \{1,2\}}$, running the algorithms $\Pi_{\textit{CS}_p}$ as stated in Algorithm~\ref{oblivcdn}, which consist of \textsf{PathRangeRet}, \textsf{PriRangeEvict}, \textsf{PerReFunc}, and \textsf{BlockRangeRet}. 
In the ideal world game, $\mathcal{A}_{\textit{CS}_p}$ interacts with a simulator $\mathcal{S}_{\textit{CS}_p}$ simulating the system.
In both experiments, the adversary provides the same ORAM parameters ($N,Z,B,R$) and can supply any number of same-length \textsf{video}s operated in \textbf{Upload}, \textbf{Fetch}, and \textbf{Sync} operations.

Again, we start from the ideal world game, where the simulator can be constructed as follows\footnote{Note that $\mathcal{S}_{\textit{CS}_p}$ only receives information in \textbf{Fetch}, and \textbf{Sync} operations, so we omit \textbf{Setup} and \textbf{Upload} operations as $\mathcal{A}_{\textit{CS}_p}$ only submit its inputs (ORAM parameters and videos) but cannot obtain a view on those operations}:

At the beginning of the simulation, the adversary chooses the $CS_p$ to compromise, denoted $\textsf{ComID} \in \{1,2\}$. Then,

\noindent \ul{\textbf{Fetch}:} 
\begin{enumerate}[leftmargin=*]
	\item Receives $2^r$ simulated tokens from the \textit{CP} for \textsf{BlockRangeRet};
	\item For each token, retrieves $\log_2N$ blocks from the designated path of $E_r$, the simulator responds with $\log_2N$ random blocks for each requests;
	\item Runs \textsf{BlockRangeRet} against retrieved blocks;
	\item Writes the result blocks to the designated locations of $S_r$. The simulator runs $\mathcal{S}_{rORAM}$ to write the corresponding block according to the instructions.
	\item Receives $|S_r|$ simulated tokens from the \textit{CP} for \textsf{PerReFunc};
	\item If $\textsf{ComID}=1$, retrieves the entire stash $S_r$ and receives $|S_r|$ random blocks from the simulator, permutes/re-encrypt all blocks inside, and sends to $\textit{CS}_2$;
	\item If $\textsf{ComID}=2,$, retrieves $|S_r|$ blocks from $\textit{CS}_1$ and receives $|S_r|$ random blocks from the simulator, permutes/re-encrypt all blocks inside, and store in $S_r$. The simulator applies the permutations and re-encryption keys to permute and re-encrypt the entire $S_r$ locally;
\end{enumerate}
During \textbf{Fetch}, $\mathcal{A}_{\textit{CS}_p}$ only can observe the number of tokens from the \textit{CP}and number of blocks from the $E_r$ and $S_r$, which directly correlates to the ORAM parameters.
On the other hand, the tokens and blocks received does not reveal any information, as they are all simulated or randomly selected by the simulator via the ORAM parameters (see our proof for Theorem~\ref{theo:csp}).

\noindent \ul{\textbf{Sync}:} 
\begin{enumerate}[leftmargin=*]
	\item Receives $2^r\times\log_2N\times Z$ simulated tokens from the \textit{CP} for \textsf{PathRangeRet};
	\item If $\textsf{ComID}=1$, retrieves $2^r\times\log_2N\times Z$ blocks from the designated locations of $E_r$, receives $2^r\times\log_2N\times Z$ random blocks from the simulator, permutes/re-encrypt all blocks inside, and sends to $\textit{CS}_2$;
	\item If $\textsf{ComID}=2,$, retrieves $2^r\times\log_2N\times Z$ blocks from $\textit{CS}_1$ and receives $2^r\times\log_2N\times Z$ random blocks from the simulator, permutes/re-encrypt all blocks inside, and store in $S_r$. The simulator runs $\mathcal{S}_{rORAM}$ to retrieve the corresponding blocks from $E_r$ to $S_r$ according to the instructions;
	\item Receives $2^r\times\log_2N\times Z$ simulated tokens from the \textit{CP} for \textsf{PriRangeEvict};
	\item Retrieves $2^r\times\log_2N\times Z$ blocks from the designated locations of $S_r$, the simulator responds with $2^r\times\log_2N\times Z$ random blocks for each requests;
	\item Runs \textsf{PriRangeEvict} against retrieved blocks;
	\item Writes the result blocks to the designated locations of $E_r$. The simulator runs $\mathcal{S}_{rORAM}$ to write the corresponding block according to the instructions.
\end{enumerate}
Similar to \textbf{Fetch}, $\mathcal{A}_{\textit{CS}_p}$ only can observe the number of tokens from the \textit{CP}and number of blocks from the $E_r$ and $S_r$, which directly correlates to the ORAM parameters in \textbf{Sync}.
No information is revealed from the simulated tokens or randomly selected blocks.

We now ready to prove that the view generated by $\mathcal{S}_{\textit{CS}_p}$ in the ideal world is indistinguishable from the real world execution of $\Pi_{\textit{CS}_p}$, which runs the algorithms for $\textit{CS}_p$ as stated in Algorithm~\ref{oblivcdn}, against the adversary $\mathcal{A}_{\textit{CS}_p}$ via the following sequence of hybrid executions $\mathcal{H}_0,\cdots,\mathcal{H}_4$.

\noindent \textbf{Hybrid 0.} 
We start with the ideal world as our initial hybrid, $\mathcal{H}_0$.

\noindent \textbf{Hybrid 1.} \textbf{Hybrid 1} is the same as \textbf{Hybrid 0}, except that the simulator $\mathcal{S}_{\textit{CS}_p}$  uses real PRP/PRF outputs to replace the permutations and encryption/re-encryption keys in the token.
Security of pseudorandom functions holds that $\mathcal{S}_{\textit{CS}_p}$'s advantage in distinguishing $\mathcal{H}_0$ from $\mathcal{H}_1$ is $negl(\lambda)$.

\noindent \textbf{Hybrid 2.} \textbf{Hybrid 2} is the same as \textbf{Hybrid 1}, except that the simulator $\mathcal{S}_{\textit{CS}_p}$  uses a real DPF scheme to generate DPF keys in the tokens.
The DPF security ensures that $\mathcal{S}_{\textit{CS}_p}$'s advantage in distinguishing $\mathcal{H}_1$ from $\mathcal{H}_0$ is $negl(\lambda)$.

\noindent \textbf{Hybrid 3.} \textbf{Hybrid 3} is the same as \textbf{Hybrid 2}, except that the simulator $\mathcal{S}_{\textit{CS}_p}$ forwards real blocks to $\textit{CS}_\textsf{ComID}$ upon the requests.
Since all blocks are encrypted with one-time PRF masks, the security of PRF ensures that $\mathcal{A}_{\textit{CS}_p}$'s advantage in distinguishing $\mathcal{H}_3$ from $\mathcal{H}_2$ is $negl(\lambda)$.

\noindent \textbf{Hybrid 4.} \textbf{Hybrid 4} is the same as \textbf{Hybrid 3}, except that the simulator $\mathcal{S}_{\textit{CS}_p}$ uses real oblivious block routing protocols to replace $\mathcal{S}_{rORAM}$.
For the compromised party $\textit{CS}_\textsf{ComID}$, our prior discussion pinpoints that it only receives instructions generated from rORAM oeprations and observes random blocks locally.
Meanwhile, it does not know actual block operations on \textit{ES}s. 
Hence, $\mathcal{A}_{\textit{CS}_p}$'s advantage in distinguishing $\mathcal{H}_4$ from $\mathcal{H}_3$ is $negl(\lambda)$.

The above hybrid executions show that the simulated view is computationally indistinguishable from the real world.
Hence, we can conclude that the $\mathcal{A}_{\textit{CS}_p}$ cannot infer sensitive information via the view of compromised server.
\end{proof}

Finally, we discuss the security of \textit{ES}s against the adversary on $\mathcal{A}_{\textit{ES}}$,in following proof for Theorem~\ref{theo:es}:

\begin{proof}\textit{(Security of the \textit{ES}s).}
Similarly, we start from constructing a simulator to simulate the view of \textit{ES}.
Note that \textit{ES} serves as data store and purely sends and receives blocks from other entities, which can be simply described as follows:

\noindent \ul{\textbf{Setup}:} 
\begin{enumerate}[leftmargin=*]
	\item Initialises $log_2R$ ORAM tree with $N\times Z\times B$ space;
	\item Initialises $log_2R$ stash with $|S_r|$ space; 
\end{enumerate}

\noindent \ul{\textbf{Upload}:} 
\begin{enumerate}[leftmargin=*]
	\item Receives random blocks and put into $S_r$;
\end{enumerate}

\noindent \ul{\textbf{Fetch}:} 
\begin{enumerate}[leftmargin=*]
	\item Send requested $2^r$ blocks to \textit{Subscriber}s for each request
	\item Send requested $2^r\times\log_2N$ blocks from $E_r$ to \textit{CS}s;
	\item Receives $2^r$ random blocks and put into $S_r$;
	\item Sends $|S_r|$ blocks from $S_r$ to $\textit{CS}_1$;
	\item Receives $|S_r|$ random blocks from $\textit{CS}_2$ and put into $S_r$;
\end{enumerate}

\noindent \ul{\textbf{Sync}:} 
\begin{enumerate}[leftmargin=*]
	\item Send requested $2^r\times \log_2N\times Z$ blocks from $E_r$ to $\textit{CS}_1$;
	\item Receives $2^r\times \log_2N\times Z$ random blocks and put into $S_r$;
	\item Sends $|S_r|$ blocks in $S_r$ to $\textit{CS}_1$;
	\item Receives $2^r\times \log_2N\times Z$ random blocks from $\textit{CS}_2$ and put into $E_r$;
\end{enumerate}

From the above description, we can see that \textit{ES} exactly implements what rORAM server did as in the original protocol~\cite{Chakraborti191}.
In addition, data-oblivious on client-side structure $S_r$ is ensured by random permutations.
Hence, we can conclude that \textit{ES} inherits the security guarantees of rORAM and is data-oblivious against the adversary who compromises \textit{ES}s only.
\end{proof}

\subsection{Security Analysis for \textsf{OblivCDN} System}
\label{app:analysis_ind}

In this section, we take real-world deployment considerations into the security model and further discuss the following two cases:
The first case is called colluding $\textit{ES}$s, where the \textit{CDN service provider} deploys one $\textit{CS}$ on the \textit{ISP partner} who hosts \textit{ES}s, which will further enable the \textit{ISP partner} to collect information from $\textit{CS}$.
Nonetheless, the following theorem indicates that the above adversary (denoted as $\mathcal{A}_{\textit{ISP}}$) cannot learn extra information other than the ORAM parameters ($N,Z,B,R$) as far as it only can compromise one $\textit{CS}$.

\begin{theorem}\label{theo:isp}
Let $\Pi$ be the scheme of ({\textbf{Setup}, \textbf{Upload}, \textbf{Fetch}, \textbf{Sync}}). 
The adversary on the \textit{ISP partner} ($\mathcal{A}_{\textit{ISP}}$) chooses the Path ORAM and rORAM parameters ($N,Z,B,R$) for \textbf{Setup} and selects a list of same-length \textsf{video}s for the \textit{CDN service provider} to \textbf{Upload}.
Later, $\mathcal{A}_{\textit{ISP}}$ can arbitrarily \textbf{Fetch} videos and observe the tokens/blocks sent in and out on one $\textit{CS}$ and $\textit{ES}$.
\textsf{OblivCDN} is data-oblivious against the $\mathcal{A}_{\textit{ISP}}$, assuming that rORAM is data-oblivious, the DPF and PRF constructions are secure, and the permutations are selected uniformly at random.
\end{theorem}

\begin{proof}
In this proof, we firstly describe how to construct a simulator to simulate the view of $\mathcal{A}_{\textit{ISP}}$.
	
At the beginning of the simulation, the adversary still chooses the $\textit{CS}_p$ to compromise, denoted $\textsf{ComID} \in \{1,2\}$. Then,

\noindent \ul{\textbf{Setup}:} 
\begin{enumerate}[leftmargin=*]
	\item Initialises $log_2R$ ORAM tree with $N\times Z\times B$ space;
	\item Initialises $log_2R$ stash with $|S_r|$ space; 
\end{enumerate}

\noindent \ul{\textbf{Upload}:} 
\begin{enumerate}[leftmargin=*]
	\item Receives random blocks and put into $S_r$;
\end{enumerate}

\noindent \ul{\textbf{Fetch}:} 
\begin{enumerate}[leftmargin=*]
	\item Receives $2^r$ simulated tokens from the \textit{CP} for \textsf{BlockRangeRet};
	\item For each token, retrieves $\log_2N$ blocks from the designated path of $E_r$
	\item Runs \textsf{BlockRangeRet} against retrieved blocks;
	\item Writes the result blocks to the designated locations of $S_r$. 
	\item Receives $|S_r|$ simulated tokens from the \textit{CP} for \textsf{PerReFunc};
	\item If $\textsf{ComID}=1$, retrieves the entire stash $S_r$, permutes/re-encrypt all blocks inside, and sends to $\textit{CS}_2$;
	\item If $\textsf{ComID}=2$, retrieves $|S_r|$ blocks from $\textit{CS}_1$ and receives $|S_r|$ random blocks from the simulator, permutes/re-encrypt all blocks inside, and store in $S_r$. The simulator applies the permutations and re-encryption keys to permute and re-encrypt the entire $S_r$ locally;
\end{enumerate}

\noindent \ul{\textbf{Sync}:} 
\begin{enumerate}[leftmargin=*]
	\item Receives $2^r\times\log_2N\times Z$ simulated tokens from the \textit{CP} for \textsf{PathRangeRet};
	\item If $\textsf{ComID}=1$, retrieves $2^r\times\log_2N\times Z$ blocks from the designated locations of $E_r$, permutes/re-encrypt all blocks inside, and sends to $\textit{CS}_2$;
	\item If $\textsf{ComID}=2,$, retrieves $2^r\times\log_2N\times Z$ blocks from $\textit{CS}_1$ and receives $2^r\times\log_2N\times Z$ random blocks from the simulator, permutes/re-encrypt all blocks inside, and store in $S_r$. The simulator runs $\mathcal{S}_{rORAM}$ to retrieve the corresponding blocks from $E_r$ to $S_r$ according to the instructions;
	\item Receives $2^r\times\log_2N\times Z$ simulated tokens from the \textit{CP} for \textsf{PriRangeEvict};
	\item Retrieves $2^r\times\log_2N\times Z$ blocks from the designated locations of $S_r$;
	\item Runs \textsf{PriRangeEvict} against retrieved blocks;
	\item Writes the result blocks to the designated locations of $E_r$. 
\end{enumerate}

For the above description, we can see that the simulation is exactly the same as the simulator of $\textit{CS}_p$ when replacing the $\mathcal{S}_{rORAM}$ with real oblivious block routing operations.
This is because we have proven that the view on $\textit{ES}$s is indistinguishable from the view of the rORAM server in Theorem~\ref{theo:es}.
It means that for an adversary who compromises $\textit{CS}_p$, it cannot learn extra information by additionally compromising $\textit{ES}$s, and vice versa.
Hence, the adversary $\mathcal{A}_{\textit{ISP}}$ does not gain extra advantage than $\mathcal{A}_{\textit{CS}_p}$, and \textsf{OblivCDN} is still data-oblivious against the new adversary.
\end{proof}

We further discuss the case where the \textit{CDN service provider} employs $\textit{CP}$ and one $\textit{CS}_p$ on the two servers from the same \textit{cloud service provider}.
In this case, the adversary ($\mathcal{A}_{\textit{CSP}}$) can further observe the memory access pattern and tokens from $\textit{CP}$, as well as the oblivious block routing process on $\textit{CS}_p$.
The following theorem captures the security of \textsf{OblivCDN} against $\mathcal{A}_{\textit{CSP}}$:

\begin{theorem}\label{theo:csp}
Let $\Pi$ be the scheme of ({\textbf{Setup}, \textbf{Upload}, \textbf{Fetch}, \textbf{Sync}}). 
The adversary on the \textit{cloud service provider} ($\mathcal{A}_{\textit{CSP}}$) chooses the Path ORAM and rORAM parameters ($N,Z,B,R$) for \textbf{Setup} and selects a list of same-length \textsf{video}s for the \textit{CDN service provider} to \textbf{Upload}.
Later, $\mathcal{A}_{\textit{CSP}}$ can arbitrarily \textbf{Fetch} videos and observe the access pattern on \textit{CP} and the tokens/blocks sent in and out on one $\textit{CS}$.
\textsf{OblivCDN} is data-oblivious against the $\mathcal{A}_{\textit{CSP}}$, assuming that rORAM is data-oblivious, the DPF and PRF constructions are secure, and the permutations are selected uniformly at random.
\end{theorem}

\begin{proof}
Since the proof structure for Theorem~\ref{theo:csp} is 	similar to that for Theorem~\ref{theo:isp}, we provide a proof sketch here.
In particular, the simulator for $\mathcal{A}_{\textit{CSP}}$ aims to generate a joint view on the access pattern of \textit{CP} for all \textsf{OblivCDN} operations and the tokens/blocks sent to/by $\textit{CS}_p$ under the given instructions.
Since the access pattern on \textit{CP} is oblivious, the above view is exactly the same as the simulator of $\textit{CS}_p$ when replacing the simulated tokens with with real tokens generated by PRF/PRP/DPF instances (i.e., \textbf{Hybrid 3} for Theorem~\ref{theo:cs}).
This means that for an adversary who compromises $\textit{CS}_p$, it cannot learn extra information by additionally compromising $\textit{CP}$, and vice versa.
As a result, the adversary $\mathcal{A}_{\textit{CSP}}$ does not gain extra advantage than $\mathcal{A}_{\textit{CS}_p}$, and \textsf{OblivCDN} is still data-oblivious against the new adversary.
\end{proof}

\begin{table}[!t]
\small
	\centering
	\caption{Computation costs (second) $|B|=2$ KB}
	\vspace{-5pt}
	\begin{tabular}{|c|c|c|c|c|}
	\hline
		Building Block & \textit{CP} & $\textit{CS}_1$ & $\textit{CS}_2$ & $E_{r=8}$ \& $S_{r=8}$ \\
	\hline
		\textsf{BlockRangeRet} & $0.1$ & $0.031$ & $0.028$  & $1.5$\\

	\hline
		\textsf{PerReFunc} & $8.5$ & $2.2$ & $2.7$ & $0.2$ \\
	\hline
		\textsf{PathRangeRet} & $6.7$ & $1.07$  & $1.3$  & $1.7$ \\
	\hline
		\textsf{PriRangeEvict} & $5.3$ & $368.5$ & $369.5$  & $0.2$ \\
	\hline
	\end{tabular}
    \vspace{-5pt}
	\label{tlb:micro_info1}
\end{table}

\section{Cost Breakdown for Building Blocks}
\label{app_sec:buildingblock}

We show the breakdown operation costs of the \textit{CP} (with the \textit{Enclave} included), each of $\textit{CS}_{p \in \{1,2\}}$, and the \textit{ES} involved in the four oblivious building blocks.
Table~\ref{tlb:micro_info1} and~\ref{tlb:micro_info2} show the costs when executing the building blocks in the ORAM $E_r$ and its stash $S_r$ with $r=8$, which support the range access of $2^8$ blocks with block size $2$ and $4$ KB, respectively.

\begin{table}[!t]
\small
	\centering
	\caption{Computation costs (second)  $|B|=4$ KB}
	\vspace{-5pt}
	\begin{tabular}{|c|c|c|c|c|}
	\hline
		Building Block & \textit{CP} & $\textit{CS}_1$ & $\textit{CS}_2$ & $E_{r=8}$ \& $S_{r=8}$ \\
	\hline
		\textsf{BlockRangeRet} & $0.1$ & $0.057$ & $0.063$  & $2.6$\\
	\hline
		\textsf{PerReFunc} & $8.5$ & $4.5$ & $5.4$ & $0.25$ \\
	\hline
		\textsf{PathRangeRet} & $6.7$ & $2.1$  & $2.6$  & $2.9$ \\
	\hline
		\textsf{PriRangeEvict} & $5.3$ & $735.3$ & $737.3$  & $0.23$ \\
	\hline
	\end{tabular}
    \vspace{-10pt}
	\label{tlb:micro_info2}
\end{table}

\begin{figure*}[!t]
	\centering
	\subfloat[Inter-continental, $r=0$]{
		\label{fig:inter_comm}
		\includegraphics[width=0.23\linewidth]{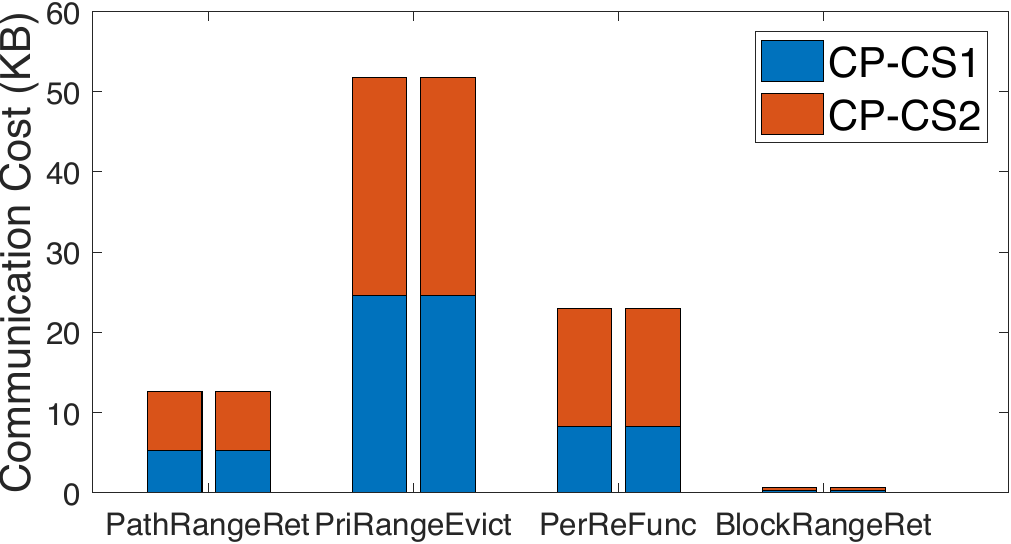}
	}
	\hfill
	\subfloat[Local region, $r=0$]{
		\label{fig:intra_comm}
		\includegraphics[width=0.23\linewidth]{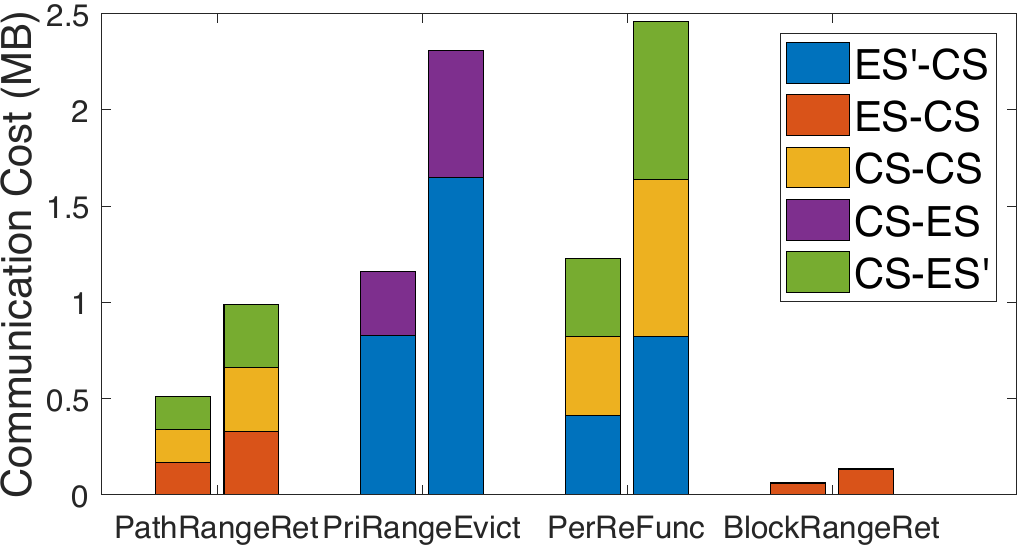}
	}
    \hfill
	\subfloat[Inter-continental, $r=1$]{
		\label{fig:inter_comm}
		\includegraphics[width=0.23\linewidth]{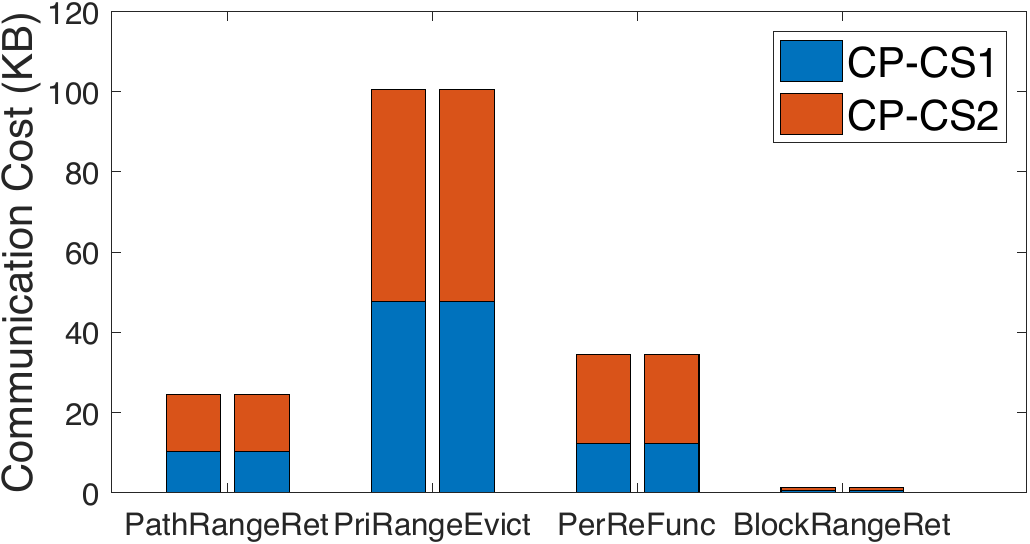}
	}
	\hfill
	\subfloat[Local region, $r=1$]{
		\label{fig:intra_comm}
		\includegraphics[width=0.23\linewidth]{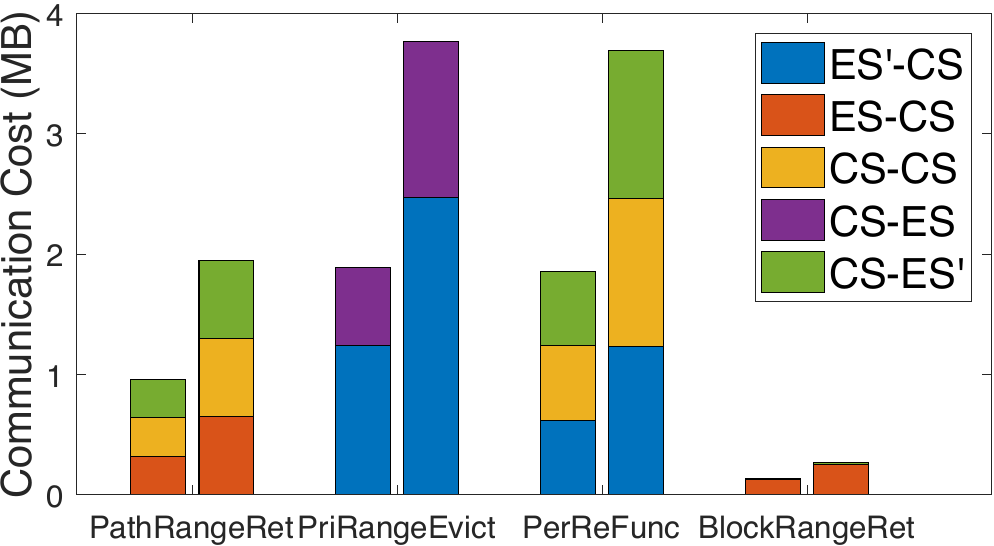}
	}
	
	\subfloat[Inter-continental, $r=2$]{
		\label{fig:inter_comm}
		\includegraphics[width=0.23\linewidth]{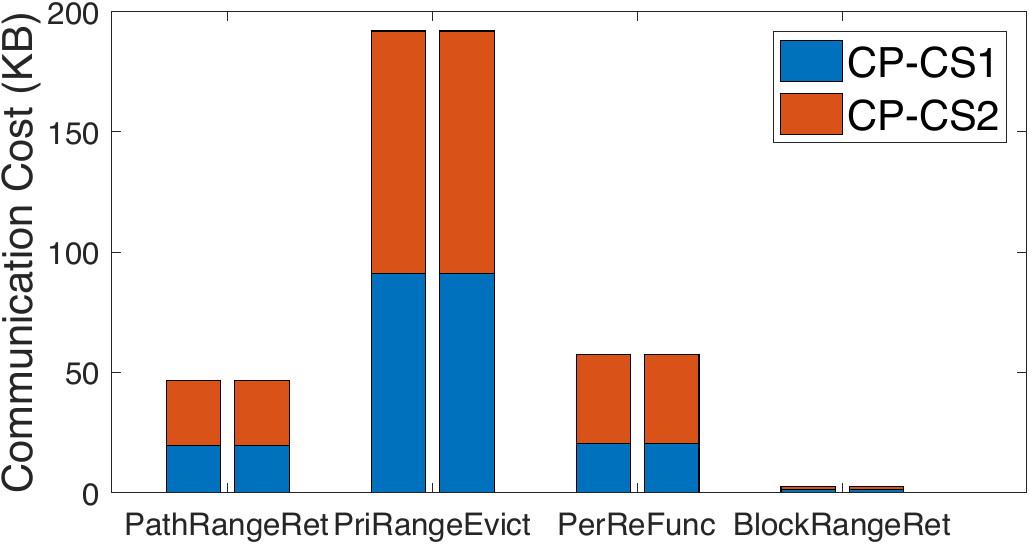}
	}
	\hfill
	\subfloat[Local region, $r=2$]{
		\label{fig:intra_comm}
		\includegraphics[width=0.23\linewidth]{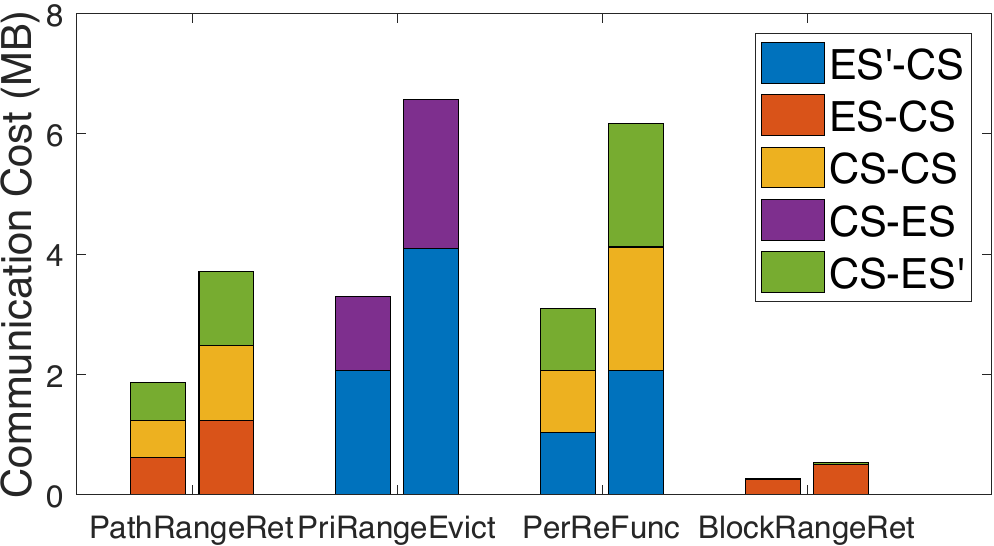}
	}
	\hfill
	\subfloat[Inter-continental, $r=3$]{
		\label{fig:inter_comm}
		\includegraphics[width=0.23\linewidth]{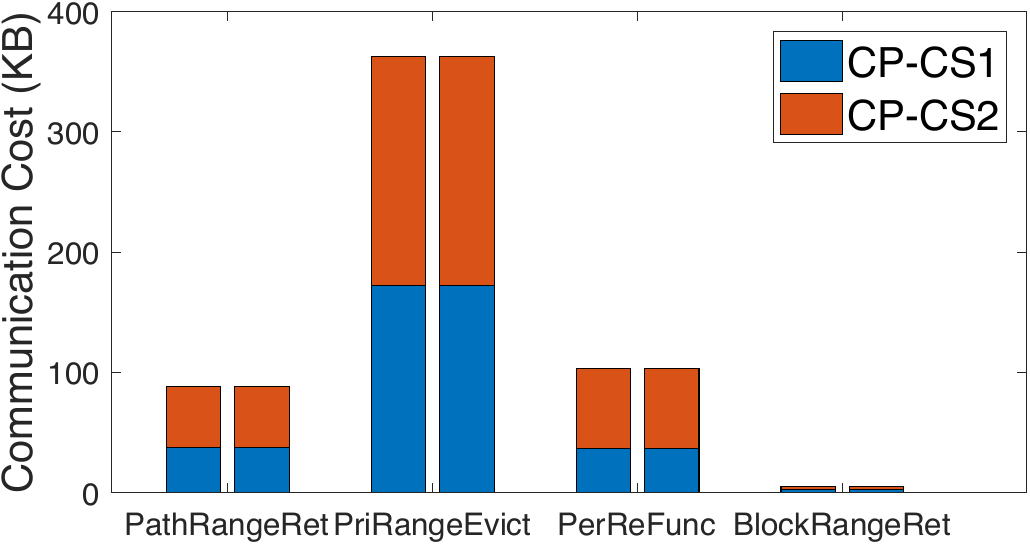}
	}
	\hfill
	\subfloat[Local region, $r=3$]{
		\label{fig:intra_comm}
		\includegraphics[width=0.23\linewidth]{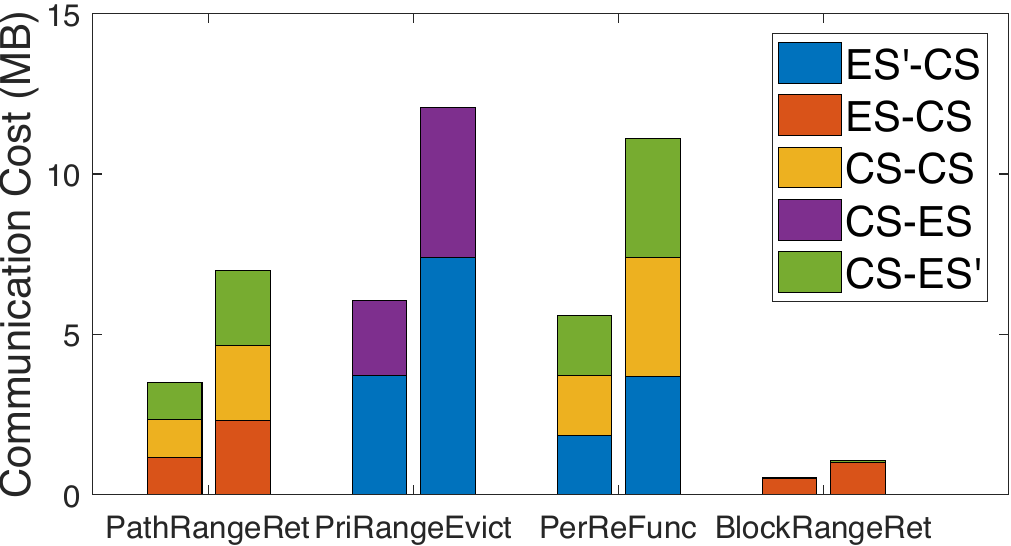}
	}
	
	\subfloat[Inter-continental, $r=4$]{
		\label{fig:inter_comm}
		\includegraphics[width=0.23\linewidth]{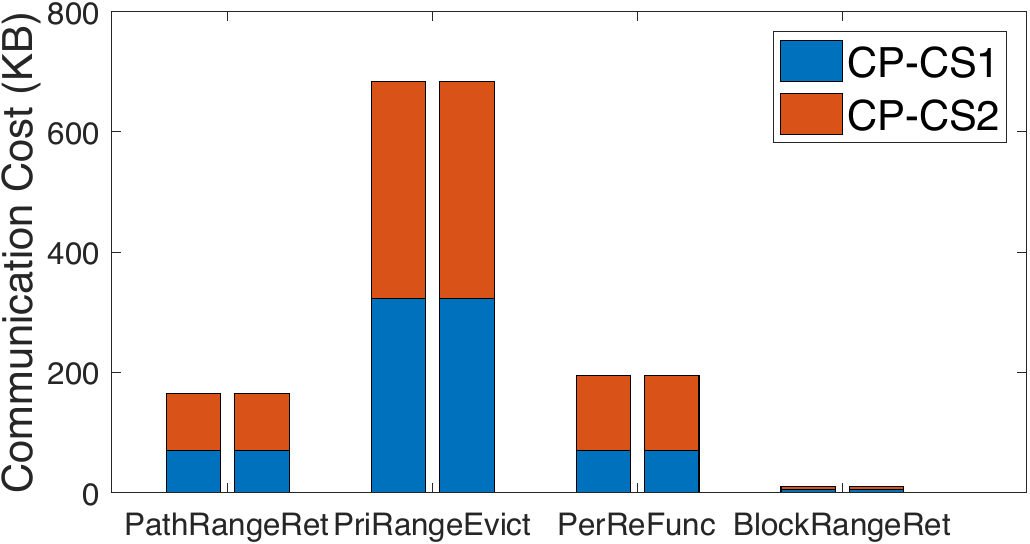}
	}
	\hfill
	\subfloat[Local region, $r=4$]{
		\label{fig:intra_comm}
		\includegraphics[width=0.23\linewidth]{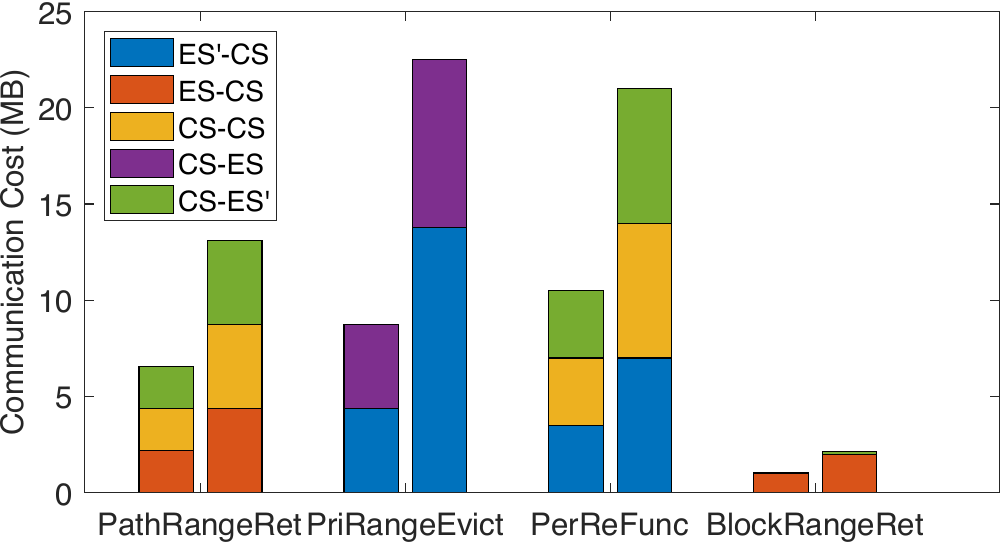}
	}
	\hfill
	\subfloat[Inter-continental, $r=5$]{
		\label{fig:inter_comm}
		\includegraphics[width=0.23\linewidth]{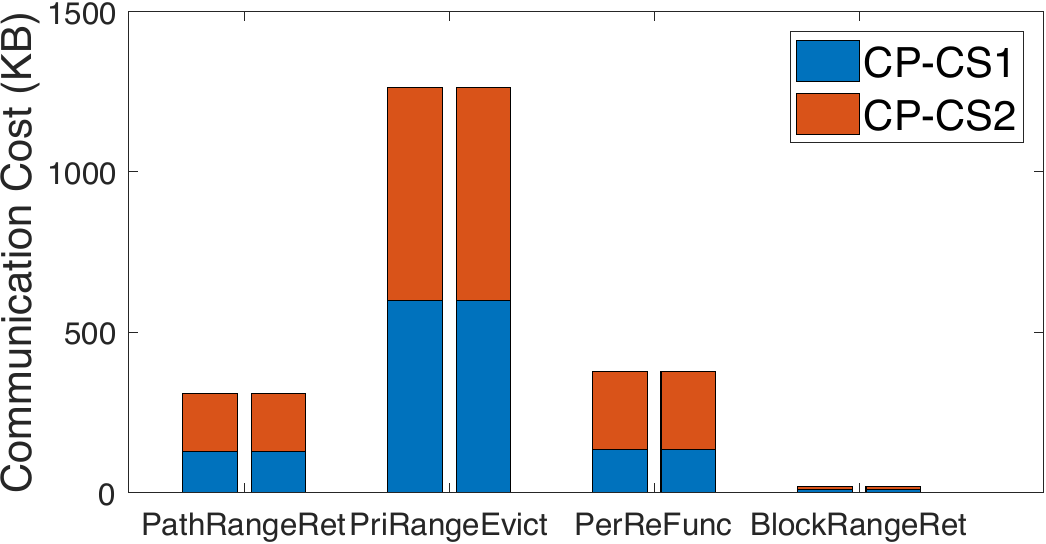}
	}
	\hfill
	\subfloat[Local region, $r=5$]{
		\label{fig:intra_comm}
		\includegraphics[width=0.23\linewidth]{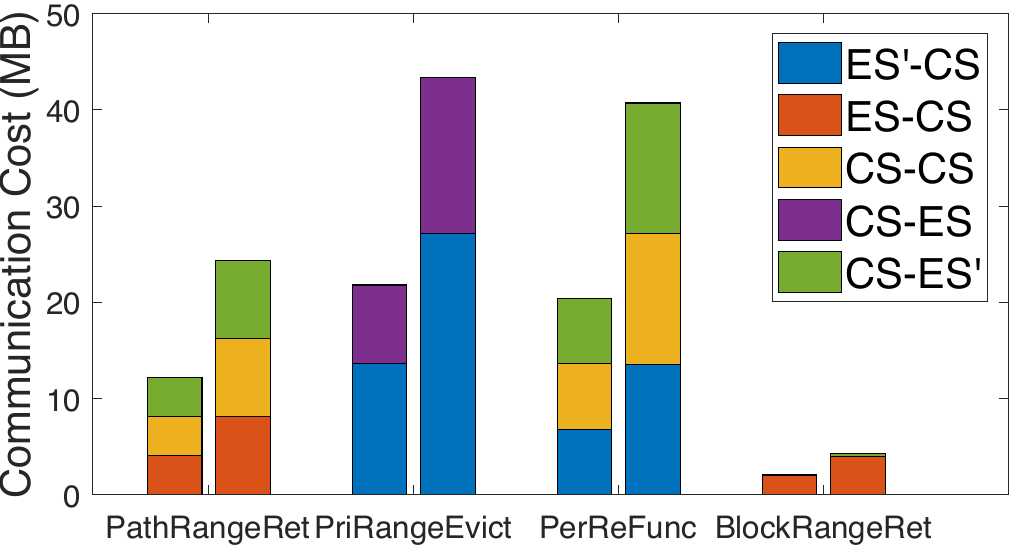}
	}
	
	\subfloat[Inter-continental, $r=6$]{
		\label{fig:inter_comm}
		\includegraphics[width=0.23\linewidth]{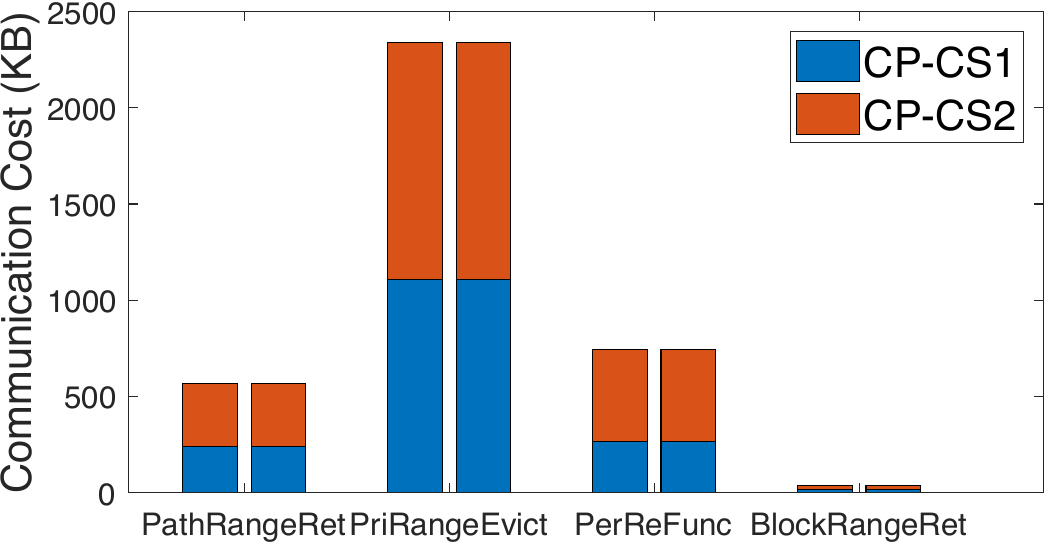}
	}
	\hfill
	\subfloat[Local region, $r=6$]{
		\label{fig:intra_comm}
		\includegraphics[width=0.23\linewidth]{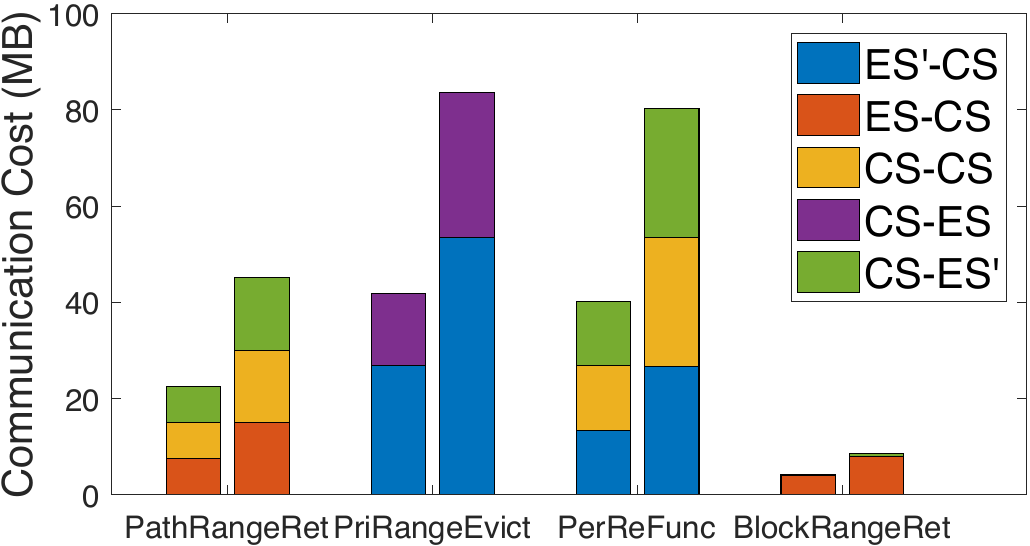}
	}
	\hfill
	\subfloat[Inter-continental, $r=7$]{
		\label{fig:inter_comm}
		\includegraphics[width=0.23\linewidth]{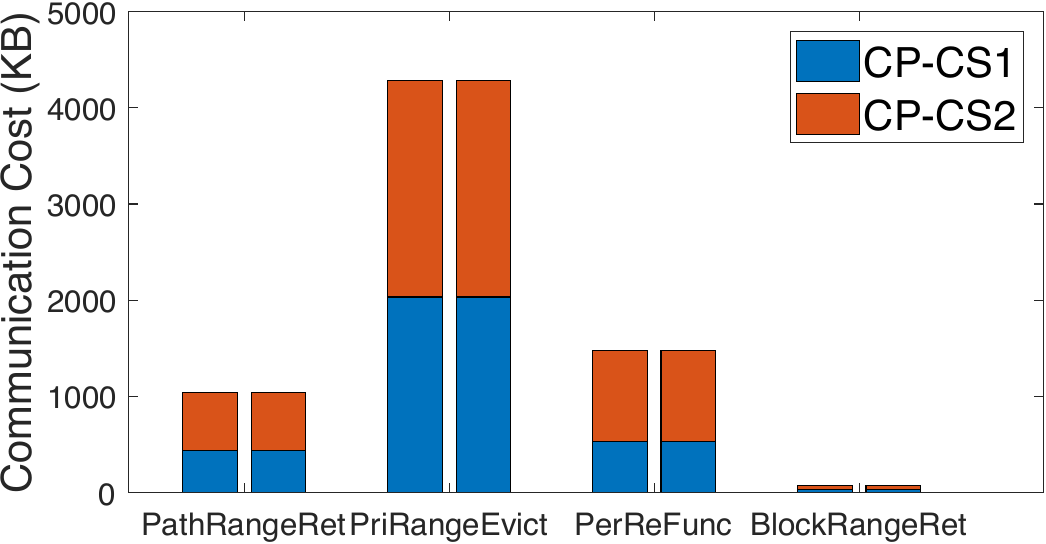}
	}
	\hfill
	\subfloat[Local region, $r=7$]{
		\label{fig:intra_comm}
		\includegraphics[width=0.23\linewidth]{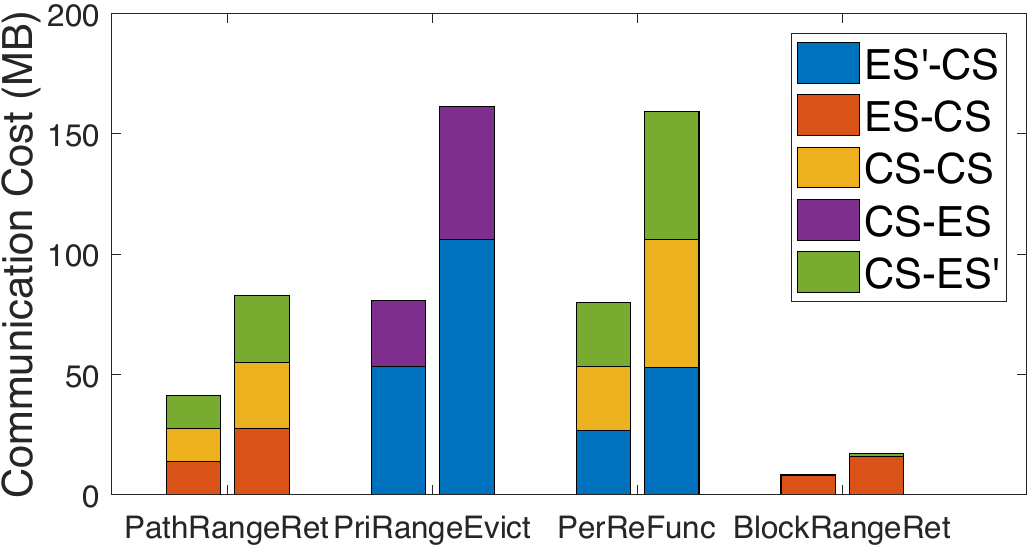}
	}
	\caption{The communication cost (MB) of oblivious operations with varying $r$. For each operation, the left bar is for $|B|=2$ KB, and the right one is for $|B|=4$ KB.}
	\label{fig:ovliv_op_comm_vary_r}
\end{figure*}

The highest computation cost of $CS_1$ and $CS_2$ in the \textsf{PriRangeEvict} is due to the DPF evaluation to write $2^r\times \log_2N\times Z$ blocks to $2^r$ ORAM paths, and each block write requires to scan $|S_r|=105 \times 2^r$ blocks.
Other building blocks only work on fewer blocks (i.e., $2^r \times \log_2N$ blocks for \textsf{BlockRangeRet}, and $|S_r|$ blocks for \textsf{PerReFunc}) or has less operations on blocks, i.e., \textsf{BlockRangeRet} can permute and re-encrypt a path of $\log_2 N \times Z$ blocks with one scan.

The communication cost breakdowns under different rORAM parameters ($r$) are given in Figure~\ref{fig:ovliv_op_comm_vary_r}. 

\end{document}